\documentclass[10pt,a4paper,reqno]{amsart}

\pdfoutput=1 

\usepackage[utf8]{inputenc}
\usepackage[T1]{fontenc}
\usepackage[english]{babel}
\usepackage{amsmath}
\usepackage{amsfonts}
\usepackage{amssymb}
\usepackage{amsthm}
\usepackage{ucs}
\usepackage{graphicx}
\usepackage{xcolor}
\usepackage{dsfont}
\usepackage{tikz}
\usepackage{wasysym}
\usepackage{physics}
\usepackage{mathtools}
\usepackage{xifthen}
\usepackage[inline]{enumitem}
\usepackage{stmaryrd}
\usepackage{constants}
\usepackage[ruled]{algorithm2e}
\usepackage{tikz,pgfplots}
\usepackage{chngcntr}
\usepackage{hyperref}

\counterwithin{figure}{section}

\graphicspath{{images/}}

\newconstantfamily{expoThm}{
	symbol=\nu
}
\newconstantfamily{csts}{
	symbol=C
}
\renewconstantfamily{normal}{
	symbol=c
}

\makeatletter
\newcommand{\pushright}[1]{\ifmeasuring@#1\else\omit\hfill$\displaystyle#1$\fi\ignorespaces}
\newcommand{\pushleft}[1]{\ifmeasuring@#1\else\omit$\displaystyle#1$\hfill\fi\ignorespaces}
\makeatother

\newcommand{\R}{\mathbb{R}}
\newcommand{\Zd}{\mathbb{Z}^d}
\newcommand{\Ztwo}{\mathbb{Z}^2}
\newcommand{\Ed}{\mathbb{E}^d}
\newcommand{\Etwo}{\mathbb{E}^2}
\newcommand{\Tc}{T_{\mathrm{\scriptscriptstyle c}}}
\newcommand{\betac}{\beta_{\mathrm{\scriptscriptstyle c}}}
\newcommand{\normI}[1]{\left\|#1\right\|_{\scriptscriptstyle 1}}
\newcommand{\normsup}[1]{\left\|#1\right\|_{\scriptscriptstyle\infty}}
\newcommand{\setof}[2]{\{#1\,:\,#2\}}
\newcommand{\bsetof}[2]{\bigl\{#1\,:\,#2\bigr\}}
\newcommand{\given}{\,|\,}
\newcommand{\bgiven}{\bigm\vert}
\newcommand{\PottsM}{\mu}
\newcommand{\PottsZ}{\mathcal{Z}}
\newcommand{\icl}{\xi_\beta}
\newcommand{\iclx}{\xi_{x'}}
\newcommand{\comp}{{\scriptscriptstyle\mathrm{c}}}
\newcommand{\eone}{\vec e_1}
\newcommand{\etwo}{\vec e_2}
\newcommand{\Jc}{J_{\mathrm{\scriptscriptstyle c}}}
\newcommand{\xpc}{x'_{\mathrm{\scriptscriptstyle c}}}
\newcommand{\xc}{x_{\mathrm{\scriptscriptstyle c}}}
\newcommand{\cone}{\mathcal{Y}}
\newcommand{\fcone}{\mathcal{Y}^\blacktriangleleft}
\newcommand{\bcone}{\mathcal{Y}^\blacktriangleright}
\newcommand{\diam}{D}

\newcommand{\nfe}{\theta_0}
\newcommand{\pfe}{\theta_1}
\newcommand{\e}{\mathbb{E}}
\newcommand{\ebf}{\mathbf{E}}
\newcommand{\p}{\mathbb{P}}
\newcommand{\pbf}{\mathbf{P}}
\newcommand{\Z}{\mathbb{Z}}
\newcommand{\tree}{\mathcal{T}}

\newcommand{\Ham}{\mathcal{H}}
\newcommand{\Line}{\mathcal{L}}

\newcommand{\bc}{\mathcal{B}}
\newcommand{\alp}{\mathcal{A}}
\newcommand{\cut}{\textnormal{Cuts}}
\newcommand{\calS}{\mathcal{S}}
\newcommand{\compl}[1]{\overline{#1}}
\newcommand{\subProb}{\Psi}
\newcommand{\bfx}{\mathbf{x}}
\newcommand{\bfy}{\mathbf{y}}

\newcommand{\bfB}{\mathbf{B}}
\newcommand{\bfX}{\mathbf{X}}

\newcommand{\clusterSet}{\mathfrak{C}}
\newcommand{\statMes}{\mathbf{p}}
\newcommand{\calI}{\mathcal{I}}
\newcommand{\calR}{\mathcal{R}}
\newcommand{\bbA}{\mathbb{A}}
\newcommand{\bbB}{\mathbb{B}}
\newcommand{\bbC}{\mathbb{C}}

\newcommand{\decouplLaw}{\Xi}

\newcommand{\displace}{V}
\newcommand{\trajApp}{\tilde{\displace}}
\newcommand{\trajAppFact}{\bar{\displace}}
\newcommand{\traj}[1]{#1^{\mathrm{\scriptscriptstyle traj}}}

\newcommand{\barK}{\overline{K}}
\newcommand{\innerBox}{\Delta}
\newcommand{\outerBox}{\overline{\Delta}}
\newcommand{\shade}{\mathfrak{sh}}
\newcommand{\etaDob}{\eta^{\rm\scriptscriptstyle Dob}}
\newcommand{\bawa}{b}
\newcommand{\fowa}{f}
\newcommand{\wired}{\mathrm{\scriptscriptstyle w}}
\newcommand{\free}{\mathrm{\scriptscriptstyle f}}
\newcommand{\IF}[1]{\mathds{1}_{\{#1\}}}
\renewcommand{\nleftrightarrow}{\mathrel{\ooalign{$\leftrightarrow$\cr\hidewidth$/$\hidewidth}}}

\renewcommand{\emptyset}{\varnothing}
\newcommand{\calB}{\mathcal{B}}

\theoremstyle{plain}
\newtheorem{theorem}{Theorem}[section]
\newtheorem{lemma}[theorem]{Lemma}
\newtheorem{proposition}[theorem]{Proposition}

\newtheorem{definition}{Definition}[section]
\newtheorem{remark}{Remark}[section]
\newtheorem{claim}{Claim}[section]

\theoremstyle{definition}
\newtheorem{obs}{Observation}

\author{S\'{e}bastien Ott}
\address{Section de Mathématiques, Université de Genève, CH-1211 Genève, Switzerland}
\email{sebastien.ott@unige.ch}
\author{Yvan Velenik}
\address{Section de Mathématiques, Université de Genève, CH-1211 Genève, Switzerland}
\email{yvan.velenik@unige.ch}

\title{Potts models with a defect line}

\begin{document}
	
\begin{abstract}
	We provide a detailed analysis of the correlation length in the direction parallel to a line of modified coupling constants in the ferromagnetic Potts model on \(\Zd\) at temperatures \(T>\Tc\). We also describe how a line of weakened bonds pins the interface of the Potts model on \(\Ztwo\) below its critical temperature. These results are obtained by extending the analysis in~\cite{Friedli+Ioffe+Velenik-2013} from Bernoulli percolation to FK-percolation of arbitrary parameter \(q\geq 1\).
\end{abstract}

\maketitle


\section{Introduction and results}
\label{sec:Intro}

In 1980--81, Abraham published two papers~\cite{Abraham-1980,Abraham-1981} on the effect of a row of modified coupling constants on the interface of the two-dimensional Ising model, discussing what would later be recognized as pinning and wetting transitions. Being based on exact computations, these results provided precise information but little understanding on the underlying mechanisms. The desire to obtain a better understanding immediately led to an an intense activity (see~\cite{vanLeeuwen+Hilhorst-1981,Burkhardt-1981,Chalker-1981,Chui+Weeks-1981,Kroll-1981,Vallade+Lajzerowicz-1981} for some examples published in 1981 and~\cite{Fisher-1984} for a well-known early review). In all these papers, the same problems were tackled in the much simpler setting of \emph{effective interface models}: basically, modeling the interface as the trajectory of some random walk in suitable potentials. This approach provided not only a better understanding, but also allowed to consider various generalizations: one-dimensional paths in higher dimension (modeling a polymer, for example), higher-dimensional interfaces, random potentials, etc.
Note that there is still interest in such issues in the physics community (see, for example, \cite{Delfino-2016} for a recent exact approach, based on more sophisticated field theoretical techniques).
The analysis of effective models has also generated a lot of interest among mathematical physicists and probabilists: see, for instance, \cite{Velenik-2006,Giacomin-2007} for reviews. In the meantime, new techniques to analyze nonperturbatively various lattice spin systems have been developed~\cite{Campanino+Ioffe+Velenik-2003,Campanino+Ioffe+Velenik-2008}, making it potentially possible to import back the results about effective interface models to the ``genuine'' spin systems that originally motivated their analysis. This is precisely the purpose of the present paper, in which we provide a detailed description of the longitudinal correlation length of the Potts model on \(\Zd\) above the critical temperature in the presence of a line of modified coupling constants, as well as an analysis of the pinning of a Potts interface by a line of defects in the two-dimensional model below its critical temperature. (More generally, our results apply to all random-cluster models with parameter \(q\geq 1\).)
The results we obtain are in full agreement with the predictions by effective models.

\subsection{Correlation length of the Potts model on \(\Zd\) above \(\Tc\)}

Thanks to the self-duality of the \(2d\) Ising model, the problems analyzed in~\cite{Abraham-1980,Abraham-1981} admit equivalent reformulations in terms of the inverse correlation length of a \(2d\) Ising model above its critical temperature, in the presence of a line along which the coupling constants are modified. Such an analysis, based on exact computations, was undertaken by McCoy and Perk~\cite{McCoy_Perk-1980}, independently of the previously mentioned works and at the same time. An advantage of this dual version is that it admits immediate generalizations to higher-dimensional lattices. In this section, we investigate this problem in the more general case of Potts models on \(\Zd\). The low temperature setting for the Potts model on \(\Ztwo\) will be discussed in Section~\ref{sec:Potts:LT}.

Given \(i=(i_1,\ldots,i_d)\in\Zd\), we write \(i^\parallel = i_1\) and 
\(i^\perp = (i_2,\ldots,i_d)\) and we set \(\Line = 
\setof{i\in\Zd}{i^\perp=0}\). Moreover, let 
\(\Ed=\setof{\{i,j\}\subset\Zd}{\normI{j-i}=1}\).

Let \(\Omega_q^d=\{1,\ldots,q\}^{\Zd}\) be the set of configurations of the 
\(q\)-state Potts model on \(\Zd\). Given \(\Lambda\Subset\Zd\)
(that is, \(\Lambda\subset\Zd\) and finite)
and \(J\geq 
0\), we associate to \(\omega\in\Omega_q^d\) the energy
\[
\Ham_{\Lambda;J}(\omega) = -\sum_{\{i,j\}\cap\Lambda\neq\emptyset} J_{i,j} \delta_{\omega_i,\omega_j} ,
\]
where the coupling constants \((J_{i,j})_{i,j\in\Zd}\) are given by
\[
J_{i,j} =
\begin{cases}
1	&	\text{if \(\{i,j\}\in\Ed\) with
	\(\{i,j\}\not\subset\Line\)
	,} \\
J	&	\text{if \(\{i,j\}\in\Ed\) with \(i,j\in\Line\),} \\
0	&	\text{otherwise.}
\end{cases}
\]
The Gibbs measure in \(\Lambda\Subset\Zd\), with boundary condition \(\eta\in\Omega_q^d\) and at inverse temperature \(\beta=1/T\), is the probability measure on \(\Omega_q^d\) given by
\[
\PottsM_{\Lambda;\beta,J}^\eta (\omega) =
\begin{cases}
\bigl( \PottsZ_{\Lambda;\beta,J}^\eta \bigr)^{-1}\, 
e^{-\beta\Ham_{\Lambda;J}(\omega)}		&	\text{if }\omega_i=\eta_i\;\forall 
i\not\in\Lambda, \\
0	&	\text{otherwise.}
\end{cases}
\]
Finally, the associated infinite-volume Gibbs measures are all probability measures \(\PottsM\) on \(\Omega_q^d\) satisfying
\[
\PottsM(\,\cdot \given \mathcal{F}_{\Lambda^\comp})(\eta) = \PottsM_{\Lambda;\beta,J}^\eta (\,\cdot\,) \qquad\forall\Lambda\Subset\Zd
\]
for \(\PottsM\)-almost every \(\eta\in\Omega_q^d\). Here, \(\mathcal{F}_{\Lambda^\comp}\) is the \(\sigma\)-algebra generated by the random variables \((\omega_i)_{i\not\in\Lambda}\).

\medskip
We first recall a few results concerning the \emph{homogeneous} model, in which 
\(J=1\). In this case, it is well-known that, for any \(d\geq 2\), there exists 
\(\betac=\betac(d)\in(0,\infty)\) such that there is a unique infinite-volume 
Gibbs measure when \(\beta<\betac=1/\Tc\), but infinitely many infinite-volume Gibbs 
measures when \(\beta>\betac\). Assume that \(\beta<\betac\) and denote by 
\(\PottsM_{\beta}\) the (unique) infinite-volume Gibbs measure. 
Then, the inverse correlation length is 
positive~\cite{Duminil-Copin+Raoufi+Tassion-2017}:
\[
\icl = \lim_{n\to\infty} -\frac1n \log \bigl( \PottsM_{\beta}(\omega_0=\omega_{n\eone}) - \frac1q \bigr) > 0 .
\]
More precisely, the following \emph{Ornstein--Zernike} asymptotics 
hold~\cite{Campanino+Ioffe+Velenik-2008}: there exists \(C_\beta=C_\beta(q,d)>0\) such that, 
as \(n\to\infty\),
\begin{equation}
\label{eq:OZPotts}
\PottsM_{\beta}(\omega_0=\omega_{n\eone}) = \frac1q + \frac{C_\beta}{n^{(d-1)/2}}\, e^{-\icl n}\, (1+o(1)).
\end{equation}

\medskip
Let us now consider general values of \(J\geq 0\). We still assume that 
\(\beta<\betac\) (with the \(\betac\) defined above). It turns out\footnote{This follows, for example, from our analysis below; see Remark~\ref{rem:uniqueness}.} that there is still a unique infinite-volume Gibbs measure in this case,
which we denote by \(\PottsM_{\beta,J}\). We define the longitudinal inverse correlation length 
as follows: for any \(x\in\Zd\),
\begin{equation}
\label{eq:InvCorLenPotts}
\icl(J) = \lim_{n\to\infty} -\frac1n \log \bigl( 
\PottsM_{\beta;J}(\omega_x=\omega_{x+n\eone}) - \frac1q \bigr).
\end{equation}
We first claim that
\begin{theorem}
\label{thm:BasicPropPotts}
For any \(\beta<\betac\), the following properties hold:
\begin{enumerate}[label=(\roman*)]
\item The limit in~\eqref{eq:InvCorLenPotts} exists and is independent of \(x\).
\item \(\icl(J) > 0\) for all \(J\geq 0\).
\item \(J \mapsto \icl(J)\) is Lipschitz-continuous and nonincreasing.
\item There exist \(c_+,c_->0\), depending on \(\beta\), \(q\) and \(d\), such that
\[
c_+ e^{-\beta J} \geq \icl(J) \geq c_- e^{-\beta J}
\]
for all \(J\) sufficiently large.
\item \(\icl(J) = \icl(1) = \icl\) for all \(J\leq 1\).
\end{enumerate}
\end{theorem}
It follows in particular that
\[
\Jc = \Jc(\beta,q,d) = \sup\setof{J\geq 0}{\icl(J) = \icl}
\]
is well defined, for any \(\beta<\betac\), and satisfies \(\infty > \Jc \geq 1\). (See Figure~\ref{fig:icl_Ising} for an illustration in the case of the two-dimensional Ising model.)
\begin{remark}
The word ``longitudinal'' above refers to the fact that we consider the correlation length in a direction parallel to the defect line. One could, in a similar fashion, define the \emph{transverse} correlation length, by replacing \(\eone\) by \(\etwo\) in the definition. However, it is not difficult to show that this quantity always coincides with the corresponding quantity in the homogeneous model.
\end{remark}
Our next result provides 
information on the value of \(\Jc\):
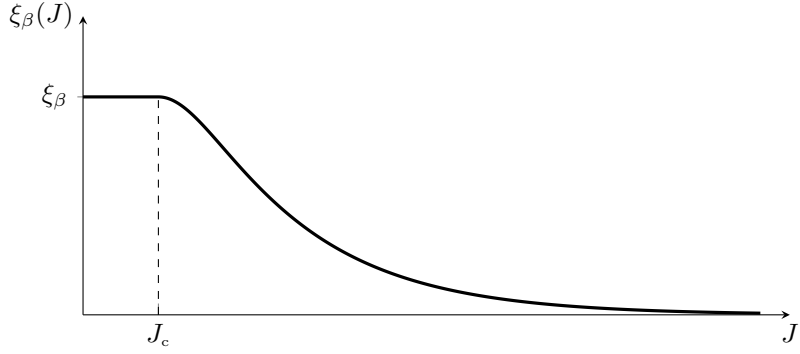
\begin{figure}
\begin{tikzpicture}
    \begin{axis}[
            clip,
            x=3cm,y=6cm,
            xmin=0, xmax=3.1,
            ymin=0, ymax=.66,
            axis lines = middle,
            xtick={0.3305},
            xtick align=inside,
            xticklabel={$\Jc$},
            ytick={.481586},
            yticklabel={$\icl$},
            xlabel={$J$},
            xlabel style={at=(current axis.right of origin), anchor=north},
            ylabel={$\icl(J)$},
            ylabel style={at=(current axis.above origin), anchor=east},
            no markers
        ] 
        \addplot+[domain=0:3,black,solid,very thick] file {images/xi_Ising.txt};
        \addplot+[mark=none,dashed,black] coordinates {(0.3305,0) (0.3305, 0.481586)};
    \end{axis}
\end{tikzpicture}
\caption{
The graph of \(\icl(J)\) for the two-dimensional Ising model at \(\beta=.75\betac\), as computed in~\cite{McCoy_Perk-1980}. As is proved for general Potts models in Theorem~\ref{thm:Jc_Potts}, \(\Jc=1\) in this case.
}
\label{fig:icl_Ising}
\end{figure}

\begin{theorem}\label{thm:Jc_Potts}
\(\Jc = 1\) when \(d=2\) or \(d=3\), but \(\Jc > 1\) when \(d\geq 4\).
\end{theorem}
When \(J>\Jc\), more precise information is available.
\begin{theorem}\label{thm:PinnedRegime}
The following properties hold, for any \(\beta<\betac\):
\begin{enumerate}[label=(\roman*)]
\item \(J \mapsto \icl(J)\) is real-analytic and strictly decreasing on \((\Jc,\infty)\).
\item\label{item:Critical2d} When \(d=2\), there exist \(c_+,c_-,\epsilon>0\), depending on 
\(\beta\), \(q\) and \(d\), such that, for all \(J\in(\Jc,\Jc+\epsilon)\),
\[
c_+ (J-\Jc)^2
\geq
\icl(\Jc)-\icl(J)
\geq
c_- (J-\Jc)^2 .
\]
\item\label{item:Critical3d} When \(d=3\), there exist \(c_+,c_-,\epsilon>0\), depending on 
\(\beta\), \(q\) and \(d\), such that, for all \(J\in(\Jc,\Jc+\epsilon)\),
\[
e^{-c_+/(J-\Jc)}
\geq
\icl(\Jc)-\icl(J)
\geq
e^{-c_-/(J-\Jc)} .
\]
\item\label{Item:quatre} For all \(J>\Jc\), there exists \(C_{\beta,J}=C_{\beta,J}(q,d)>0\) such that, as 
\(n\to\infty\),
\[
\PottsM_{\beta,J}(\omega_0=\omega_{n\eone}) = \frac1q + C_{\beta,J}\, e^{-\icl(J) n}\, (1+o(1)).
\]
\end{enumerate}
\end{theorem}
The behavior in the last statement should be contrasted with~\eqref{eq:OZPotts}.

\subsection{Pinning of the interface of the \(2d\) Potts model below \(\Tc\)}
\label{sec:Potts:LT}
We now restrict our attention to the lattice \(\Ztwo\). Let
\[
\Line^* =
\bsetof{\{(x,0),(x,1)\}\in\Etwo}{x\in\Z}.
\]
We now consider the \(q\)-state Potts model on \(\Ztwo\) with coupling 
constants \((J_{i,j})_{i,j\in\Ztwo}\) given by
\[
J_{i,j} =
\begin{cases}
1	&	\text{if\(\{i,j\}\in\Etwo\setminus\Line^*\),} \\
J	&	\text{if \(\{i,j\}\in\Line^*\),} \\
0	&	\text{otherwise.}
\end{cases}
\]
Let \(\Lambda_n = \{-n,\ldots,n\}\times\{-n+1,\ldots,n\}\) and let 
\(\etaDob\in\Omega_q^2\) be the Dobrushin-type boundary condition defined by
\[
\etaDob_i =
\begin{cases}
1	&	\text{if }i^\perp \geq 1,\\
2	&	\text{if }i^\perp \leq 0.
\end{cases}
\]
We denote by \(\PottsM_{n;\beta,J}^\pm\) the Gibbs measure in \(\Lambda_n\) with boundary condition 
\(\etaDob\) at inverse temperature \(\beta\).

In the remainder of this section, we assume that \(\beta>\betac\). In 
that 
case, there is long-range order and it is convenient to describe configurations 
in terms of their Peierls contours. First, given \(\{i,j\}\in\Etwo\), denote by 
\(\{i,j\}^* = \setof{x\in\mathbb{R}^2}{\normsup{x-i}=\normsup{x-j}=\tfrac12}\) the 
dual edge separating \(i\) and \(j\). The contours of a configuration \(\omega\in\Omega_q^2\)
are the maximal connected components of
\[
\bigcup_{\{i,j\}\in\Etwo:\,\omega_i\neq\omega_j} \{i,j\}^*.
\]
When \(\omega_i=\etaDob_i\) for all \(i\not\in\Lambda_n\), there is a unique 
unbounded contour. We call its intersection with 
\([-n-\tfrac12,n+\tfrac12]\times\mathbb{R}\) the \emph{interface} 
and denote it by \(\Gamma_n\).
Note that \(\Gamma_n\) is a two-dimensional object, but with a macroscopic extension only along the first coordinate axis and an (essentially) bounded width, as we explain now.

\medskip
Consider first the homogeneous case \(J=1\). It can then be 
shown~\cite{Campanino+Ioffe+Velenik-2008} that, under 
\(\PottsM_{n;\beta,1}^\pm\), the interface has a width of order \(\log n\). 
Namely, for each \(x\in[-n-\tfrac12,n+\tfrac12]\), define
\[
\Gamma_n^+(x) = \max\setof{y\in\mathbb{R}}{(x,y)\in\Gamma_n},\quad
\Gamma_n^-(x) = \min\setof{y\in\mathbb{R}}{(x,y)\in\Gamma_n}.
\]
\begin{figure}[t]
\resizebox{5cm}{!}{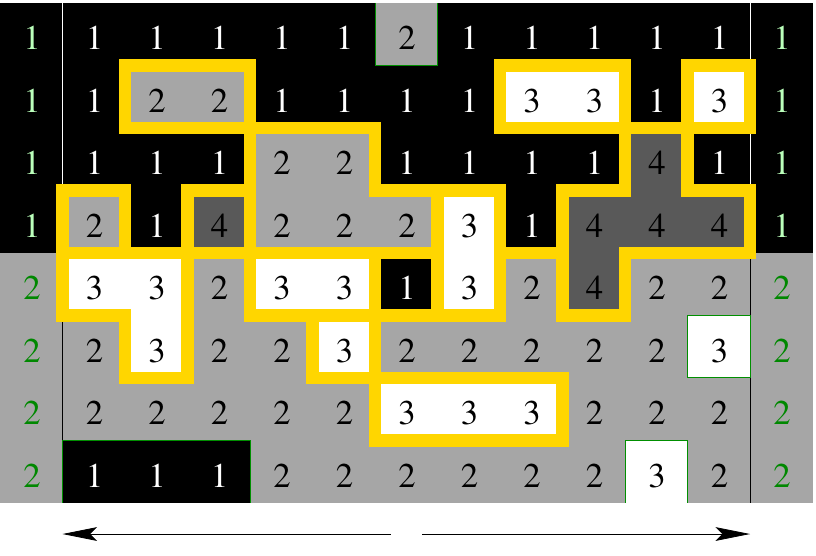}
\hspace{5mm}
\resizebox{5cm}{!}{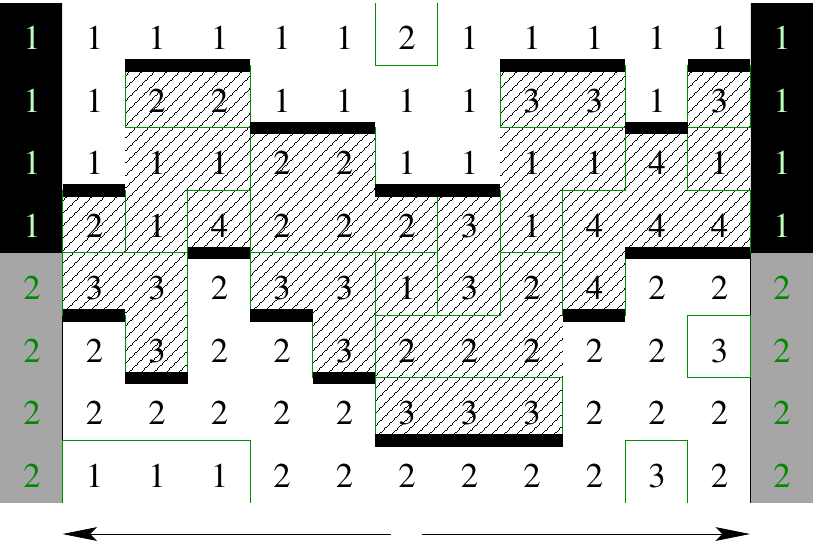}
\caption{Left: The interface of the \(2d\) Potts model. Right: the functions \(\Gamma^+_n\) and \(\Gamma^-_n\); note that the interface is contained in the regions they delimit (hatched area).}
\end{figure}
Then, there exists \(C(\beta,q)\) such that
\[
\lim_{n\to\infty} \PottsM_{n;\beta,1}^\pm \bigl( 
\max_x(\Gamma_n^+(x)-\Gamma_n^-(x)) \leq C(\beta)\log n \bigr) = 1.
\]
Moreover, under diffusive scaling, the interface weakly converges to a Brownian 
bridge~\cite{Campanino+Ioffe+Velenik-2008}: for any \(\beta>\betac\), there exists \(\kappa_\beta>0\) 
such that, as \(n\to\infty\),
\[
\frac1{\sqrt{n}} \Gamma_n^+ (n\,\cdot) \Rightarrow \kappa_\beta B_\cdot,
\]
where \((B_t)_{-1\leq t\leq 1}\) denotes the standard Brownian bridge on 
\([-1,1]\).

\medskip
The main result of this section is that, whenever \(J<1\), the interface ceases 
to behave diffusively and instead localizes along the defect line:
\begin{theorem}\label{thm:PottsScalingInterface}
For any \(\beta>\betac\) and any \(0\leq J<1\), there exists \(C_{\beta,J}=C_{\beta,J}(q)\) such that
\[
\lim_{n\to\infty} \PottsM_{n;\beta,1}^\pm \bigl( 
\max_x \Gamma_n^+(x) \leq C_{\beta,J}\log n,
\min_x \Gamma_n^-(x) \geq -C_{\beta,J}\log n \bigr) = 1.
\]
\end{theorem}
\begin{figure}[t]
\centering
\includegraphics[width=5cm]{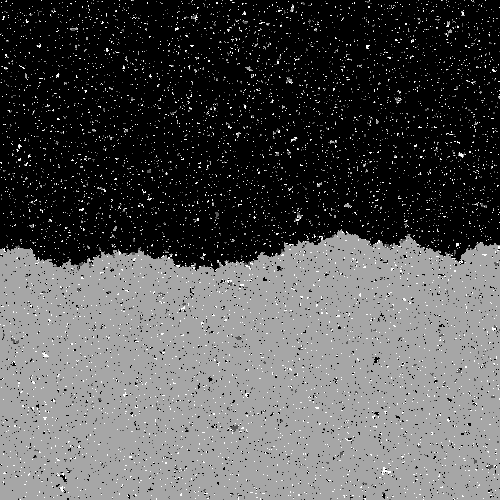}
\hspace{5mm}
\includegraphics[width=5cm]{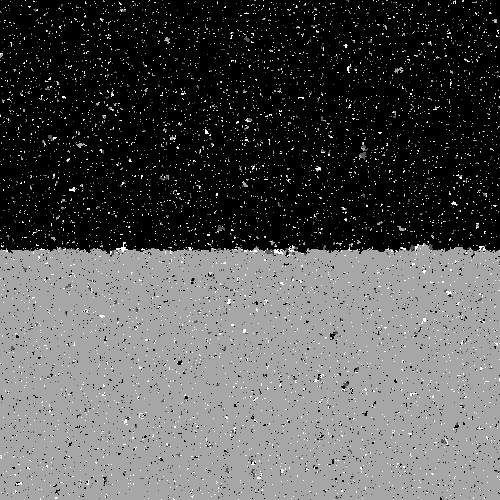}
\caption{Interface in the \(2d\) 4-state Potts model. Left: \(J=1\). Right: \(J=1/2\). The same inverse temperature \(\beta>\betac\) is used in both cases.}
\end{figure}
Note that, under diffusive scaling, the limit is then identically \(0\): an 
arbitrary weakening of the coupling constants along \(\Line^*\) pins the 
interface.
Actually, the claim in the theorem will follow from a detailed description of the structure of the interface (see Theorem~\ref{thm:CPdensity}), which provides a much stronger claim than what is stated above. In particular, the width of the interface is typically bounded, with only rare deformations of order \(\log n\).
(In fact, Theorem~\ref{thm:PottsScalingInterface} will essentially be a corollary of Item~\ref{Item:quatre} of Theorem~\ref{thm:PinnedRegime}.)

\bigskip
Before closing this introduction, let us briefly mention that although we restricted our attention to a defect along a line of the lattice, this is by no means necessary. Straightforward adaptation of our arguments would allow the analysis, for example, of a defect along the lattice approximation of any line with rational ``slope'', or other periodic structures. Similarly, the restriction to nearest-neighbor interactions is only necessary for the statement of Theorem~\ref{thm:PottsScalingInterface} (the proof of which relies on duality); for the other claims, any finite-range, translation-invariant, reflection-symmetric interaction would do.

\subsection{Open problems}

In view of the results presented above, there remain a few interesting open 
problems:
\begin{itemize}
\item Determine the behavior of \(\icl(J)\) in the neighborhood of \(\Jc\) in 
dimensions \(d\geq 4\). By analogy with the results for effective models (see~\cite[Theorem~2.1]{Giacomin-2007}), we conjecture that the qualitative behavior of \(\icl(\Jc)-\icl(J)\) as \(J\downarrow\Jc\) is as follows: \(\Theta((J-\Jc)^2)\) when \(d=4\), \(\Theta((J-\Jc)/|\log(J-\Jc)|)\) when \(d=5\) and \(\Theta(J-\Jc)\) when \(d\geq 6\).
\item Determine the sharp asymptotics of the \(2\)-point function when \(J\leq\Jc\). Only the case \(J=1\) has been treated in complete generality up to now. For the two-dimensional Ising model, the asymptotic behavior was explicitly computed in~\cite{McCoy_Perk-1980} and found to be of the form
\[
\PottsM_{\beta,J}(\omega_0=\omega_{n\eone}) = \frac12 + \frac{C_\beta}{n^{3/2}}\, e^{-\icl n}\, (1+o(1))
\]
when \(J<\Jc=1\). Note the exponent of the prefactor, which is not of the usual Ornstein--Zernike form.
Again, by analogy with what happens in effective models (see~\cite[Theorem~2.2]{Giacomin-2007}), we expect the prefactor to be of order \(n^{-(d-1)/2}\) when \(d\geq 4\) and \(n^{-1}(\log n)^{-2}\) when \(d=3\).
\item Closely related to the previous problem, determine the scaling limit of the interface in the two-dimensional model 
when \(J<\Jc\). We expect the latter to be given by a Brownian excursion after diffusive scaling, as a consequence of entropic repulsion away from \(\Line\). This is fully compatible with the exponent in the prefactor mentioned in the previous point.
\end{itemize}
Moreover, there are a number of natural generalizations, to which we plan to return in future works:
\begin{itemize}
\item What happens when the defect is located along the boundary of the system? In dimension \(2\), this amounts to studying the wetting problem for the Potts model.
\item What happens when the defect is of dimension \(d'\in (1,d)\)? Note that, 
in this case, the system may display long-range order along the defect even when 
the bulk is disordered. In particular, the longitudinal inverse correlation length vanishes for finite values of \(J\).
\item Is it possible to adapt some of the technology used to deal with pinning 
of a random walk by a disordered potential to cover the case of random 
(quenched, ferromagnetic) coupling constants along the defect?
\end{itemize}


\section{Random cluster representation, notations and strategy of the proof}

In this section, we introduce a few notations which will be recurrent throughout this article, we recall briefly the random-cluster (or Fortuin Kastelyn) representation of the Potts model and we give a short outline of the proofs of the theorems of Section~\ref{sec:Intro}.

\subsection{Random-cluster representation of the Potts model}

The Potts model on a finite graph $G=(V_G,E_G)$ can be mapped to a percolation model defined on $\{0,1\}^{E_G}$
(identifying the value \(1\) with the presence of an edge and the value \(0\) with its absence)
in the following way. 
For any edge configuration $\omega\subset E_G$, we denote by $\kappa(\omega)$ the number of connected components in $(V_G,\omega)$.
Writing \(\vec{x} = (x_e)_{e\in E_G}\), with
$x_e = e^{\beta J_e}-1$ for each $e\in E_G$, we associate to $\omega\subset E_G$ the probability
\[
\p_{\vec{x},q}(\omega) = (Z_{\vec{x},q})^{-1}\, q^{\kappa(\omega)}\prod_{e\in\omega} x_e ,
\]
where $Z_{\vec{x},q} = \sum_{\omega\subset E_G} q^{\kappa(\omega)} \prod_{e\in\omega}x_e$. The corresponding expectation will be denoted by $\e_{\vec{x},q}$.
We say that an edge $e$ with $\omega_e=1$ is open and denote by $|\omega|$ or $o(\omega)$ the number of open edges.
We say that $u,v\in G$ are connected, which we write $u\leftrightarrow v,$ if they lie in the same connected component.
For $A\subset E_G$, denote by $\omega_A$ the configuration $\omega$ restricted to $A$ and, for $e\in E_G$, by $\omega_{\setminus e}$ the configuration $\omega_{E_G\setminus\{e\}}$.

The random-cluster measures with $q\geq 1$ enjoy the following properties.
\begin{itemize}
	\item[]\textbf{Finite energy:} For any \(e\in E_G\) and any configuration \(\omega_{\setminus e}\),
	\[
	\frac{x_e}{x_e+q}\leq\p_{\vec{x},q}(\omega_e=1|\omega_{\setminus e})\leq\frac{x_e}{x_e+1}.
	\]
	\item[]\textbf{Positive association:} Let $f,g$ be two nondecreasing functions (w.r.t.\ the partial order induced by $0\leq 1$ on $\{0,1\}^{E_G}$). Then the FKG inequality holds:
	\[
	\e_{\vec{x},q}\left[fg\right]\geq\e_{\vec{x},q}\left[f\right]\e_{\vec{x},q}\left[g\right].
	\]
	\item[]\textbf{Stochastic monotonicity:} Assume that $x_e\leq y_e$ for all $e\in E_G$. Then $\p_{\vec{x},q}\preceq \p_{\vec{y},q}$.
\end{itemize}
The random-cluster model does not enjoy the usual spatial Markov property but an analogue can be used: for $\Lambda\subset G$, the random-cluster measure in $\Lambda$ with boundary condition $\omega_{G\setminus \Lambda}$ depends only on the connectivity properties of the vertices in the inner boundary of $\Lambda$, thus a boundary condition is a partition of those vertices (every set of the partition is a connected component).
In particular, the measure with \emph{wired boundary condition} (denoted $\p_{\vec{x},q,\Lambda}^\wired$) is obtained by setting \(\omega_{G\setminus \Lambda}\equiv 1\), while the measure with \emph{free boundary condition} ($\p_{\vec{x},q,\Lambda}^\free$) is obtained using \(\omega_{G\setminus \Lambda}\equiv 0\).
Stochastic monotonicity then implies that these two measures are extremal with respect to stochastic ordering.

In the sequel, we will work with the random-cluster measure on $\Zd$ induced by the weights
\(
x_e =
\begin{cases}
	x=e^{\beta}-1 		& \text{if } e\in\Line^\comp\\
	x'=e^{\beta J}-1 	& \text{if } e\in\Line
\end{cases}
\).
We denote the corresponding law $\p_{x'}$; it corresponds to the random-cluster measure associated with the Potts measure described in the previous section. In particular, the $2$-point correlation function of the Potts model can be rewritten as (see, for example, \cite[(1.16)]{Grimmett-2006})
\begin{align}
\label{eq:PottsFK}
	\PottsM_{\beta;J}(\sigma_u=\sigma_v) = \frac{1}{q} + \frac{q-1}{q} \p_{x'}(u\leftrightarrow v) .
\end{align}
From this, it immediately follows that the inverse correlation length $\icl(J)$ is equal to
\begin{equation}\label{eq:icl_FK}
\iclx =  \lim_{n\to\infty} -\frac1n \log \p_{x'}(0\leftrightarrow n\eone) .
\end{equation}

We will write $\p\equiv\p_x$ and $\e\equiv\e_x$ for the law and expectation of the homogeneous model; the corresponding measure in a finite volume \(\Lambda\Subset\Zd\) with boundary condition \(\#\in\{\mathrm{f,w}\}\) will be denoted \(\p_{\Lambda}^{\#}\).
Everywhere in the analysis below, except in the proof of Theorem~\ref{thm:PottsScalingInterface}, we will implicitly assume that \(x<\xc=e^{\betac}-1\) and \(q\geq 1\) are fixed and we will thus omit them from the notation.

We also write $\xi\equiv\xi_x$ for the corresponding inverse correlation length. The following exponential decay of connectivities under \(\p\), established in~\cite{Duminil-Copin+Raoufi+Tassion-2017}, plays a crucial role in our analysis.
\begin{lemma}
	\label{lem:expoDecFKwired}
	Let $x<\xc$. Then there exists $\Cl[expoThm]{decayFK}>0$ such that, for $n$ large enough,
	\[
		\p_{\Lambda_n}^{\wired}\left(0\leftrightarrow\Lambda_n^\comp\right)\leq e^{-\Cr{decayFK}n}.
	\]
\end{lemma}

We will prove all the results of Section~\ref{sec:Intro} in the random-cluster representation. They can then be translated straightforwardly to the Potts model language via~\eqref{eq:PottsFK}.

\begin{remark}
	\label{rem:finVol}
	Since \(x<\xc\), we can always work in large but finite boxes. Indeed, for any event $A$ depending on a finite number of edges, we can find a finite box $\Lambda\subset\Z^d$ such that
	\[
	\frac{1}{2}\p_{\Lambda}^{\free}(A)\leq\p(A)\leq 2\p_{\Lambda}^{\free}(A).
	\]
	This will be done in several instances for technical reasons, but we will keep the same notation as for the infinite-volume measure for readability purposes. The choice of boundary condition does not matter, thanks to the uniqueness of the infinite volume measure in the sub-critical regime.
\end{remark}

\subsection{Notations}
For $u,v\in\Z^d$, we denote $d(u,v) = \normI{v-u}$ the graph distance between $u$ and $v$; for $A,B\subset V(\Z^d)$, we set $d(A,B)=\min_{u\in A,v\in B}d(u,v)$.

For $a<b\in\mathbb{R}$, the notation $\Line_{[a,b]}$ denotes the subgraph of $\Z^d$: $\llbracket a,b \rrbracket\times\{0_{d-1}\}$ where $0_{d}$ is the origin of $\Z^d$ and $\llbracket a,b \rrbracket=[a,b]\cap\Z$; we also use the notations \(\Line_{<n}\equiv \Line_{(-\infty,n-1]}\) and \(\Line_{>n}\equiv \Line_{[n+1,\infty)}\).

Let $A\subset\Z^d$. We denote $u\stackrel{A}{\leftrightarrow} v$ the event in which $u\leftrightarrow v$ using only edges originally present in $A$. We will use the following notion of boundaries: \(\partial A = \setof{i\in A}{\exists j\in A^\comp, j\sim i}\) and \(\partial^{\rm ext} A = \setof{i\in A^\comp}{\exists j\in A, j\sim i}\). We will also use the notation \(\partial A\) to denote the set of edges having exactly one endpoint in $A$.

Sums of the form $\sum_{i=a}^{b}$ for $a,b$ not integers are to be understood as the corresponding sums with \(a,b\) replaced by the appropriate integers; for example, if this notation is used in the course of proving an upper bound, and the summand is nonnegative, then \(\sum_{i=a}^b = \sum_{i=\lfloor a\rfloor}^{\lceil b\rceil}\) (taking integer part would not change our estimates, so we chose not to write them explicitly for readability purposes).

In the following proofs, we will say that a quantity $f_r(K)$ is $o_K(1)$ if the following is true: for every $m>0$, one can find $r>0$ and $K_0>0$ such that, for every $K>K_0$, $f_r(K)\leq K^{-m}$ (the quantities $r,K$ will make sense later and the notation will become clear from the context; we define this here for easy reference, since this appears in several places in the following sections).

We will also use the notation \(\pfe = \min_{e\in\Zd}\min_{\omega_{\setminus e}} \p_{x'}(\omega_e=1|\omega_{\setminus e}) = \frac{x}{x+q} \wedge \frac{x'}{x'+q}\) and 
\(\nfe = \min_{e\in\Zd}\min_{\omega_{\setminus e}} \p_{x'}(\omega_e=0|\omega_{\setminus e}) = \frac{1}{1+(x\vee x')}\).

Finally, all constants appearing in the proofs below depend a priori on \(q,x\) and \(d\), but this will not be mentioned explicitly every time.

\newcommand{\ouv}{\mathsf{o}}
For a set \(E\subset E_G\) and a random-cluster configuration \(\omega\), we write \(\ouv_E(\omega)\) for the number of open edges of \(\omega\) in \(E\).

Given \(x\in\Zd\) and a random-cluster configuration \(\omega\), we denote by \(C_x=C_x(\omega)\) the cluster of \(x\) in \(\omega\).

\subsection{Outline of the paper}
In the next section, we provide the proof of Theorem~\ref{thm:BasicPropPotts}. In the process, we introduce some tools and calculations that will reappear in the proofs of Theorems~\ref{thm:Jc_Potts}, \ref{thm:PinnedRegime} and~\ref{thm:PottsScalingInterface}.

The procedure leading to the main claims is as follows: in Section~\ref{sec:RWrep}, we reinterpret long connections in the homogeneous model in terms of a random walk with i.i.d.\ increments. This is done combining the coarse-graining procedure of~\cite{Campanino+Ioffe+Velenik-2008} with a variant of the construction of~\cite{Comets+Fernandez+Ferrari-2002} (see the comments at the beginning of Appendix~\ref{sec:RenLongRange}), which is described in a self-contained way in Appendix~\ref{sec:RenLongRange}. The statement of Theorem~\ref{thm:Jc_Potts} and the second and third points of Theorem~\ref{thm:PinnedRegime} follow, on the one hand, by studying a pinning problem for the random walk obtained in Section~\ref{sec:RWrep} (see Section~\ref{sec:UB}) and, on the other hand, by an energy/entropy argument induced by the Russo-like formula described in Appendix~\ref{app:Russo} (see Section~\ref{sec:LB}). Finally, the first and fourth points of Theorem~\ref{thm:PinnedRegime} are established in Section~\ref{sec:PuExpDecAnalytStrDecr} by studying the localization of the random walk trajectory in a small neighborhood of $\Line$ via a coarse-graining argument. The claim of Theorem~\ref{thm:PottsScalingInterface} follows from the same analysis combined with self-duality, as explained in Section~\ref{sec:LocalizationLT}.


\section{Basic properties and estimates}
\label{sec:BasicProp}

In this section, we prove Theorem~\ref{thm:BasicPropPotts}.
We assume throughout that \(x<\xc\).
Using the correspondence described in the previous section, it is sufficient to establish the following lemma. 

\begin{lemma}[Basic properties of $\iclx$]
\label{lem:BasicProp}
	The limit in~\eqref{eq:icl_FK} exists and defines a function $\iclx$ with the following properties.
	\begin{enumerate}[label=\textit{\alph*)}]
		\item \label{it:unifZd} $\forall u\in\Zd$,
		\[
		\lim_{n\to\infty} -\frac{1}{n} \log \p_{x'}(u\leftrightarrow u + n\eone) = \iclx.
		\]
		Moreover, $\forall u,v\in\Zd$,
		\begin{equation}\label{eq:decay_xprime}
		\p_{x'}(u\leftrightarrow v) \leq e^{-\iclx|v_1-u_1|}.
		\end{equation}
		\item \label{it:eq} $\iclx = \xi$ for all $x'\leq x$ and $\iclx < \xi$ for $x'$ sufficiently large.
		\item \label{it:monot} $x'\mapsto\iclx$ is locally Lipschitz continuous, nonincreasing on $[0,\infty]$ and strictly positive for $x'\in[0,\infty)$.
		\item \label{it:asymBehav} There exist $c_+,c_- > 0$ such that $\frac{c_-}{x'}\leq\iclx\leq\frac{c_+}{x'}$ for $x'$ large enough.
	\end{enumerate}
	In particular, there exists $x'_c\in [x,\infty)$ such that $\iclx = \xi$ for all $x'\leq x'_c$ and $\iclx < \xi$ for all $x'>x'_c$.
\end{lemma}
\begin{remark}
	We actually prove something stronger than strict positivity of \(\iclx\): we show that there exists $c>0$ such that, for all \(n\),
	\begin{align}
		\label{eq:expDecUnif}
		\p_{x',\Lambda_{n}}^{\wired}(0\leftrightarrow \partial\Lambda_n)\leq e^{-cn}.
	\end{align}
	By stochastic monotonicity, this implies the same bound for any boundary condition.
\end{remark}
\begin{proof}
	\begin{itemize}[leftmargin=*]
		\item The existence and the first part of \ref{it:unifZd} are shown using Fekete's lemma. We first prove existence of \(\iclx^u = \lim_{n\to\infty} -\frac{1}{n} \log \p_{x'}(u\leftrightarrow u + n\eone)\).
		Define $\pi_{n}=\log\p_{x'}(u\leftrightarrow u + n \eone)$. We see that $(-\pi_{n})_n$ is a subadditive sequence: by FKG and translation invariance in the $\eone$-direction,
		\begin{align*}
			\pi_{n+m}&=\log\p_{x'}(u\leftrightarrow u + (n+m)\eone)\\
			&\geq \log\p_{x'}(u\leftrightarrow u + n\eone\leftrightarrow u + (n+m)\eone)\\
			&\geq \log\bigl(\p_{x'}(u\leftrightarrow u + n\eone)\p_{x'}(u + n \eone\leftrightarrow u + (n+m)\eone)\bigr)\\
			&=\log\p_{x'}(u\leftrightarrow u + n\eone) + \log\p_{x'}(u\leftrightarrow u + m\eone) = \pi_{n}+\pi_{m}.
		\end{align*}
		Fekete's lemma then implies that $ \iclx^u=\lim_{n\to\infty}\frac{-\pi_{n}}{n}=\inf_{n}\frac{-\pi_{n}}{n}$ exists; in particular,
		\begin{equation}\label{eq:UBiclxu}
		\p_{x'}(u\leftrightarrow u + n\eone)\leq e^{-\iclx^u n} .
		\end{equation}
		To prove that \(\iclx^u = \iclx^0 \equiv \iclx\), just observe that, for all $u\in\Zd$,
		\begin{align*}
			\p_{x'}\left(u\leftrightarrow u+n\eone\right)&\geq\p_{x'}\left(u\leftrightarrow (u_1,0_{d-1})\leftrightarrow (u_1+n,0_{d-1})\leftrightarrow u + n\eone\right)\\
			&\geq \pfe^{2d(u,\Line)}\p_{x'}\left(0\leftrightarrow n\eone\right),
		\end{align*}
		and therefore \(\iclx^u \leq \iclx\). The same argument, exchanging the role of \(0\) and \(u\), yields the reverse inequality.
		
		\item The second part of \ref{it:unifZd} follows from
		\begin{equation*}
			\p_{x'}(u\leftrightarrow v)^2
			\leq
			\p_{x'}(u\leftrightarrow v\leftrightarrow \bar{u}_v)
			\leq
			\p_{x'}(u\leftrightarrow \bar{u}_v)
			\leq
			e^{-\iclx2|v_1-u_1|},
		\end{equation*}
		where $\bar{u}_v$ denotes the point obtained from $u$ by a reflection through the hyperplane orthogonal to $\Line$ containing $v$.
		The last inequality is a direct consequence of the bound~\eqref{eq:UBiclxu} and the identity \(\iclx^u=\iclx\).
				
		\item The monotonicity of $x'\mapsto\iclx$ follows from the stochastic domination $\p_{x'_1}\succcurlyeq \p_{x'_2}$ when $x'_1\geq x'_2$.
		
		\item To get the first point of item~\ref{it:eq}, we fix \(x'\leq x\) and work in a finite volume (see Remark~\ref{rem:finVol}). We will use a coupling $\Phi(\omega,\eta)$ between $\p$ and $\p_{x'}$ satisfying (we denote $C_{\Line}(\omega)$ the connected component of the line $\Line$ in $\omega$):
		\begin{enumerate}[label=(\roman*)]
			\item $\omega\sim \p$ and $\eta\sim \p_{x'}$;
			\item $\omega\geq \eta$;
			\item $\omega = \eta$ outside of $C_{\Line}(\omega)$.
		\end{enumerate}
		A sketch of the construction of such a coupling  (as well as references) is provided in Appendix~\ref{app:Couplings}.
		Choosing $1>\alpha>\beta>1/2$ and setting $[j]=n^{\alpha}\etwo + j\eone$, we have
		\begin{align*}
			\p_{x'}(0\leftrightarrow n\eone)
			&\geq
			\pfe^{2n^{\alpha}} \p_{x'}([0] \leftrightarrow [n])\\
			&\geq
			\pfe^{2n^{\alpha}}\prod_{i=1}^{n^{1-\beta}}\p_{x'}( [(i-1)n^{\beta}] \leftrightarrow [in^{\beta}])\\
			&=
			\pfe^{2n^{\alpha}}\left(\p_{x'}([0] \leftrightarrow [n^{\beta}])\right)^{n^{1-\beta}}.
		\end{align*}
		\noindent
		Let $\Delta=\setof{u\in\Zd}{d(u,\Line)\leq n^{\alpha}/2}$ and $\Lambda_{n}(u)=\setof{v\in\Zd}{\norm{u-v}_{\scriptscriptstyle\infty}\leq n}$. Then,
		\begin{align*}
			\p_{x'}([0]\leftrightarrow [n^{\beta}])
			&\geq
			\Phi(\mathds{1}_{[0]\leftrightarrow [n^{\beta}]}(\eta)\IF{C_{[0]}(\omega)\cap \Line =\emptyset})\\
			&=
			\Phi(\mathds{1}_{[0]\leftrightarrow [n^{\beta}]}(\omega)\IF{C_{[0]}(\omega)\cap \Line =\emptyset})\\
			&=
			\p([0]\leftrightarrow [n^{\beta}])\p(\Line \nleftrightarrow [0] \given [0]\leftrightarrow [n^{\beta}])\\
			&=
			\p([0]\leftrightarrow [n^{\beta}])\p(\Line \nleftrightarrow [0] \given [0]\leftrightarrow [n^{\beta}], \Line_{<-n}\nleftrightarrow [0],\Line_{>2n}\nleftrightarrow [0])\\
			&\pushright{\times\p(\Line_{<-n}\nleftrightarrow [0],\Line_{>2n}\nleftrightarrow [0] \given [0]\leftrightarrow [n^{\beta}])}\\
			&\geq
			\frac{C}{n^{\beta(d-1)/2}}\, e^{-\xi n^{\beta}}
			\bigl(1-\p^w_{\Delta}(\Line_{[-n,2n]} \leftrightarrow \partial \Delta)\bigr)\\
			&\pushright{\times\bigl(1-\p( [0]\leftrightarrow \partial \Lambda_n([0]) \given [0]\leftrightarrow [n^{\beta}])\bigr)}\\
			&\geq
			\frac{C}{n^{\beta(d-1)/2}}\, e^{-\xi n^{\beta}}(1-3ne^{-cn^{\alpha}/2})(1-n^{d\beta}e^{-cn}) .
		\end{align*}
		Together with $\p_{x'}(0\leftrightarrow n\eone) \leq \p(0\leftrightarrow n\eone)$ when $x'\leq x$, we get the result.
		
		\item For the second point of~\ref{it:eq}, notice first that, for any edge $e\in\Line$, $\p_{x'}(\omega_e=1 \given \omega_{\setminus e})\geq 1/(1+\frac{q}{x'})$ uniformly on $\omega_{\setminus e}$. Therefore, by opening all edges from $\Line_{[0,n]}$,
		\begin{align*}
			\p_{x'}(0\leftrightarrow n\eone)&\geq e^{-\log(1+\frac{q}{x'})n}\geq e^{-\frac{q}{x'}n} .
		\end{align*}
		Choosing $x'$ such that $\frac{q}{x'}<\xi$, the result follows. Moreover, we obtain that $\iclx\leq \frac{q}{x'}$, which corresponds to one side of item \ref{it:asymBehav}.
		
		\item We now prove a variant of Lemma~\ref{lem:expoDecFKwired}, establishing exponential decay of connectivities uniformly over boundary conditions under the measure \(\p_{x'}\).
		\begin{lemma}
		\label{lem:UniformExpDecayLine}
		Assume that \(x<\xc\). Then, for any \(x'\geq 0\), there exists a constant $\Cl[expoThm]{ED1}=\Cr{ED1}(x,x',q,d)>0$ such that
		\begin{equation}\label{eq:ExpDecayLine}
		\p_{x',\Lambda_n(u)}^{\wired}(u\leftrightarrow\partial\Lambda_n(u)) \leq e^{-\Cr{ED1}n}
		\end{equation}
		uniformly over $u\in\Z^d$.		
		\end{lemma}
		\begin{proof}
		First observe that the claim is an immediate consequence of FKG and Lemma~\ref{lem:expoDecFKwired} when \(x'\leq x\). We thus assume from now on that \(x'>x\).

		Let us write
		\begin{multline}
		\label{eq:decompTouchLine}
		\p_{x',\Lambda_n(u)}^{\wired}(u\leftrightarrow\partial\Lambda_n(u))\\
		=
		\p_{x',\Lambda_n(u)}^{\wired}(u\leftrightarrow\partial\Lambda_n(u), u\nleftrightarrow\Line)
		+
		\p_{x',\Lambda_n(u)}^{\wired}(u\leftrightarrow\partial\Lambda_n(u), u\leftrightarrow\Line)
		\end{multline}
		and treat separately the two terms in the right-hand side. For the first term, we rely again on the coupling $\Phi(\omega,\eta)$ between $\p_{x,\Lambda_n(u)}^{\wired}$ and $\p_{x',\Lambda_n(u)}^{\wired}$ as above:
		\begin{align*}
		\p_{x',\Lambda_n(u)}^{\wired}(u\leftrightarrow\partial\Lambda_n(u), u\nleftrightarrow\Line)
		&=
		\Phi(u\stackrel{\eta}{\leftrightarrow}\partial\Lambda_n(u), u\stackrel{\eta}{\nleftrightarrow}\Line)\\
		&=
		\Phi(u\stackrel{\omega}{\leftrightarrow}\partial\Lambda_n(u), u\stackrel{\eta}{\nleftrightarrow}\Line)\\
		&\leq
		\Phi(u\stackrel{\omega}{\leftrightarrow}\partial\Lambda_n(u))
		=
		\p(u\leftrightarrow\partial\Lambda_n(u)),
		\end{align*}
		so that the claim follows again from Lemma~\ref{lem:expoDecFKwired}.

		Let us finally consider the second term in the right-hand side of~\eqref{eq:decompTouchLine}.
		The proof in this case relies on a coarse-graining procedure similar to the one used in~\cite{Campanino+Ioffe+Velenik-2008}. Fix a scale $K$ and a number $r$ (both of which will be later chosen sufficiently large, independently of $n$) and define
		\begin{align*}
			\Delta_k(v)=\llbracket-k, k\rrbracket^d+v,\quad
			\outerBox(v)=\Delta_{\barK}(v) ,
		\end{align*}
		where $\barK=K+r\log(K)$. Given a set of vertices $A\subset\Z^d$, we write $[A]_k=
		\bigcup_{v\in A}\Delta_k(v)$. Set $\Lambda\equiv\Lambda_{n+2K}(u)$ and $\Lambda_n\equiv\Lambda_n(u)$.
		
		Let \(F = \setof{i\in\Lambda_n\setminus [\Line]_{2K}}{i \leftrightarrow [\Line]_{2K}}\) and \(D=\partial^{\rm ext}[\Line]_{2K}\cap\Lambda_n\). Note that \(u\in F\cup [\Line]_{2K}\). We first coarse-grain the connected components of \(F\) using the following algorithm:

		\begin{figure}[t]
			\scalebox{.6}{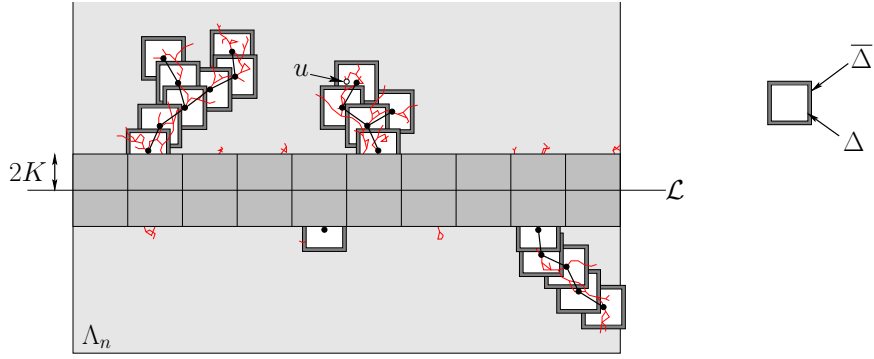}
			\caption{Coarse-graining of \(F\).}
			\label{fig:CG1}
		\end{figure}

		\begin{algorithm}[H]
			\label{alg:CoarseGrainingED}
			Set $V = \emptyset$, $n=1$\;
			\While{$\exists v\in F\cap D$ such that $v\xleftrightarrow{F\setminus [V]_{\barK}}\partial\Delta_{K}(v)$} {
				Let \(v\) be the smallest such vertex and add it to \(V\)\;
				set \(V_n = \{v\}\) and \(E_n=\emptyset\)\;
				\While{$\exists w\in F\cap\partial^{\rm ext} [V]_{\barK}$ such that $w\xleftrightarrow{F\setminus [V]_{\barK}}\partial\Delta_{K}(w)$} {
					Let \(w\) be the smallest such vertex and add it to \(V\) and to \(V_n\)\;
					Let \(w'\in V\) be the smallest vertex such that \(w\in \partial^{\rm ext} [w]_{\barK}\) and add the edge \(\{w',w\}\) to \(E_n\)\;
				}
				Set \(\tree_n = (V_n,E_n)\)\;
				Update \(n \rightarrow n+1\)\;
			}
			\caption{Coarse-graining procedure}
		\end{algorithm}
		
		\smallskip\noindent
		This algorithm yields a (possibly empty) family of trees \(\tree_1=(V_1,E_1),\ldots,\tree_M=(V_M,E_M)\), possessing the following properties:
		\begin{enumerate*}[label=(\roman*)]
		\item\label{prop:tree_i}   the root of each \(\tree_k\) belongs to \(D\);
		\item\label{prop:tree_ii}  every edge \(\{w,w'\}\in E_k\), \(1\leq k\leq M\), satisfies \(w'\in\partial^{\rm ext} [w]_{\barK}\);
		\item\label{prop:tree_iii} all connected components of \(F \setminus \bigcup_{i=1}^M [V_i]_{\barK}\) have (\(\ell^\infty\)-)diameter at most \(2\barK\).
		\end{enumerate*}

		In view of the property~\ref{prop:tree_i}, it is convenient to relabel the trees according to the position of their root. Namely, for any \(v\in D\), we denote by \(\tree_v=(V_v,E_v)\) the (possibly empty) tree with root at \(v\) obtained using the above algorithm.

		Denote by $C_K$ the number of vertices in $\partial\outerBox$. The number of possible configurations of the tree \(\tree_v\), with fixed root \(v\), is at most equal to the number of trees with branching number $C_K$, which is in turn at most $e^{\Cl{cst:KestenTree}\log(C_K)
		|V_v|}$ by an argument due to Kesten (see~\cite[Section~4.2]{Grimmett-1999}). Therefore, by Lemmas~\ref{lem:expoDecFKwired} and~\ref{lem:LfreeEst}
		(which can be applied provided we choose \(r\) large enough),
		the probability that the algorithm yields a given collection of trees \((\tree_v)_{v\in D}\) with total number of vertices \(N\) is bounded above by \(e^{-\Cr{decayFK} NK}e^{\Cr{cst:KestenTree}\log(C_K)N} \leq e^{-\frac12 \Cr{decayFK} NK}\) once \(K\) is chosen large enough.
		Therefore, for any \(\rho>0\), there exists \(K_0>0\) and \(c>0\) such that, for all \(K\geq K_0\),
		\begin{align}
			\p_{x',\Lambda_n}^{\wired}\bigl(\sum_{v\in D} |V_v| \geq \rho n/K \bigr)
			&\leq
			\sum_{N\geq\rho n/K} \sum_{\substack{\ell_v\geq 0, v\in D:\\\sum_v \ell_v = N}} e^{-\frac12 \Cr{decayFK} K \sum_v \ell_v} \notag\\
			&\leq
			\sum_{N\geq\rho n/K} e^{-\frac14 \Cr{decayFK} K N}
			\Bigl \{ \sum_{\ell=0}^\infty e^{-\frac14 \Cr{decayFK} K \ell} \Bigr\}^{|D|}
			\leq
			e^{-c\rho n},
		\label{eq:BoundSizeTrees}
		\end{align}
		for all \(n\) large enough. This immediately implies that, whenever $u$ connects to a side $H'$ of $\Lambda_n$ with $H'\cap\Line=\emptyset$, the desired exponential decay follows, since, in that case, \(\sum_{v\in D} |V_v| \geq n/(2K)\).

		\begin{figure}[t]
			\scalebox{.6}{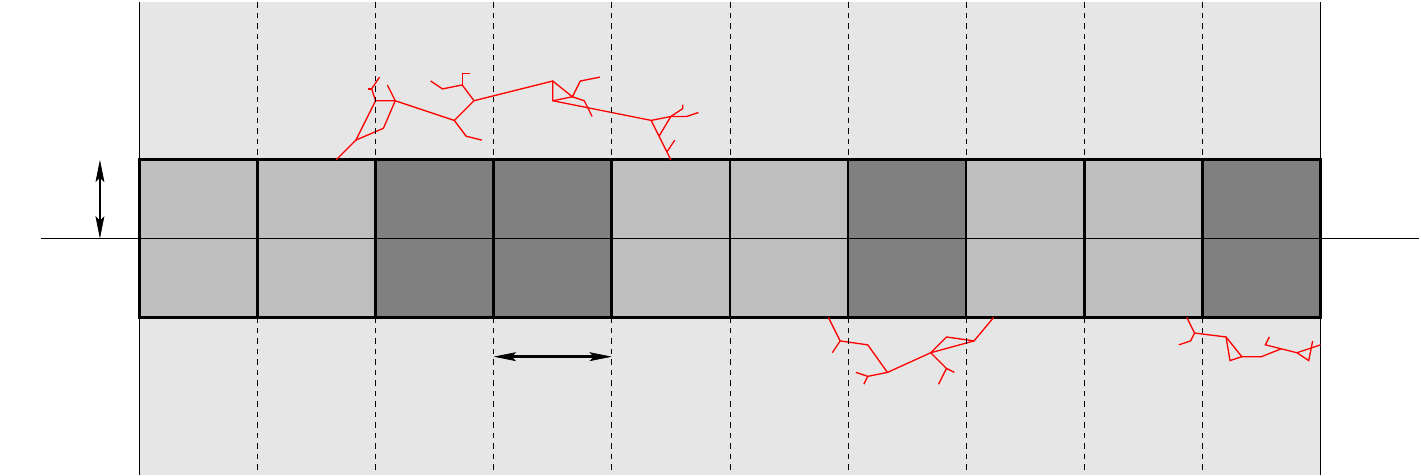}
			\caption{Splitting of \([\Line]_{2K}\) into boxes. The four covered boxes are darker (only the relevant clusters of \(F\) are drawn).}
			\label{fig:CoveredBoxes}
		\end{figure}
		
		It only remains to take care of connexions to the two sides of \(\Lambda_n\) intersecting \(\Line\); by symmetry, it suffices to consider the side with largest \(\eone\) component, which we denote by \(H\). Let us split $\Lambda_n$ into slices (see Figure~\ref{fig:CoveredBoxes}). Define
		\[
			S_{i}
			=
			\bigl( \llbracket (i-1)3K,i3K \rrbracket\times \Z^{d-1} +u-n\eone \bigr)\cap\Lambda_n(u),\ i=1,2,\ldots,2n/3K,
		\]
		and set \(B_i = S_{i}\cap [\Line]_{2K}\). We say that the box \(B_i\) is \textit{covered} if \(S_{i-1} \xleftrightarrow{F} S_{i+1}\) and \textit{uncovered} otherwise. Observe that, by property~\ref{prop:tree_iii} above, \(\sum_{v\in D} |V_v|\) cannot be smaller than the number of covered boxes. Denoting by \(\calB^+_{\rm uncov} = \setof{n/3K\leq i\leq 2n/3K}{B_i \text{ is uncovered}}\) the indices of all the uncovered boxes ``on the right of'' \(u\), it thus follows from~\eqref{eq:BoundSizeTrees} that there exists \(c>0\) such that, for all large \(n\),
		\[
			\p_{x',\Lambda_n}^{\wired}\bigl(u\leftrightarrow H, u\leftrightarrow\Line\bigr)
			\leq
			\p_{x',\Lambda_n}^{\wired}\bigl(u\leftrightarrow H, |\calB^+_{\rm uncov}|\geq n/6K\bigr) + e^{-c n}.
		\]
		The proof will be complete once we prove that the first term in the right-hand side decays exponentially with \(n\). Let us decompose
		\begin{multline*}
			\p_{x',\Lambda_n}^{\wired}\bigl(u\leftrightarrow H, |\calB^+_{\rm uncov}|\geq n/6K\bigr)\\
			=
			\sum_{A:\,|A|\geq n/6K}
			\p_{x',\Lambda_n}^{\wired}\bigl(u\leftrightarrow H \given \calB^+_{\rm uncov} = A\bigr)\,
			\p_{x',\Lambda_n}^{\wired}\bigl(\calB^+_{\rm uncov} = A\bigr) .
		\end{multline*}
		Observe now that, in order for \(u\) to be connected to \(H\), it is necessary that none of the boxes \(B_i\), \(i\in\calB^+_{\rm uncov}\) is \textit{empty}, in the sense of all the edges inside of it being closed.
		Clearly, \(\calB_{\rm uncov}\) only depends on the state of the edges in \(\Lambda_n\setminus [\Line]_{2K}\).
		Since the probability that all the edges inside an uncovered box \(B_i\) are closed is bounded below by \(\nfe^{2d|B_i|}>0\), uniformly in the state of all the other edges, we conclude that
		\begin{align*}
			\p_{x',\Lambda_n}^{\wired}\bigl(u\leftrightarrow H \given \calB^+_{\rm uncov} = A\bigr)
			&\leq
			\p_{x',\Lambda_n}^{\wired}\bigl(B_i \text{ not empty, for all }i\in\calB^+_{\rm uncov} \given \calB^+_{\rm uncov} = A\bigr) \\
			&\leq
			\bigl( 1-\nfe^{|B_i|} \bigr)^{n/6K} ,			
		\end{align*}
		and the conclusion follows.

		\end{proof}

		\begin{remark}
		\label{rem:uniqueness}
		Note that, using a standard coupling argument, \eqref{eq:ExpDecayLine} implies that there is a unique infinite-volume random-cluster measure for any
		\(x'\geq 0\). Since there is a.s.\ no infinite cluster under this measure, we conclude from the Edwards--Sokal coupling that there is a unique infinite-volume Potts measure for any finite value of \(J\).
		\end{remark}
		\item We can now prove the other half of item~\ref{it:asymBehav}. Notice that the same procedure as in the previous point ensures that,
		on the event $0\leftrightarrow n\eone$, we can find $K\equiv K(x,d), \Cl{uncovBox}\equiv\Cr{uncovBox}(K,d,x)$ (uniformly in \(x'\)) such that at least half of the boxes $B_i$ are uncovered with $\p_{x'}$-probability at least $1-e^{-\Cr{uncovBox}n}$. Then, by finite energy, we can find $\epsilon\equiv\epsilon(K,d,x)>0$ and $\Cl{pivEdge}\equiv\Cr{pivEdge}(K,d,x)$ such that at least $\epsilon n/K$ boxes contain an edge in $\Line$ that is pivotal for $0\leftrightarrow n\eone$ with $\p_{x'}$-probability at least $1-e^{-\Cr{pivEdge}n}$ (again, both $\epsilon$ and $\Cr{pivEdge}$ do not depend on $x'$). Denote this event $B_{\epsilon}$. Then, proceeding as before,
		\begin{align*}
			\p_{x'}(0\leftrightarrow n\eone)
			&\leq \p_{x'}(B_{\epsilon}, 0\leftrightarrow n\eone) + e^{-\Cr{pivEdge}n}\\
			&\leq \Bigl(\frac{x'}{1+x'}\Bigr)^{\frac{\epsilon n}{K}} + e^{-\Cr{pivEdge}n}\\
			&\leq e^{-\log(1+\frac{1}{x'})\frac{\epsilon n}{K}}
			(1+e^{-(\Cr{pivEdge}- \log(1+\frac{1}{x'})\frac{\epsilon n}{K})}).
		\end{align*} 
		Choosing $x'$ such that $\Cr{pivEdge}- \log(1+\frac{1}{x'})\frac{\epsilon}{K}>0$, we obtain
		\begin{equation*}
			\iclx
			=
			\lim\limits_{n\to\infty}-\frac{1}{n}\log\p_{x'}\left(0\leftrightarrow n\eone\right)
			\geq
			\log\bigl(1+\frac{1}{x'}\bigr)\frac{\epsilon}{K}
			\geq
			\frac{\epsilon}{2K}\frac{1}{x'},
		\end{equation*}
		for $x'$ large enough.
		
		\item To prove continuity, we work again in large but finite boxes (following Remark~\ref{rem:finVol}). We start with a small computation (which will be used again in Section~\ref{sec:UB}). Let $x'_1\leq x'_2$ and write  $\lambda=\log(x'_2/x'_1)$. Then,
		\begin{align*}
			\xi_{x'_1}-\xi_{x'_2}
			&=
			\lim_{n\to\infty} \frac{1}{n}\log\frac{\p_{x'_2}(0 \leftrightarrow n\eone)}{\p_{x'_1}(0 \leftrightarrow n\eone)}\\
			&=
			\lim_{n\to\infty} \frac{1}{n}\log\frac{\e_{x'_1}\bigl[e^{\lambda \ouv_{\Line}}\bigm|0\leftrightarrow n\eone\bigr]}{\e_{x'_1}\bigl[e^{\lambda \ouv_{\Line}}\bigr]} .
		\end{align*}
		Now, we partition the numerator in the logarithm w.r.t.\ the cluster of $0$:
		\begin{align}
			\label{eq:upBndRep}
			\frac{\e_{x'_1}\bigl[e^{\lambda \ouv_{\Line}}\bigm| 0\leftrightarrow n \eone\bigr]}{\e_{x'_1}\bigl[e^{\lambda \ouv_{\Line}}\bigr]}
			&=
			\sum_{C\ni 0,n\eone} \frac{1}{\e_{x'_1}\bigl[e^{\lambda \ouv_{\Line}}\bigr]} \p_{x'_1}\bigl(C_0=C \bigm| 0\leftrightarrow n \eone\bigr)\notag\\
			&\pushright{\times\e_{x'_1}\bigl[e^{\lambda \ouv_{\Line}} \bigm| C_0=C\bigr]}\notag\\
			&\leq
			\sum_{C\ni 0,n\eone} \frac{\p_{x'_1}\bigl(C_0=C \bigm| 0\leftrightarrow n \eone\bigr)}{\e_{x'_1}\bigl[e^{\lambda \ouv_{\Line\setminus C}} \bigm| \partial C \text{ closed}\bigr]} e^{\lambda|C\cap\Line|}\notag\\
			&\pushright{\times\e_{x'_1}\bigl[e^{\lambda \ouv_{\Line\setminus C}} \bigm| \partial C \text{ closed}\bigr]}\notag\\
			&=
			\sum_{C\ni 0,n\eone} \p_{x'_1}\bigl(C_0=C \bigm| 0\leftrightarrow n \eone\bigr) e^{\lambda|C\cap\Line|}\notag\\
			&=
			\e_{x'_1}\bigl[e^{\lambda|C_0\cap\Line|} \bigm| 0\leftrightarrow n \eone\bigr] .
		\end{align}
		Partitioning w.r.t.\ the leftmost and rightmost point of $C_0\cap\Line$ (denoted $L$ and $R$), we then obtain
		\begin{align*}
			\e_{x'_1}\bigl[e^{\lambda|C_0\cap\Line|} \bgiven 0\leftrightarrow n \eone\bigr]
			&\leq
			\sum_{l=-\infty}^{0}\sum_{r=n}^{\infty} e^{\lambda(r-l)}\p_{x'_1}(L=l,R=r \given 0\leftrightarrow n\eone)\\
			&\leq e^{\lambda 3n} + e^{\lambda 3n} \sum_{l=0}^{\infty}\sum_{r=0}^{\infty} e^{\lambda(r+l)}\frac{\p_{x'_1}(0\leftrightarrow (3n+l+r)\eone)}{\p_{x'_1}(0\leftrightarrow n\eone)}\\
			&\leq e^{\lambda 3n}\Bigl(1 + \sum_{l=0}^{\infty}\sum_{r=0}^{\infty} e^{\lambda(r+l)}e^{-\xi_{x'_1}(3n+l+r)}e^{\xi_{x'_1}(1+o(1))n}\Bigr)\\
			&\leq e^{\lambda 3n}\Bigl(1 + e^{-\xi_{x'_1}n} \sum_{l=0}^{\infty}\sum_{r=0}^{\infty} e^{-(\xi_{x'_1}-\lambda)(l+r)}\Bigr).
		\end{align*}
		Note that $\lambda<\xi_{x'_1}$ when $x'_2$ is close enough to $x'_1$ (since $\xi_{x'_1}>0$). In this case, the last double sum converges and we get
		\begin{align*}
			\xi_{x'_1}-\xi_{x'_2}
			&\leq
			\lim_{n\to\infty} \frac{1}{n} \log\bigl(e^{\lambda 3n}(1 + e^{-\xi_{x'_1}n}C_{x'_2-x'_1})\bigr)\\
			&\leq 3\bigl(\log(x'_2)-\log(x'_1)\bigr)\\
			&=\int_{x_1'}^{x_2'}\frac{3}{s} \,\dd s \leq \frac{3}{x_1'}(x_2'-x_1') .\qedhere
		\end{align*}
	\end{itemize}
\end{proof}


\section{Random Walk representation}
\label{sec:RWrep}

In this section, we explain how one can couple the cluster $C_0$ under $\p(\cdot\given 0\leftrightarrow n\eone)$ (remember that \(\p\) denotes the homogeneous (that is, when \(x'=x\)) random-cluster measure on $\Zd$) with a directed random walk on \(\Zd\) for all \(x<\xc\). This coupling will allow us to analyze in detail the large-scale properties of \(C_0\).
The construction is based on the decomposition of the cluster into irreducible pieces, as described in~\cite{Campanino+Ioffe+Velenik-2008}, and on the arguments exposed in Appendix~\ref{sec:RenLongRange} that explain how to get rid of the dependency between the irreducible pieces.
The exposition is not self-contained and its goal is mostly to setup notations and remind the reader of the main steps of the construction. A reader not familiar with~\cite{Campanino+Ioffe+Velenik-2008} should refer to that work for details and additional explanations.

We start with a brief description of the coarse-graining in~\cite{Campanino+Ioffe+Velenik-2008}. As we only consider the $\eone$ direction, the construction simplifies slightly.
Let us first introduce the geometric objects required for the coarse-graining procedure. Let $0<\psi\leq\pi/2$ and let
\[
	\fcone_\psi = \setof{u\in \Zd}{\langle u, \eone\rangle \geq \norm{u}_{\scriptscriptstyle 2} \cos(\psi/2)}
\]
be the cone of angular aperture \(\psi\) and axis direction \(\eone\); we will usually omit \(\psi\) from the notation and simply write $\fcone$.
We also set \(\bcone = -\fcone\) and introduce the ``diamonds''
\[
\diam(v_1,v_2) = (v_1+\fcone)\cap(v_2+\bcone) .
\]
Let us say that $v\in C_0$ is a \emph{cone-point} if 
\[
	C_0\subset (v+\bcone) \cup (v+\fcone).
\]

We introduce three families of clusters:
\begin{itemize}
\item \(\bc_L\) is the set of all clusters \(C\) such that: \(0\in C\); \(C\) has a cone-point \(v\) such that \(C\subset v+\bcone\); \(C\) possesses no other cone-point with nonnegative \(\eone\)-coordinate.
\item \(\bc_R\) is the set of all clusters \(C\) such that: \(0\in C\); \(C\) has a cone-point \(v\) such that \(C\subset v+\fcone\); \(C\) possesses no other cone-point with nonpositive \(\eone\)-coordinate.
\item \(\alp\) is the set of all clusters \(C\) such that: \(C\) possesses exactly two cone-points, \(0\) and \(v\in\fcone\), and \(C\subset \diam(0,v)\).
\end{itemize}
(Note that the single-vertex cluster \(\{0\}\) belongs to both \(\bc_L\) and \(\bc_R\).)
We define a \emph{displacement} application \(D\) from each of these three sets into \(\fcone\) by setting (\(v\) is the vertex appearing in the previous definitions): \(D(C)=v\) when \(C\in\bc_L\cup\alp\) and \(D(C)=-v\) when \(C\in\bc_R\).

\begin{figure}[ht]
\centering
\resizebox{!}{1.7cm}{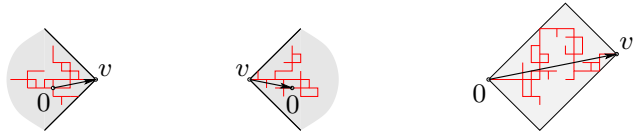}
\caption{Example of clusters in \(\bc_L\), \(\bc_R\) and \(\alp\). The corresponding values \(D\) are depicted as vectors.}
\end{figure}

A cluster \(\gamma^b\in\bc_L\) can be naturally concatenated with \(m\geq 0\) clusters \(\gamma_1,\dots,\gamma_m\in\alp\) by first translating each \(\gamma_k\) by \(D(\gamma^b)+D(\gamma_1)+\dots+D(\gamma_{k-1})\). We can then also concatenate a cluster \(\gamma^f\in\bc_R\) by first translating it by
\(D(\gamma^b)+\sum_{k=1}^m D(\gamma_k) + D(\gamma^f)\).
The resulting object is denoted
\(
\gamma^b \sqcup \gamma_1 \sqcup \dots \sqcup \gamma_m \sqcup \gamma^f
\).

Let \(v_1,v_2,\ldots,v_{m+1}\) be all the cone-points of \(C_0\) with \(\eone\)-coordinate in \(\{0,\ldots,n\}\). We assume that they are ordered according to increasing $\eone$ coordinates. (We also assume that \(m\geq 1\), since this will occur with high probability, as explained below.)
These vertices induce a decomposition of \(C_0\) into a string of \(m\) irreducible components (belonging to \(\alp\)) and two boundary-components (belonging, respectively, to \(\bc_L\) and \(\bc_R\)):
\[
C_0 = \gamma^b \sqcup \gamma_1 \sqcup \cdots \sqcup \gamma_m \sqcup \gamma^f.
\]
Note that all the pieces are unambiguously identified after inverting the translations due to the concatenation, except for \(\gamma^f\). The latter ambiguity disappears if we impose that \(D(\gamma^b,\gamma_1,\dots,\gamma_m,\gamma^f) \equiv D(\gamma^b)+\sum_k D(\gamma_k)+D(\gamma^f) = n\eone\).

\begin{figure}[ht]
\centering
\resizebox{!}{2.5cm}{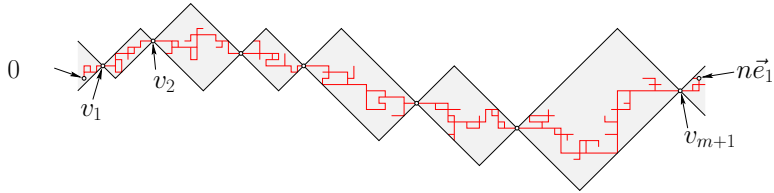}
\caption{The decomposition of the common cluster of \(0\) and \(n\eone\) into irreducible pieces.}
\end{figure}

As shown in~\cite{Campanino+Ioffe+Velenik-2008}, there exists \(\Cl[expoThm]{nu3}>0\) and \(\Cl[expoThm]{nu3b}>0\) such that the number of irreducible pieces is at least $\Cr{nu3b}n$ with \(\p(\,\cdot\given 0\leftrightarrow n\eone)\)-probability at least $1-e^{- \Cr{nu3} n}$. In particular,
\begin{align}
	\label{eq:irrRep}
	\p(0\leftrightarrow n\eone)
	&=
	(1+o(1))\!\! \sum_{m\geq \nu_2 n}\sum_{\substack{\gamma^b,\gamma_1,\ldots,\gamma_m,\gamma^f:\\D(\gamma^b,\gamma_1,\dots,\gamma_m,\gamma^f) = n\eone}} \!\!\!\!\! \p(C_0=\gamma^b\sqcup\gamma_1\sqcup\cdots\sqcup\gamma_m\sqcup\gamma^f) \notag\\
	&=
	(1+o(1))\!\! \sum_{m\geq \nu_2 n}\sum_{\substack{\gamma^b,\gamma_1,\ldots,\gamma_m,\gamma^f:\\D(\gamma^b,\gamma_1,\dots,\gamma_m,\gamma^f) = n\eone}} \!\!\!\!\! \p(\Gamma^b\cap\Gamma_1\cap\cdots\cap\Gamma_m\cap\Gamma^f),
\end{align}
where the percolation events \(\Gamma^b,\Gamma_k\) and \(\Gamma^f\) are defined as follows. Let \(\vec\gamma_k\) be the translate of \(\gamma_k\) obtained after the concatenation operation and denote by \(v_{k},v_{k+1}\) the corresponding cone-points. We set
\begin{gather*}
\mathcal{W}(\vec\gamma_k) = \{\text{all edges of \(\vec\gamma_k\) are open}\}, \\
\bar\partial{\vec\gamma_k} = (\partial\vec\gamma_k)\setminus \bigl\{\{v_k-\eone,v_k\}, \{v_{k+1},v_{k+1}+\eone\}\bigr\}, \\
\mathcal{N}(\vec\gamma_k) = \{\text{all edges of \(\bar\partial\vec\gamma_k\) are closed}\}, \\
\Gamma_k = \mathcal{W}(\vec\gamma_k) \cap \mathcal{N}(\vec\gamma_k);
\end{gather*}
the definitions of \(\Gamma^b,\Gamma^f\) are completely similar.

In order to apply the results of Appendix~\ref{sec:RenLongRange}, let us reformulate the above in the language of Appendix~\ref{sec:RenLongRange} (see the latter for details). Let us write \(D_1(\gamma)=D(\gamma)\cdot \eone\) and set
\begin{align*}
\subProb(\gamma^b)
&=
e^{\xi D_1(\gamma^b)}\, \p(\Gamma^b), \\
\subProb_m(\gamma^b\sqcup\gamma_1\sqcup\dots\sqcup\gamma_m\sqcup\gamma^f)
&=
e^{\xi D_1(\gamma^b,\gamma_1,\ldots,\gamma_m,\gamma^f)}\, \p(\Gamma^b\cap \Gamma_1\cap \ldots\cap \Gamma_m\cap \Gamma^f), \\
\subProb(\gamma_k\given \gamma^b\sqcup\gamma_1\sqcup\dots\sqcup\gamma_{k-1})
&=
e^{\xi D_1(\gamma_k)}\, \p(\Gamma_k\given\Gamma^b\cap \Gamma_1\cap \ldots\cap \Gamma_{k-1}), \\
\subProb(\gamma^f\given \gamma^b\sqcup\gamma_1\sqcup\dots\sqcup\gamma_{m})
&=
e^{\xi D_1(\gamma^f)}\, \p(\Gamma^f\given\Gamma^b\cap \Gamma_1\cap \ldots\cap \Gamma_{m}).
\end{align*}
Of course, with these definitions, we have
\begin{multline*}
\subProb_m(\gamma^b\sqcup\gamma_1\sqcup\dots\sqcup\gamma_m\sqcup\gamma^f)\\
=
\subProb(\gamma^b)
\Bigl( \prod_{k=1}^m \subProb(\gamma_k\given \gamma^b\sqcup\gamma_1\sqcup\dots\sqcup\gamma_{k-1}) \Bigr)
\subProb(\gamma^f\given \gamma^b\sqcup\gamma_1\sqcup\dots\sqcup\gamma_{m}),
\end{multline*}
as desired. Moreover, the required properties are satisfied.
\begin{proposition}\label{pro:PropertiesOK}
Properties~\ref{hyp:unifSumability}, \ref{hyp:expMix}, \ref{hyp:subExpDecOfMass}, \ref{hyp:azero}, \ref{prop:expDecSteps}, \ref{prop:directed}, \ref{prop:aperiod}, \ref{prop:irred} and \ref{prop:trajSym} of Appendix~\ref{sec:RenLongRange} all hold in the present setting.
\end{proposition}
\begin{proof}
\ref{hyp:unifSumability}, \ref{hyp:expMix} and~\ref{prop:expDecSteps} are direct consequences of~\cite[Theorem~2.2]{Campanino+Ioffe+Velenik-2008}.
\ref{prop:directed} and \ref{prop:trajSym} are obvious. \ref{hyp:azero} and~\ref{prop:aperiod} hold by finite energy, since the edge \(\{(0,0_{d-1}),(1,0_{d-1})\}\) belongs to \(\alp\);
we can argue similarly for \ref{prop:irred}.
Let us check~\ref{hyp:subExpDecOfMass}.

\newcommand{\uga}{\underline{\gamma}}
To shorten notation, we simply write \(\uga_m = (\gamma^b,\gamma_1,\dots,\gamma_m,\gamma^f)\).
We first assume that \(s>1\). In this case, it follows from~\eqref{eq:irrRep} that
\begin{align*}
\sum_{m\geq 0} s^m \sum_{\uga_m} \subProb_m(\uga_m)
&\geq
\sum_{n\geq 1} s^{\Cr{nu3b}n} \sum_{m = \Cr{nu3b}n}^n  \sum_{\substack{\uga_m: \\ D(\uga) = n\eone}} \subProb_m(\uga_m) \\
&\geq
C \sum_{n\geq 1} s^{\Cr{nu3b}n} e^{\xi n} \p(0\leftrightarrow n\eone) = +\infty,
\end{align*}
since \(\p(0\leftrightarrow n\eone) = e^{-\xi n (1+o(1))}\).
Let us now assume that \(s<1\).
Since, by FKG, \(\p(0\leftrightarrow (n,x)) \leq e^{-\xi n}\) for any \(x\in\Z^{d-1}\) and \(n\geq 1\), it again follows from~\eqref{eq:irrRep} that
\begin{align*}
\infty
>
\sum_{n\geq 1} s^n \sum_{\substack{x\in\Z^{d-1}:\\(n,x)\in\fcone}} e^{\xi n} \p(0\leftrightarrow (n,x))
&\geq
\sum_{n\geq 1} s^n \sum_{m = \nu_2 n}^n \sum_{\substack{\uga_m: \\ D_1(\uga) = n}} \subProb_m(\uga_m) \\
&\geq
\sum_{m\geq 1} s^{m/\Cr{nu3b}} \sum_{n = m}^{m/\Cr{nu3b}} \sum_{\substack{\uga_m: \\ D_1(\uga) = n}} \subProb_m(\uga_m) \\
&\geq
C \sum_{m\geq 1} s^{m/\Cr{nu3b}} \sum_{\uga_m} \subProb_m(\uga_m).
\end{align*}
We conclude that the radius of convergence of \(z\mapsto \sum_{m\geq 1} z^m \sum_{\uga_m} \subProb_m(\uga_m)\) is equal to \(1\), which establishes~\ref{hyp:subExpDecOfMass}.
\end{proof}

In view of the above, we can import the results of Appendix~\ref{app:Ren_Setting} to the present context.
Let $\calS=\bigcup_{n\geq 0}\alp^n$, $\calS^*=\bigcup_{n\geq 1}\alp^n$, $\compl{\bc}_L = \setof{b_L\sqcup \mathbf{x}}{b_L\in\bc_L, \mathbf{x}\in\calS}$ and $\compl{\bc}_R=\setof{\mathbf{x}\sqcup b_R}{b_R\in\bc_R, \mathbf{x}\in\calS}$. One can then define (see~\eqref{eq:weights}) two finite, positive measures \(\rho_L\) and \(\rho_R\) on \(\compl{\bc}_L\) and \(\compl{\bc}_R\) respectively, and a probability measure \(\statMes\) on \(\calS^*\).

To any family \(\boldsymbol{\gamma}=(\gamma^b,\gamma_1,\dots,\gamma_m,\gamma^f)\), with \(m\geq 1\), we can associate uniquely a cluster of \(0\) (not necessarily containing \(n\eone\)), with cone-points \(v_1,\dots,v_{m+1}\) (more precisely: \(\gamma^b\) is not translated, while the other ones are concatenated as explained above). Any subset \(\bfx=\{x_1,\dots,x_{k+1}\}\subset\{v_1,\dots,v_{m+1}\}\) induces a decomposition \(\boldsymbol{\tilde\gamma}=(\tilde\gamma^b,\tilde\gamma_1,\dots,\tilde\gamma_k,\tilde\gamma^f)\) by concatenating irreducible pieces not separated by cone-points in \(\bfx\). We then introduce a (positive, finite) measure on triples \((\boldsymbol{\tilde\gamma},\bfx,y)\), with \(y\in\fcone\), by setting
\[
\hat{\mathbb{Q}}(\boldsymbol{\tilde\gamma}, \bfx, y) = \IF{D(\boldsymbol{\tilde\gamma})=y}\, \rho_L(\tilde\gamma^b) \, \rho_R(\tilde\gamma^f) \, \prod_{i=1}^k \statMes(\tilde\gamma_i).
\]

By Lemma~\ref{lem:decoupling} and Theorem~\ref{thm:procToRW}, there exists \(c>0\) such that, for any bounded function \(f\) of the cluster \(C_0\),
\begin{equation}
\label{eq:C0marginal}
\bigl| \hat{\mathbb{Q}}(f, D(\boldsymbol{\tilde\gamma})=n\eone) - e^{\xi n} \e(f \IF{0\leftrightarrow n\eone}) \bigr|  \leq e^{-cn},
\end{equation}
for all \(n\) large enough.

Given \(k\) distinct vertices \(x_1,\ldots,x_k\) such that \(x_1\in\fcone\), \(x_k\in n\eone+\bcone\) and \(x_{i+1}\in x_i+\fcone\) for \(1\leq i<k\), and an additional vertex \(y\in x_k+\fcone\), we can write
\[
\hat{\mathbb{Q}}(\bfx,y) =  \hat\rho_L(x_1) \, \hat\rho_R(y-x_{k}) \, \prod_{i=1}^{k-1} \hat\statMes(x_{i+1}-x_i),
\]
where \(\hat\rho_L\), \(\hat\rho_R\) and \(\hat\statMes\) are the push-forwards of \(\rho_L\), \(\rho_R\) and \(\statMes\) by the displacement map \(D\). By Theorem~\ref{thm:procToRW}, these measures have exponential tails: there exist \(c>0\) and \(c_p>0\) such that, for any \(x\in\fcone\),
\begin{equation}
\label{eq:expTailsrhop}
\hat\rho_L(x) \leq e^{-c\norm{x}},
\quad
\hat\rho_R(x) \leq e^{-c\norm{x}},
\quad
\hat\statMes(x) \leq e^{-c_p\norm{x}},
\end{equation}
for all \(\norm{x}\) large enough. Moreover,
\begin{equation}
\label{eq:rhoLrhoRpositive}
\hat{\rho}_L(0) > 0,\qquad
\hat{\rho}_R(0) > 0,
\end{equation}
as follows from Lemma~\ref{lem:finEnergProc} and Remark~\ref{rem:BCtrivial}.

Let us denote by \(\mathbf{P}_u\) and \(\mathbf{E}_u\) the distribution and expectation associated to the random walk \((S_\ell)_{\ell\geq 0}\), starting at \(u\in\Zd\) with transition probabilities given by \(\hat\statMes\). Write $X_i=S_i-S_{i-1}$ for the increments of \(S\). Let also \(\calR(v) = \{\exists \ell\geq 0:\, S_\ell=v\}\) and set \(\mathbf{P}_{u,v}=\mathbf{P}_u(\,\cdot\given \calR(v))\) and \(\mathbf{E}_{u,v}=\mathbf{E}_u[\,\cdot\given \calR(v)]\).

As a direct consequence of~\eqref{eq:C0marginal}, observe that
\[
e^{\xi n} \p(0\leftrightarrow n\eone) = \sum_{\substack{u\in\fcone\\v\in n\eone+\bcone}} \hat\rho_L(u)\, \hat\rho_R(n\eone-v)\, \mathbf{P}_{u}(\calR(v)) + e^{-cn},
\]
for some \(c>0\). By the local limit theorem (see~\cite{Ioffe-2015}), uniformly in \(u,v\) such that \(\norm{u},\norm{n\eone-v}\leq n^{1/2-\alpha}\) (for some fixed \(\alpha>0\)),
\begin{equation}
\label{eq:RatioRuv}
\frac{\mathbf{P}_u(\calR(v))}{\mathbf{P}_0(\calR(n\eone))} = 1 + o(1),
\end{equation}
and
\begin{equation}
\label{eq:LLT}
\mathbf{P}_0(\calR(n\eone)) = (1 + o(1))\, \Cl{Cst:LLT}\, n^{-(d-1)/2}.
\end{equation}
In particular, we obtain the Ornstein-Zernike asymptotics:
\begin{equation}
\label{eq:OZzzz}
e^{\xi n} \p(0\leftrightarrow n\eone) = (1 + o(1))\, \Cl{Cst:OZ}\, n^{-(d-1)/2} ,
\end{equation}
with \(\Cr{Cst:OZ} = \Cr{Cst:LLT} \sum_{\substack{u\in\fcone\\v\in n\eone+\bcone}} \hat\rho_L(u)\, \hat\rho_R(n\eone-v)\).


\section{Upper bound on \(\xi_{x}-\iclx\)}
\label{sec:UB}

In this section, we prove the upper bounds in Items~\ref{item:Critical2d} and~\ref{item:Critical3d} of Theorem~\ref{thm:PinnedRegime}, as well the \(d\geq 4\) part of Theorem~\ref{thm:Jc_Potts}.
The argument in this section is a variant of the argument in~\cite{Friedli+Ioffe+Velenik-2013}, which applied to the case of Bernoulli percolation.

We work once more in large but finite volumes (as explained in Remark~\ref{rem:finVol}).
In view of Theorem~\ref{thm:BasicPropPotts}, we can assume, without loss of generality, that \(x'\geq x\). In particular, \(\lambda=\log(x'/x)\geq 0\).
By~\eqref{eq:upBndRep}, we have the upper bound
\begin{align*}
	\frac{\e\bigl[e^{\lambda \ouv_{\Line}} \bigm| 0\leftrightarrow n \eone\bigr]}{\e\bigl[e^{\lambda \ouv_{\Line}}\bigr]}
	\leq
	\e\bigl[e^{\lambda |\Line\cap C_0|} \bigm| 0\leftrightarrow n \eone\bigr].
\end{align*}

By Lemma~\ref{lem:UniformExpDecayLine} and~\eqref{eq:C0marginal}, there exist \(\lambda_0>0\) and \(c>0\) such that, for any \(\lambda<\lambda_0\),
\begin{align*}
e^{\xi n}\e\bigl[e^{\lambda |\Line\cap C_0|} \IF{0\leftrightarrow n \eone}\bigr]
&\leq
e^{\xi n}\e\bigl[e^{\lambda |\Line\cap C_0|} \IF{0\leftrightarrow n \eone, |\Line\cap C_0|\leq c'n}\bigr] + e^{-cn}\\
&\leq
\hat{\mathbb{Q}}(e^{\lambda |\Line\cap C_0|}, D(\boldsymbol{\tilde\gamma})=n\eone) + e^{-cn}.
\end{align*}
In particular, using~\eqref{eq:OZzzz} and~\eqref{eq:LLT},
\[
\e\bigl[e^{\lambda |\Line\cap C_0|} \given 0\leftrightarrow n \eone\bigr] \leq \C \hat{\mathbb{Q}}(e^{\lambda |\Line\cap C_0|} \given D(\boldsymbol{\tilde\gamma})=n\eone) + e^{-cn},
\]
for some \(c>0\).
Then,
\begin{align*}
	|\Line\cap C_0|
	&=
	|\tilde\gamma^b\cap\Line| + |\tilde\gamma^f\cap\Line| + \sum_{i=1}^{M} |\tilde\gamma_i\cap\Line|\\
	&\leq
	|\tilde\gamma^b\cap\Line| + |\tilde\gamma^f\cap\Line| + \sum_{i=1}^{M} |D(S_{i-1},S_i)\cap\Line|\\
	&\leq
	|\tilde\gamma^b\cap\Line| + |\tilde\gamma^f\cap\Line| + \sum_{i=1}^{M} X_i^{\parallel}\IF{X_i^{\parallel}\geq|S_{i-1}^{\perp}|} ,
\end{align*}
where the last inequality relies on the fact that the angle of the ``diamonds'' is at most $\pi/2$ and the fact that a step cannot cross the line if its parallel component is smaller than the distance between its starting point and the line. We get
\begin{align*}
	\hat{\mathbb{Q}}(e^{\lambda |\Line\cap C_0|} \given D(\boldsymbol{\tilde\gamma})=n\eone)
	\leq C\sum_{u,v}\rho_b^{\lambda}(u)\rho_f^{\lambda}(n\eone-v)\mathbf{E}_{u,v}\Bigl[e^{\lambda\sum_{i=1}^{M}X_i^{\parallel}\IF{X_i^{\parallel}\geq|S_{i-1}^{\perp}|}} \Bigr] ,
\end{align*}
where \(\rho_b^{\lambda}(y)\) and \(\rho_f^{\lambda}(y)\) decay exponentially in \(\norm{y}\), provided that \(\lambda\) be small enough.
In particular, we can restrict the sum to the pairs $u,v$ with $|u|,|v|\leq n^{1/2-\alpha}$ and we have $\sum_{u,v}\rho_b^{\lambda}(u)\rho_f^{\lambda}(v)<\infty$. At this stage, notice that the problem has been reduced to the analysis of a variant of the random-walk pinning problem. Then, for any \(m_0\geq 1\), we can write
\begin{align}
\label{eq:m0bnd}
	\ebf_{u,v}\Bigl[ e^{\lambda\sum_{i=1}^{M} X_i^{\parallel} \IF{X_i^{\parallel}\geq|S_{i-1}^{\perp}|}} \Bigr]
	&=
	\sum_{m=1}^{n} \ebf_{u,v}\Bigl[ e^{\lambda\sum_{i=1}^{m} X_i^{\parallel} \IF{X_i^{\parallel}\geq|S_{i-1}^{\perp}|}},M=m \Bigr]\notag\\
	&\leq
	e^{\lambda n} \pbf_{u,v}(M<m_0) + n\sup_{m_0 \leq m\leq n} A_{u,v}(m) ,
\end{align}
with 
\begin{align*}
	A_{u,v}(m)
	&=
	\ebf_{u,v}\Bigl[ e^{\lambda\sum_{i=1}^{m} X_i^{\parallel} \IF{X_{i}^{\parallel}\geq|S_{i-1}^{\perp}|}} \Bigr]\\
	&\leq
	O(n^{(d-1)/2}) \sum_{k=1}^{m} \sum_{\substack{\ell_1,\ldots,\ell_k\\\sum \ell_i = m}} \ebf_{u}\Bigl[\prod_{i=1}^{k} (e^{\lambda X_{t_i}^{\parallel}}-1)\; \IF{X_{t_i}^{\parallel}\geq |S_{t_i-1}^{\perp}|} \Bigr]\\
	&\leq
	O(n^{(d-1)/2}) \sum_{k=1}^{m} \lambda^k \sum_{\substack{\ell_1,\ldots,\ell_k\\\sum \ell_i = m}} \ebf_{u}\Bigl[\prod_{i=1}^{k} e^{\lambda X_{t_i}^{\parallel}} X_{t_i}^{\parallel}\; \IF{X_{t_i}^{\parallel} \geq |S_{t_i-1}^{\perp}|} \Bigr] ,
\end{align*}
where $t_i=\sum_{j=1}^{i} \ell_j$. The first inequality is obtained using~\eqref{eq:LLT}, writing
\[
\prod_{i=1}^{m}e^{\lambda X_i^{\parallel}\IF{X_{i}^{\parallel}\geq|S_{i-1}^{\perp}|}} = \prod_{i=1}^{m} \bigl( (e^{\lambda X_i^{\parallel}} - 1) \IF{X_{i}^{\parallel}\geq|S_{i-1}^{\perp}|} + 1 \bigr)
\]
and expanding the product. The second inequality is obtained using $e^{\lambda x}-1 = e^{\lambda x}(1-e^{-\lambda x})\leq e^{\lambda x}\lambda x$. Now, we use the Markov property and the local limit theorem in dimension $d-1$ to get that, for all $j$,
\begin{align*}
	\ebf_{u}\Bigl[ e^{\lambda X_{t_j}^{\parallel}} X_{t_j}^{\parallel}\; \IF{X_{t_j}^{\parallel}\geq |S_{t_j-1}^{\perp}|} \Bigm| S_{t_{j-1}}, X_{t_j} \Bigr]
	&=
	X_{t_j}^{\parallel} e^{\lambda X_{t_j}^{\parallel}} \pbf_{S_{t_{j-1}}}\bigl(|S_{t_j-1}^{\perp}|\leq X_{t_j}^{\parallel} \bigr)\\
	&\leq
	X_{t_j}^{\parallel} e^{\lambda X_{t_j}^{\parallel}} \C(X_{t_j}^{\parallel})^{d-1} \ell_j^{-(d-1)/2}.
\end{align*}
As $\pbf_u(X_i^{\parallel}>a) \leq e^{-c_p a}$, we obtain
\[
	\ebf_{u}\Bigl[ e^{\lambda X_{t_j}^{\parallel}} X_{t_j}^{\parallel}\; \IF{X_{t_j}^{\parallel}\geq |S_{t_j-1}^{\perp}|} \Bigm| S_{t_{j-1}} \Bigr]
	\leq
	\C\, \ell_j^{-(d-1)/2} \sum_{a\geq 1} e^{-c_p a} a^d e^{\lambda a}
	\equiv
	\Cl{UB1}\, \ell_j^{-(d-1)/2},
\]
with $\Cr{UB1}<\infty$ provided that $\lambda<c_p$. Therefore, defining
\begin{equation}
\label{eq:Am}
	A(m) = \sum_{k=1}^{m} (\Cr{UB1}\lambda)^k \sum_{\substack{\ell_1,\ldots,\ell_k\\\sum \ell_j = m}} \prod_{j=1}^{k} \ell_j^{-(d-1)/2},
\end{equation}
we get $A_{u,v}(m)\leq \C\, n^{(d-1)/2}A(m)$.

\medskip
In dimension $4$ and larger, we bound $A(m)$ uniformly over $m$ by ignoring the constraint $\sum \ell_j = m$:
\[
	A(m) \leq \sum_{k=1}^{\infty} \Bigl(\Cr{UB1}\, \lambda \sum_{\ell=1}^{\infty} \ell^{-(d-1)/2} \Bigr)^k ,
\]
which is convergent for $\lambda>0$ small enough. Using~\eqref{eq:m0bnd} with $m_0=1$, this implies that
\begin{align*}
	\e\bigl[ e^{\lambda |\Line\cap C_0|} \bigm| 0\leftrightarrow n \eone\bigr]
	\leq
	\C\, n^{(d+1)/2},
\end{align*}
which in turn yields $\xi\leq\iclx$. Since, \(\xi\geq\iclx\) always holds, we conclude that $\iclx=\xi$ for $\lambda>0$ small enough, and thus that $x'_c>x$ in dimension $4$ and larger. This proves the \(d\geq 4\) part of Theorem~\ref{thm:Jc_Potts}.

\medskip
In dimension $2$ and $3$, we get a diverging (in $m$) upper bound on $A(m)$. Consider the generating function associated to the sequence $(A(m))_{m\geq 1}$ and define $\mathbb{B}(z)$:
\[
	\mathbb{A}(z) = \sum_{m=1}^{\infty} A(m)z^m,
	\qquad
	\mathbb{B}(z) = \Cr{UB1}\lambda\sum_{\ell=1}^{\infty} \ell^{-(d-1)/2}z^\ell.
\]
Using~\eqref{eq:Am}, we have the relation
\[
	\mathbb{A}(z) = \sum_{k\geq 1} \mathbb{B}(z)^k .
\]
Note that \(\mathbb{B}\) is increasing on \(\mathbb{R}_+\).
Let $f(\lambda)>0$ be the unique number such that $\mathbb{B}(e^{-f(\lambda)})=1$.
Since $\mathbb{A}(e^{-2f(\lambda)})<\infty$, we conclude that $A(m)\leq e^{2f(\lambda)m}$ for all large enough $m$.
Now, using~\eqref{eq:m0bnd} with $m_0=\Cl{UB2}\, n$, $\Cr{UB2}>0$ small enough, and taking $\lambda$ sufficiently small, we obtain, when \(n\) is large,
\[
	\e\bigl[e^{\lambda |\Line\cap C_0|} \bigm| 0\leftrightarrow n \eone\bigr]
	\leq
	\C (1+O(n^{\frac{d-1}{2}}) e^{2f(\lambda)n}).
\]
It then follows from Theorem A.2 in \cite{Giacomin-2007} that, as $\lambda\downarrow 0$, $f(\lambda)$ behaves as
\[
	f(\lambda)
	=
	\begin{cases}
		\C \lambda^2 (1+o(1))					& \text{when } d=2,\\
		\exp\bigl(-\C/\lambda (1+o(1)) \bigr) 	& \text{when }d=3.
	\end{cases}
\]
This proves the upper bounds in Items~\ref{item:Critical2d} and~\ref{item:Critical3d} of Theorem~\ref{thm:PinnedRegime}.


\section{Lower bound on \(\xi_{\xpc} - \iclx\) when \(d=2,3\)}
\label{sec:LB}

In this section, we prove the lower bounds in Items~\ref{item:Critical2d} and~\ref{item:Critical3d} of Theorem~\ref{thm:PinnedRegime}, which will then also imply the \(d\in\{2,3\}\) part of Theorem~\ref{thm:Jc_Potts}.

For technical reasons, we work with large but finite systems (see Remark~\ref{rem:finVol}). The proof is based on an energy-entropy argument induced by
\[
	\frac{\p_{x'}(0\leftrightarrow n \eone)}{\p(0\leftrightarrow n \eone)}
	\geq
	\frac{\p_{x'}(M_{\delta} )}{\p(0\leftrightarrow n \eone)}
	\geq
	\frac{\p_{x'}(M_{\delta}) }{\p(M_{\delta})}\p(M_{\delta}|0\leftrightarrow n \eone) ,
\]
where $\delta\in(0,1)$ is arbitrary and $M_{\delta}$ is the event \{there exists an open path $\gamma\in\Gamma_{\delta}$\} with $\Gamma_{\delta}$ the set of self-avoiding paths from $0$ to $n\eone$ with at least $\delta n$ cone-points on $\Line_{[0,n]}$ (cone-points for the path itself, not the cluster of $0$). The analysis below applies to arbitrary values of the parameter \(\delta\). A specific choice will be made at the end of the section.
\begin{lemma}
	\label{lem:Russo}
	Let $A$ be an increasing event and take $x'>x$. Then,
	\begin{align*}
		\frac{\p_{x'}(A)}{\p(A)}
		\geq
		\exp\Bigl(\int_{x}^{x'} \frac{1}{s(1+s)} \sum_{e\in\Line_{[0,n]}} \p_s(e\in\textnormal{Piv}_A \given A) \,\dd s \Bigr).
	\end{align*}
\end{lemma}
\begin{proof}
	This is a straightforward application of Lemma~\ref{lem:RussoGen} and of the following consequence of the FKG inequality: $\p_{x'}(A)\geq\p_{x'}^{(n)}(A)$ where $\p_{x'}^{(n)}$ denotes the random-cluster measure with probabilities modified to \(x'\) only on $\Line_{[0,n]}$.
\end{proof}

This inequality allows us to control the ``energy'' part.
\begin{lemma}\label{lem:Ene}
	There exists $\rho>0$, depending on $x$, such that
	\[
		\frac{\p_{x'}(M_{\delta})}{\p(M_{\delta})}
		\geq
		\exp\Bigl( \frac{(x'-x)}{x'(1+x')}\rho\delta n \Bigr).
	\]
\end{lemma}
\begin{proof}
	Notice that $M_{\delta}$ is increasing. We can thus use Lemma~\ref{lem:Russo}:
	\begin{align*}
		\frac{\p_{x'}(M_{\delta}) }{\p(M_{\delta})}
		&\geq
		\exp\Bigl(\int_{x}^{x'} \frac{1}{s(1+s)} \sum_{e\in\Line_{[0,n]}} \p_s(e\in\textnormal{Piv}_{M_{\delta}} \given M_{\delta}) \,\dd s \Bigr)\\
		&\geq
		\exp\Bigl(\frac{1}{x'(1+x')} \int_{x}^{x'} \sum_{e\in\Line_{[0,n]}} \p_s(e\in\textnormal{Piv}_{0\leftrightarrow n\eone} \given M_{\delta})\, \dd s \Bigr),
	\end{align*}
	the last inequality following from the fact that $\textnormal{Piv}_{0\leftrightarrow n\eone}\subset\textnormal{Piv}_{M_{\delta}}$ on \(M_\delta\).
	The claim will thus follow if we can prove that 
	\[
	\sum_{e\in\Line_{[0,n]}} \p_s(e\in\textnormal{Piv}_{0\leftrightarrow n\eone} \given M_{\delta})
	\geq
	\rho\delta n
	\] 
	for some $\rho>0$, uniformly in $s\in[x,x']$.
	Fix an arbitrary total order on $\Gamma_{\delta}$.
	For $\gamma\in\Gamma_{\delta}$, define $B_{\gamma}=\{\gamma \textnormal{ is the first path in \(\Gamma_\delta\) realizing }M_{\delta}\}$. Then,
	\begin{align*}
		\sum_{e\in\Line_{[0,n]}} \p_s(e\in\textnormal{Piv}_{0\leftrightarrow n\eone} \given M_{\delta})
		&=
		\sum_{e\in\Line_{[0,n]}} \sum_{\gamma\in\Gamma_{\delta}} \p_s(B_{\gamma} \given M_{\delta})\, \p_s(e\in\textnormal{Piv}_{0\leftrightarrow n\eone} \given B_{\gamma})\\
		&\geq
		\sum_{\gamma\in\Gamma_{\delta}} \sum_{e\in \mathcal{C}(\gamma)} \p_s(B_{\gamma} \given M_{\delta}) \frac{\p_s(e\in\textnormal{Piv}_{0\leftrightarrow n\eone},B_{\gamma} \given \gamma \textnormal{ open})}{\p_s(B_{\gamma} \given \gamma \textnormal{ open})}\\
		&\geq
		\sum_{\gamma\in\Gamma_{\delta}} \p_s(B_{\gamma} \given M_{\delta}) \sum_{e\in \mathcal{C}(\gamma)} \p_s(e\in\textnormal{Piv}_{0\leftrightarrow n\eone} \given \gamma \textnormal{ open}) ,
	\end{align*}
	where $\mathcal{C}(\gamma)$ is the set of edges in $\gamma\cap\Line_{[0,n]}$ having a cone-point of $\gamma$ as an endpoint. The last inequality uses FKG: on the one hand, the measure $\p_s(\cdot \given \gamma \textnormal{ open})$ is a random-cluster measure on the complement of \(\gamma\) with wired boundary condition on \(\gamma\), and is thus positively associated; on the other hand, $B_{\gamma}$ and $\{e\in\textnormal{Piv}_{0\leftrightarrow n\eone}\}$, for $e\in\gamma$, are positively correlated as they are decreasing events on configurations in which the edges of $\gamma$ are open.
	
	We are thus left with showing that
	$\p_s(e\in\textnormal{Piv}_{0\leftrightarrow n\eone} \given \gamma \textnormal{ open}) \geq \rho$
	uniformly over $\gamma\in\Gamma_{\delta}$, $e\in\mathcal{C}(\gamma)$ and \(s\in[x,x']\).
	
	Fix $K\geq K_0$ large enough.
	Consider the cone $\fcone_{\psi}$, as introduced in Section~\ref{sec:RWrep}.
	Write $\cone_{\psi}(u)=u+(\bcone_{\psi}\cup\fcone_{\psi})$.

	Let $\Delta_K=[-K,K]^d$ (with a slight abuse of notation, $\Delta_K$ will denote the sites and/or the edges of $\Delta_K\cap \Zd$) and $u_{\psi}^{*}=\partial(u+\cone^{*}_{\psi})$ for $*\in\{\blacktriangleleft,\blacktriangleright\}$ and $u\in\Zd$. Write $v$ for the end-point of $e$ which is a cone-point for $\gamma$ (the right one if both are). Then, writing \(\tilde{\Delta}_K(v) = (v+\Delta_K) \setminus \cone_{\psi}(v)\),
	\begin{multline}
		\label{eq:ePiv}
		\p_s(e\in\textnormal{Piv}_{0\leftrightarrow n\eone} \given \gamma \textnormal{ open})
		\geq
		\p_s(e\in\textnormal{Piv}_{0\leftrightarrow n\eone}, \tilde{\Delta}_K(v) \textnormal{ closed} \given \gamma \textnormal{ open})\\
		\geq
		\rho_K\bigl( 1-\p_s(v^\blacktriangleleft_{\psi}\xleftrightarrow{\cone_{\psi}(v)^\comp} v^\blacktriangleright_{\psi} \given \gamma \textnormal{ open}, \tilde{\Delta}_K(v) \textnormal{ closed}) \bigr)
	\end{multline}
	by finite energy ($\rho_K$ depends polynomially on $K$).
	
	\begin{figure}[h]
	 	\centering
	 	\resizebox{4.5cm}{!}{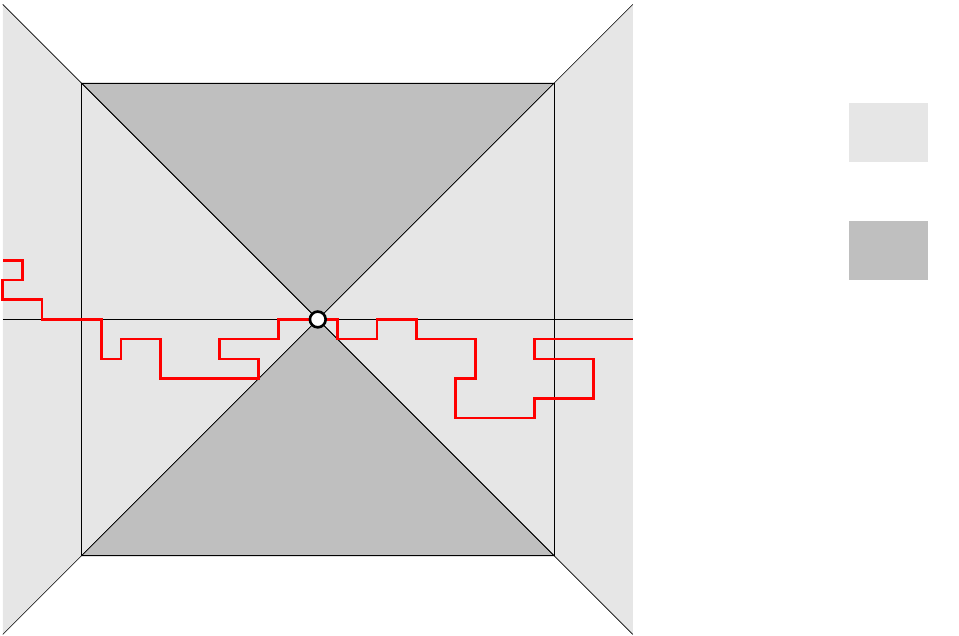}
	 	\hspace{2cm}
	 	\resizebox{2.95cm}{!}{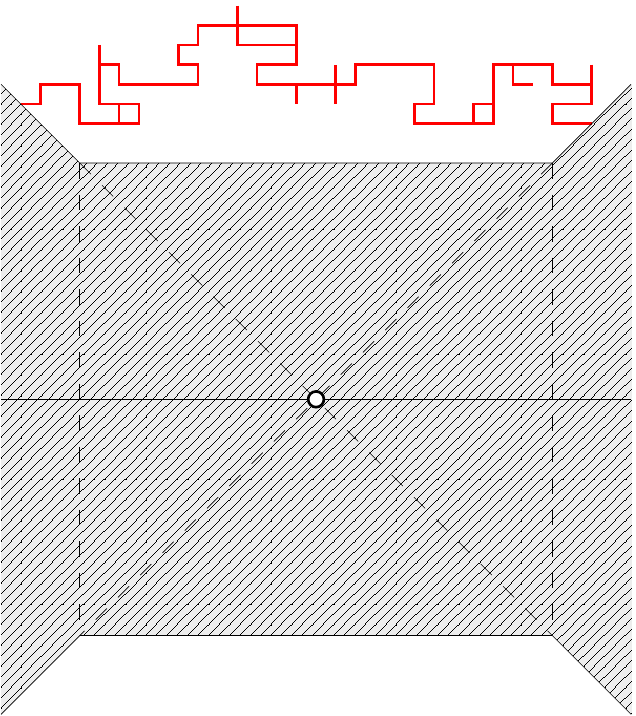}
	 	\caption{Left: The set $\cone_{\psi}(v)$. Right: The type of connections we want to prevent.}
	\end{figure}
	
	Now,
	\begin{align*}
		\p_s(v^\blacktriangleright_{\psi}&\xleftrightarrow{\cone_{\psi}(v)^\comp} v^\blacktriangleleft_{\psi} \given \gamma \textnormal{ open}, \tilde{\Delta}_K(v) \textnormal{ closed})\\
		&\leq
		\sum_{r,\ell=K}^{\infty}
		\sum_{\substack{u_L\in v_{\psi}^\blacktriangleright\\(u_L)_1=v_1-\ell}}
		\sum_{\substack{u_R\in v_{\psi}^\blacktriangleleft\\(u_R)_1=v_1+r}} \p_{\mathcal{Y}_{\psi}(v)^\comp}^{w}(u_L\xleftrightarrow{\cone_{\psi}(v)^\comp} u_R \given \tilde\Delta_K(v)\textnormal{ closed})\\
		&\leq
		\sum_{r,\ell=K}^{\infty} e^{-c_1(r+\ell)} \bigl\{2r\tan\bigl((\psi)/2)\bigr)^{d-2} (2\ell\tan\bigl((\psi)/2\bigr)\bigr\}^{d-2}\\
		&\leq
		e^{-2 c_1 K} \sum_{k=0}^{\infty} e^{-c_1 k} P_{K,d,\psi}(k)
		\leq
		C_K e^{-c_2K} ,
	\end{align*}
	where $P_{K,d,\psi}(k)$ is a polynomial in $k$ of degree at most $2d-3$ and
	$C_K$ is a constant depending polynomially on $K$; Lemma~\ref{lem:expoDecFKwired} was used to derive
	the second inequality. We take $K$ large enough for the right-hand side to be at most $1/2$. Plugging this into \eqref{eq:ePiv}, we get the desired result with $\rho=\frac{1}{2}\rho_K$ and thus the initial claim.
\end{proof}

Let us now consider the ``entropy'' term $\p(M_{\delta} \given 0\leftrightarrow n\eone)$.
\begin{lemma}
\label{lem:Ent}
	There exist $c_1,c_2>0$, depending on $x$, such that, for small enough $\delta>0$,
	\[
	\p(M_{\delta} \given 0\leftrightarrow n\eone)
	\geq
	\begin{cases}
		e^{-c_1 \delta^2 n} 				& \text{when } d=2,\\
		e^{-c_2 (\delta/\log(\delta)) n} 	& \text{when } d=3.
	\end{cases}
	\]
\end{lemma}
\begin{proof}
	Proceeding as in the previous section, we work with the measure \(\hat{\mathbb{Q}}\) and the random-walk $S=(S^{\perp},S^{\parallel})$ associated to $C_0$ (with increments $X_i=(X_i^{\perp},X_i^{\parallel})$).
	Notice that, every time $S$ steps on $\Line$, the corresponding point is a cone-point for \emph{any} open path in \(\Gamma_{\delta}\).
	Let $\mathcal{C}_{\delta}$ be the event $\{S$ hits $\Line$ at least $\delta n$ times$\}$.
	Define the sequence of hitting times of $\Line$: $\tau_0=0$ and $\tau_{k}=\inf\setof{m>\tau_{k-1}}{S_m\in\Line}$ for $k\geq 1$.
	Using~\eqref{eq:rhoLrhoRpositive}, we can restrict to the case where $\tilde\gamma^b$ and $\tilde\gamma^f$ are reduced to \(\{0\}\),
	respectively, and writing $\mathcal{R}_n=\mathcal{R}(n\eone)$, we get
	\[
		\p(\mathcal{M}_{\delta} \given 0\leftrightarrow n\eone)
		\geq
		c\, \mathbf{P}_{0}(\mathcal{C}_{\delta} \given \mathcal{R}_n) .
	\]
	From~\eqref{eq:LLT},
	we get that, for all sufficiently large $n$ and $k\leq n/2$, $\frac{\mathbf{P}_0(\mathcal{R}_{n-k})}{\mathbf{P}_0(\mathcal{R}_n)}\geq c$ for some $c>0$. Then, using the strong Markov property,
	\begin{align*}
		\mathbf{P}_{0}(\mathcal{C}_{\delta} \given \mathcal{R}_n)
		&=
		\mathbf{P}_{0}(S^{\parallel}_{\tau_{\delta n}}\leq n \given \mathcal{R}_n)\\
		&\geq
		\sum_{k=0}^{n/2} \mathbf{P}_{0}(S^{\parallel}_{\tau_{\delta n}}=k) \frac{\mathbf{P}_0(\mathcal{R}_{n-k})}{\mathbf{P}_0(\mathcal{R}_n)}\\
		&\geq
		c\,\mathbf{P}_{0}(S^{\parallel}_{\tau_{\delta n}}\leq n/2) .
	\end{align*}
	Now, denote by $\mathcal{N}_{n/2}^{\parallel}=\max\setof{k\leq n}{S_{k}^{\parallel}\leq n/2}$ the number of steps before exiting $[0,\frac{n}{2}\eone]\times\Z^{d-1}$ and by $L^{\perp}(\ell)=\#\setof{0\leq i\leq \ell}{S_i^{\perp}=0}$ the local time at $0$ of $S^{\perp}$ up to time $\ell$.
	Writing $\bar{n}=\frac{n}{4\mathbf{E}[X_1^{\parallel}]}$ and $\delta_*=\delta 4\mathbf{E}[X_1^{\parallel}]$,
	\begin{align*}
		\mathbf{P}_{0}(S^{\parallel}_{\tau_{\delta n}}\leq n/2)
		&\geq
		\mathbf{P}_{0}(S^{\parallel}_{\tau_{\delta n}}\leq n/2,\ \mathcal{N}_{n/2}^{\parallel}\geq \bar{n} )\\
		&\geq
		\mathbf{P}_{0}(L^{\perp}(\bar{n})\geq \delta n,\ \mathcal{N}_{n/2}^{\parallel}\geq \bar{n})\\
		&\geq
		\mathbf{P}_{0}(L^{\perp}(\bar{n})\geq \delta_* \bar{n}) -\mathbf{P}_{0}\Bigl(\sum_{k=1}^{\bar{n}}X_k^{\parallel}> n/2\Bigr)\\
		&\geq
		\mathbf{P}_{0}(L^{\perp}(\bar{n})\geq \delta_* \bar{n}) - e^{-cn} ,
	\end{align*}
	where we used an elementary large deviation estimate for a sum of independent random variables in the last line.
	Finally, the event $\{L^{\perp}(\bar{n})\geq \delta_* \bar{n}\}$ depends only on $S^{\perp}$ which is a random walk with i.i.d.\ increments in $\Z^{d-1}$, and thus (see Corollary B.3 in \cite{Friedli+Ioffe+Velenik-2013}):
	\[
	\mathbf{P}_{0}(L^{\perp}(\bar{n}) \geq \delta_* \bar{n})
	\geq
	\begin{cases}
		e^{-c\delta_*^2 n} 					& \text{when } d=2,\\
		e^{-c(\delta_*/|\log\delta_*|) n} 	& \text{when } d=3.
	\end{cases}
	\]
\end{proof}

Now, choosing $\delta$ to be
\[
	\delta
	=
	\begin{cases}
		C_{x}(x'-x)							& \text{when } d=2,\\
		\exp\bigl( -C'_{x}/(x'-x) \bigr)	& \text{when } d=3,
	\end{cases}
\]
with \(C_x,C'_x\) large enough (observe that \(x'(1+x')>x(1+x)/2\) when \(x'-x\) is sufficiently small),
in Lemmas~\ref{lem:Ene} and~\ref{lem:Ent}, we obtain the lower bounds stated in Items~\ref{item:Critical2d} and~\ref{item:Critical3d} of Theorem~\ref{thm:PinnedRegime}, which also implies the \(d\in\{2,3\}\) part of Theorem~\ref{thm:Jc_Potts}.


\section{Coarse-graining procedure and advanced properties}
\label{sec:PuExpDecAnalytStrDecr}

In this section, we prove the first and last items of Theorem~\ref{thm:PinnedRegime}, as well as Theorem~\ref{thm:PottsScalingInterface}. Namely, we show that, for $x'>\xpc$, the connectivity function along the $\eone$ axis has pure exponential decay and that $x'\mapsto\iclx$ is real analytic and strictly decreasing on $(\xpc,\infty)$.
This will be done with the help of a coarse-graining procedure similar to the one we used in Section~\ref{sec:BasicProp}.

\subsection{Coarse-graining}
We first describe our coarse-graining procedure. Fix a scale $K$ and a number $r$ (both to be chosen later, independent of $n$) and define
\begin{align*}
	\innerBox(v)
	&=
	\llbracket-K, K\rrbracket\times\llbracket-2K-3r\log(K),2K +3r\log(K) \rrbracket^{d-1}+v,\\
	\outerBox(v)
	&=
	\llbracket-\barK, \barK\rrbracket\times\llbracket-2K -4r\log(K),2K+4r\log(K) \rrbracket^{d-1}+v ,
\end{align*}
where $\barK=K+r\log(K)$.

\begin{figure}[h]
	\includegraphics[scale=0.45]{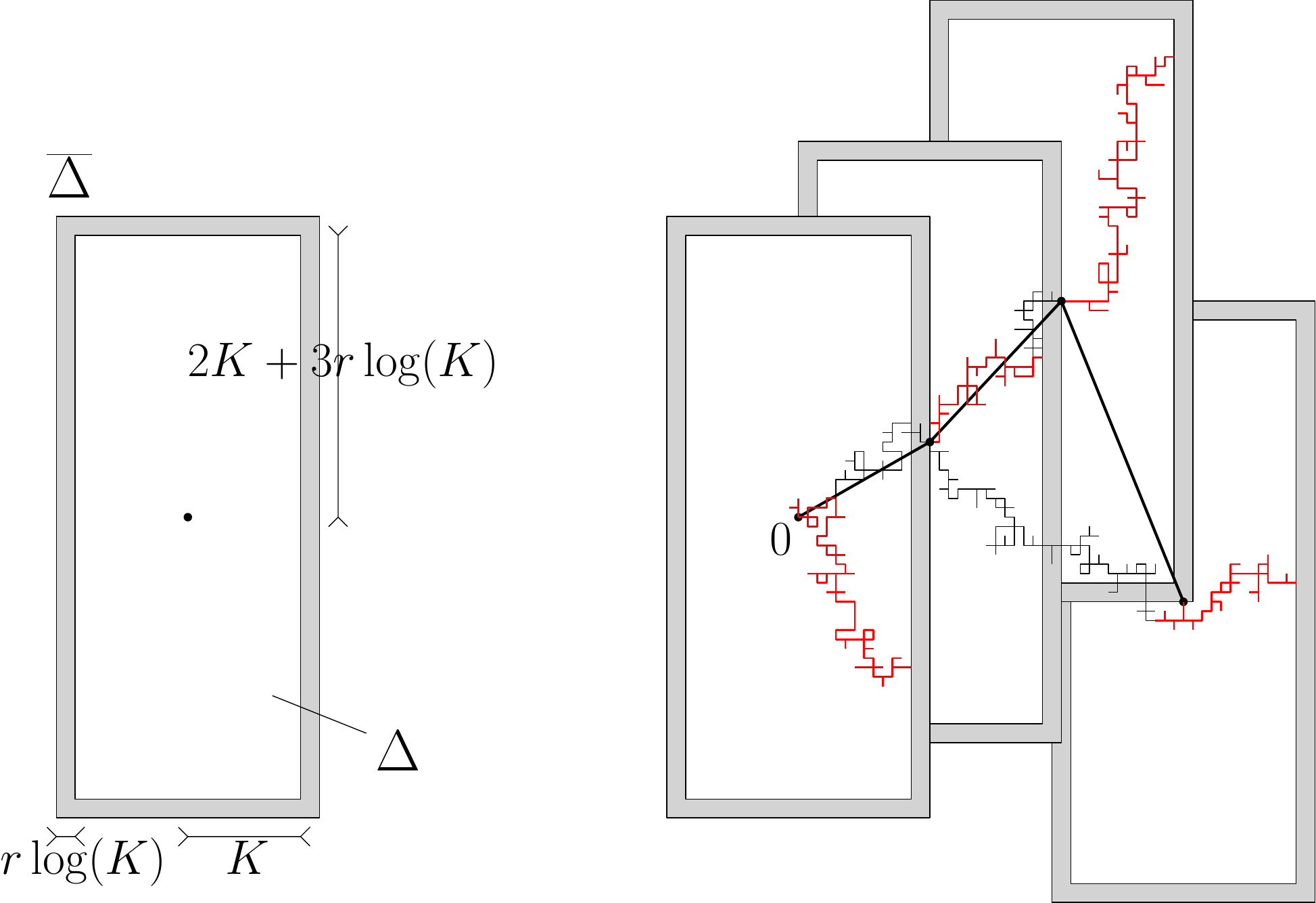}
	\caption{Elementary piece of the coarse-graining.}
	\label{fig:CG2}
\end{figure}

Given a set of vertices $A\subset\Zd$, define $\overline{A}=\bigcup_{u\in A}\outerBox(u)$.
Given \(B\subset\Zd\), we denote by \(C_0|_{B}\) the connected component of \(0\) in \(C_0\cap B\).
We coarse-grain $C_0$ using the following algorithm:\\

\begin{algorithm}[H]
	\label{alg:CoarseGrainingPED}
	$v_0 = 0$\;
	$V = \{v_0\}$ and $E=\emptyset$\;
	$i=0$\;
	$A=\setof{v\in\partial\overline{V}\cap C_0|_{\overline{V}}}{v
	\xleftrightarrow{\Delta(v)\setminus\overline{V}}
	\partial\Delta(v)}$\;
	\While{$A\neq\emptyset$}{
		$i=i+1$\;
		Set $v_{i}$ to be the smallest site of $A$ w.r.t. the lexicographical order\;
		Set $e_i=(v^*,v_i)$, where $v^*$ is the smallest vertex in $V$ among those closest to $v_i$ \;
		Update $V=V\cup\{v_i\}$ and $E=E\cup\{e_i\}$\;
		Update $A=\{v\in\partial \overline{V}\cap C_0|_{\overline{V}}:v
		\xleftrightarrow{\Delta(v)\setminus\overline{V}}
		\partial\Delta(v) \}$\;
	}
	$\tree_0=(V,E)$\;
	\caption{Coarse-graining procedure of $C_0$}
\end{algorithm}

\smallskip\noindent
Adapting the definitions introduced in the coarse-graining argument used in the proof of Item~\ref{it:monot} of Lemma~\ref{lem:BasicProp} to our new boxes, denote by $C_K$ the number of vertices in $\partial\outerBox$.
As already used there, the number of trees $\tree$ with $N$ vertices that can be obtained via Algorithm~\ref{alg:CoarseGrainingPED} is at most $e^{\C\log(C_K)N}$. We say that $v\in\Zd$ is \emph{$\Line$-free} if $\outerBox(v)\cap\Line=\emptyset$; otherwise, we call it an \emph{$\Line$-vertex}. The next lemma will give us control on the probability to see a specific tree $T$.
\begin{lemma}
	The probability of a given tree $T$ with $m$ $\Line$-free vertices and $N$ $\Line$-vertices satisfies
	\[
		\p_{x'}(\tree_0=T) \leq e^{-\xi K m} e^{-\iclx K N} (1+o_K(1))^{m+N}.
	\]
\end{lemma}
\begin{proof}
	First, notice that whenever an $\Line$-free vertex is created, a connection as described in Lemma~\ref{lem:LfreeEst} is induced forcing a cost $e^{-\xi K}$ uniformly over the tree constructed so far (we fix \(r\) large enough to be able to apply Lemma~\ref{lem:LfreeEst}).
	When an $\Line$-vertex is created, two things can happen. The first possibility is that a crossing (in the easy direction) of a box $\llbracket -K,K \rrbracket\times\llbracket0,K \rrbracket^{d-1}$ at distance at least $r\log(K)$ from the line $\Line$ is induced (see Figure~\ref{fig:CGboxCross}), costing $e^{-\xi K}(1+o_K(1))\leq e^{-\iclx K}(1+o_K(1))$ by Lemma~\ref{lem:LfreeEst}. The second possibility is that the vertex is connected to a side of $\innerBox$ crossed by $\Line$. The procedure used in the proof of Lemma~\ref{lem:LfreeEst} together with~\eqref{eq:decay_xprime} and~\eqref{eq:expDecUnif}
	yield probability $e^{-\iclx K}(1+o_K(1))$ for such a crossing (uniformly over the tree constructed so far).
	\begin{figure}[h]
		\includegraphics[scale=0.45]{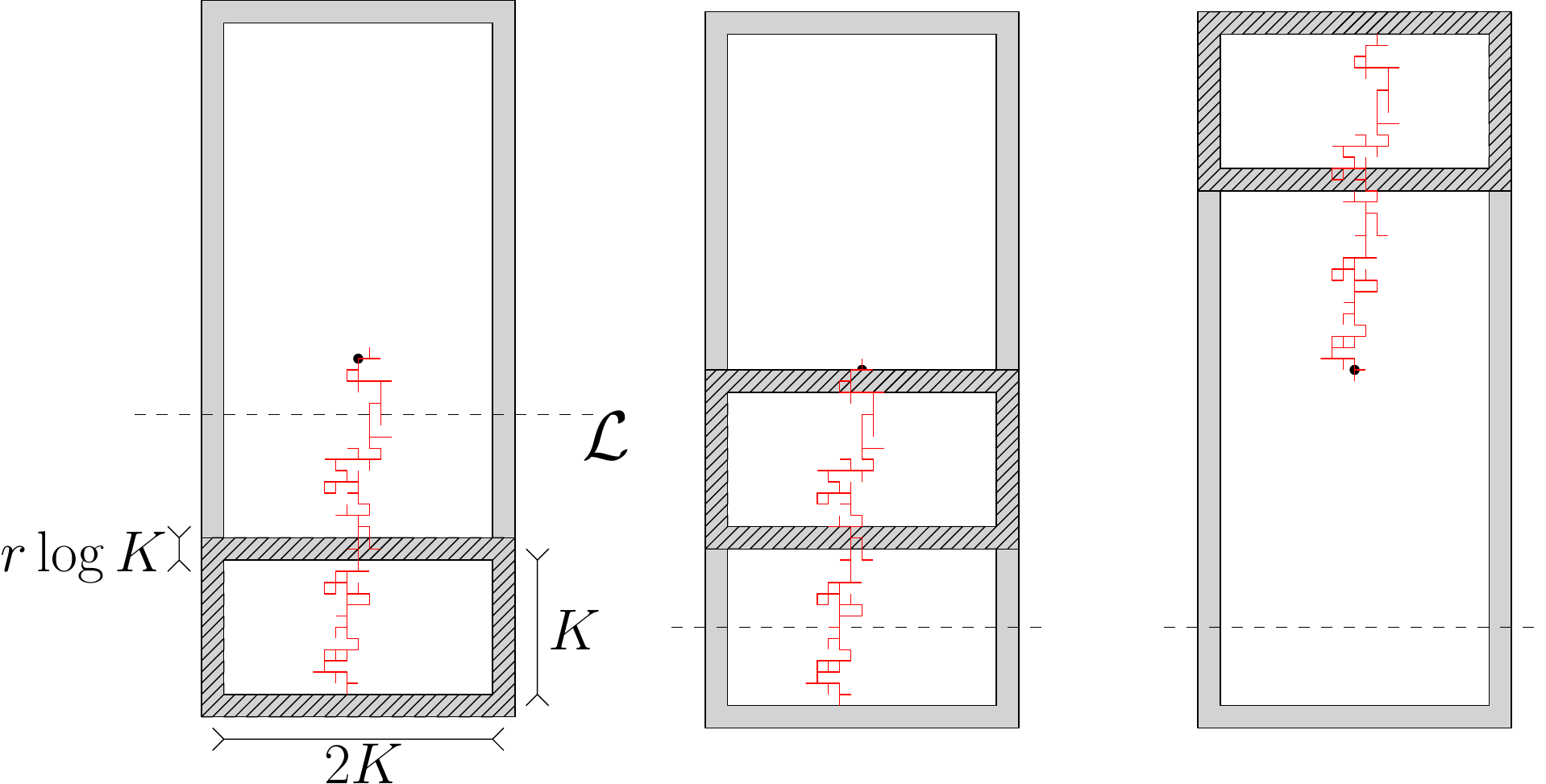}
		\caption{Connection to the boundary of $\innerBox$.}
		\label{fig:CGboxCross}
	\end{figure}
\end{proof}

Therefore, as $\tau=\xi-\iclx>0$ when $x'>\xpc$, we have (for $T$ containing $m$ $\Line$-free vertices and $N$ $\Line$-vertices):
\[
	\p_{x'}(\tree_0=T) \leq e^{-\tau K m} e^{-\iclx K (N+m)} (1+o_K(1))^{m+N} .
\]
Using this, we argue similarly as in the proof of Item~\ref{it:monot} of Lemma~\ref{lem:BasicProp}. Remark that, up to a term of order $e^{-2\iclx n(1+o(1))}$, we can restrict connections $0\leftrightarrow n\eone$ to those not connecting to $\Z_{<-n}\times\Z^{d-1}$ or to $\Z_{>2n}\times\Z^{d-1}$. Then, for any $\rho\in(0,1)$, we have (with $\Cl{BoxPerLpiece2}>0$ a constant depending on the dimension):
\begin{align*}
	\p_{x'}(\tree_0&\text{ contains at least }\rho n/\barK\ \Line\text{-free vert.}, 0\leftrightarrow n\eone)\\
	&\leq
	\sum_{m=\rho n/\barK}^{\infty} \sum_{N=1\vee n/\barK-m}^{\Cr{BoxPerLpiece2}n/\barK} \p_{x'}(\tree_0\text{ contains }m\ \Line\text{-free v. and } N\ \Line\text{-vert.})\\
	&\leq
	\sum_{m=\rho n/\barK}^{\infty} \sum_{N=1\vee n/\barK-m}^{\Cr{BoxPerLpiece2}n/\barK} e^{-\tau K m} e^{-\iclx K (N+m)} e^{c_d\log(K)(m+N)}\\
	&\leq
	e^{-\frac{K}{\barK} \iclx n} \frac{\C[csts]n}{\barK} e^{-(\tau\rho K/\barK-\C \log(K)/\barK)n}.
\end{align*}
Thus, for any $\rho\in(0,1)$, we can find $K_0\equiv K_0(\rho,\tau)$ such that, for $K\geq K_0$, there exists $\nu(\rho)>0$ depending on $\rho,x$ and $x'$ such that
\begin{equation}
\label{eq:BoundOnFreeVertices}
	\p_{x'}(\tree_0\text{ contains at least }\rho n/\barK\ \Line\text{-free vert.}, 0\leftrightarrow n\eone)
	\leq
	e^{-\frac{K}{\barK} \iclx n}e^{-\nu(\rho)n} .
\end{equation}

\subsection{Renewal on $\Line$}
\label{sec:renOnL}
We now use the coarse-graining of the previous section to show that, under \(\p_{x'}(\,\cdot\given 0 \leftrightarrow n\eone)\), $C_0$ possesses a number of cone-points on $\Line$ of order \(n\) when $x'>\xpc$ (as defined in Section~\ref{sec:RWrep}). For convenience, we look at cones $\cone$ (and diamonds) having angular aperture $\pi/2$.
\begin{theorem}
	\label{thm:CPdensity}
	When $x'>\xpc$, there exist $\rho_{\text{cp}}\equiv\rho_{\text{cp}}(x')\in(0,1)$ and $\Cl[expoThm]{cpDec}>0$ such that
	\[
	\p_{x'}\bigl(C_0\text{ contains at least }\rho_{\text{cp}}n\text{ cone-points on }\Line \bigm| 0\leftrightarrow n\eone\bigr)
	\geq
	1-e^{-\Cr{cpDec}n} .
	\]
\end{theorem}
\begin{proof}
Start by observing that $C_0$ is included in a $K$-neighborhood of $\tree_0$. Then, define the \emph{shade} $\shade(v)$ of a point $v$ by
\begin{align*}
	\shade(v)=\Line_{[-\norm{v^{\perp}}+v^{\parallel},\norm{v^{\perp}}+v^{\parallel}]}.
\end{align*}
This corresponds to the portion of $\Line$ that cannot contain cone-points of $C_0$ as soon as $v\in C_0$.
In the same fashion, define the shade of $v\in\tree_0$ to be the union over $u\in[\outerBox(v)]_K$ (the $K$ neighborhood of $\outerBox(v)$) of the shade of $u$.
Finally, define the shade of $\tree_0$ as the union of the shades of the $\Line$-free vertices of $\tree_0$.
We will show that, with high probability, this shade does not cover a substantial proportion of $\Line$; then we will use a finite-energy argument to show that a positive fraction of the unshaded points are cone-points of $C_0$.

A first observation is that there exists $\Cl{shadeCst}$ not depending on $K$, such that the size of the shade of $\tree_0$ is at most
\[
\Cr{shadeCst} K \#\{\Line\text{-free vertices of }\tree_0\}.
\]
This is proved by induction on the number of $\Line$-free vertices of $\tree_0$.
The first one is at distance at most $5K$ from the line and the inequality thus holds by definition of the shade.
Then, adding an $\Line$-free vertex either adds the same shade size to the total shade (the vertex is far from the existing ones) or it increases the shade size by at most \(\Cl{Cst:IncrShadeSize}K\), for some constant \(\Cr{Cst:IncrShadeSize}<10\), (see figure~\ref{fig:shade}).

\begin{figure}[h]
	\includegraphics[scale=0.7]{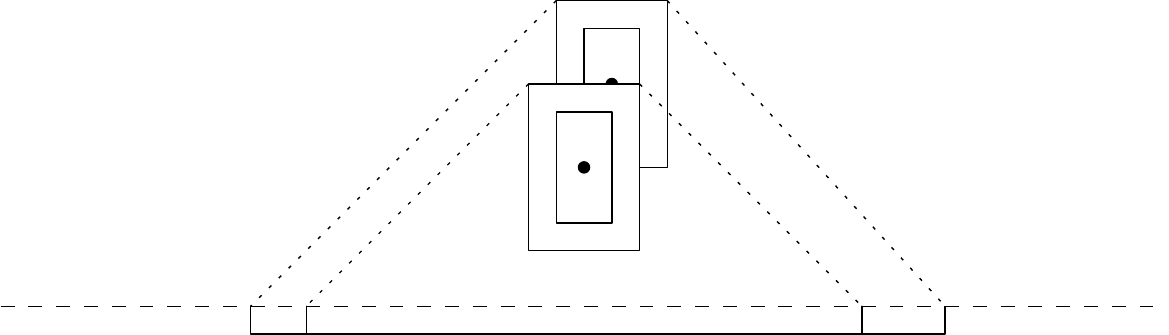}
	\caption{Evolution of $\tree_0$ shade.}
	\label{fig:shade}
\end{figure}

Now, we split $\llbracket 0,n \rrbracket\times \Z^{d-1}$ into slices. For \(i=1,2,\ldots,n/7K\), define
\[
S_{i}
=
\llbracket (i-1)7K,i7K \rrbracket\times \Z^{d-1}\quad\text{ and }\quad
B_i
=
S_{i}\cap[\Line]_{3K},
\]
where $[\Line]_{3K}$ is the $3K$-neighborhood of $\Line$. We will say that $B_i$ is \emph{illuminated} if $\Line_{[(i-1)7K+3K,i7K-3K]}$ is not included in the shade of $C_0\setminus[\Line]_{3K}$ (which is included in the shade of the $\Line$-free vertices of $\tree_0$). We have, using~\eqref{eq:BoundOnFreeVertices},
\begin{align*}
	\p_{x'}\bigl(\#\{\text{illuminated }B_i\}\geq \frac{n}{14K}\bigr)
	&\geq
	1 - \p_{x'}\bigl(|\shade(\tree_0)| > \frac{n}{2} \bigr)\\
	&\geq
	1 - \p_{x'}\bigl(\Cr{shadeCst}\#\{\Line\text{-free vertices of }\tree_0\} > \frac{n}{2} \bigr)\\
	&\geq
	1 - e^{-\frac{K}{\barK} \iclx n}e^{-\C[expoThm]n}.
\end{align*}
Thus, as $\p_{x'}(0\leftrightarrow n\eone)\geq e^{-\iclx n(1+o(1))}$,
\begin{equation}
	\label{eq:posDensLight}
	\p_{x'}\bigl(\#\{\text{illuminated }B_i\}\geq \frac{n}{14K} \bigm| 0\leftrightarrow n\eone \bigr)
	\geq
	1 - e^{-\Cl[expoThm]{nu-huit}n},
\end{equation}
for some \(\Cr{nu-huit}>0\).
Noticing that the number of $B_i$ is $\frac{n}{7K}$, this implies that at least half the boxes are illuminated with high probability.

Now, we describe a surgery procedure creating a cone-point on $\Line$ from an illuminated $B_i$ and bound its cost uniformly over the rest of the cluster of $0$. Given the restriction of a configuration $\omega|_{[\Line]_{3K}^\comp}$ outside $[\Line]_{3K}$ such that $B_i$ is illuminated, let $v_L$ and $v_R$ be the leftmost and rightmost vertices of $\Line_{[(i-1)7K+3K,i7K-3K]}$ not shaded by $C_0\cap[\Line]_{3K}^\comp$. Notice that, by definition of the shade, the whole segment $[v_L,v_R]$ is not in the shade of $C_0\cap[\Line]_{3K}^\comp$. Denote by $A_{v_L,v_R}$ the event that all edges of $[v_L,v_R]$ are open, all edges inside $(v_L-\cone)\cap B_i$ and $(v_R+\cone)\cap B_i$ are open and the remaining edges of $B_i$ are closed (see Figure~\ref{fig:locSurgery}).

\begin{figure}[h]
	\includegraphics[scale=0.7]{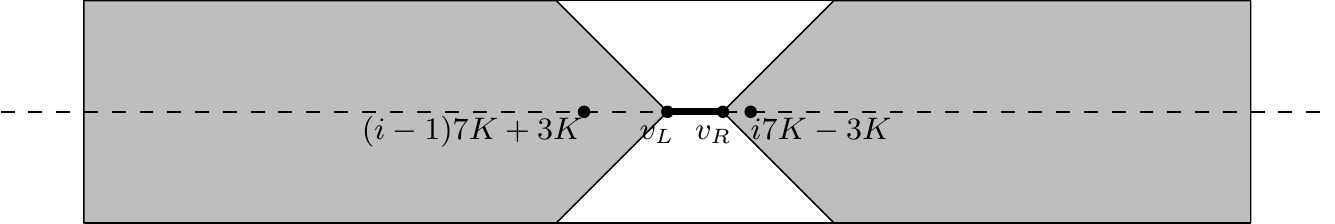}
	\caption{Local surgery procedure to create cone-points. The edges between $v_L$ and $v_R$ and those in the shaded regions are all open, those in the white regions are all closed.}
	\label{fig:locSurgery}
\end{figure}

By finite energy, $\min_{\omega|_{[\Line]_{3K}^\comp}}\p_{x'}\left(A_{v_L,v_R} \given 0\leftrightarrow n\eone\right)\geq \theta >0$ where the minimum is taken over configurations such that $B_i$ is illuminated. Thus, a positive density of illuminated $B_i$'s contain cone-points and, by~\eqref{eq:posDensLight}, a positive density of $B_i$'s are illuminated, thereby completing the proof of Theorem \ref{thm:CPdensity}.
\end{proof}

\subsection{Pure exponential decay when \(x'>\xpc\)}
\label{sec:pureExpDec}
We are now in position to prove the last item of Theorem~\ref{thm:PinnedRegime}. This will be done in the same fashion as in Section~\ref{sec:RWrep}, except that the ``random walk'' will here be pinned to the line, replacing the power-law correction present in~\eqref{eq:OZzzz} by a constant (which is related to the frequency of occurrence of cone-points on the line). We work here with cones of angular aperture $\pi/2$. As in Section~\ref{sec:RWrep}, let $w_1,\ldots,w_m$ be the cone-points of $C_0$ lying on $\Line$ (by Theorem~\ref{thm:CPdensity}, $m$ is typically of order $n$). Let 
\[
\zeta_i = C_0\cap\llbracket w_i\cdot\eone,w_{i+1}\cdot\eone\rrbracket\times\Z^{d-1}
\]
define the cone-confined irreducible components of $C_0$, and let $\zeta^{\bawa}$ and $\zeta^{\fowa}$ be the two components of $C_0\setminus(\zeta_1\cup\zeta_2\cup...\cup\zeta_{m-1})$ containing respectively $0$ (backward-irreducible) and $n\eone$ (forward-irreducible); they can possibly be reduced to a single vertex.
All definitions of Section~\ref{sec:RWrep} extend with almost no modification to the irreducible components $\zeta$. In particular, we can define percolation events $\Xi^{\bawa},\Xi_1,\ldots,\Xi_{m-1},\Xi^{\fowa}$ associated with $\zeta^{\bawa},\zeta_1,\ldots,\zeta_{m-1},\zeta^{\fowa}$ so that
\begin{multline*}
	\p_{x'}(C_0= \zeta^{\bawa}\sqcup\zeta_1\sqcup\cdots\sqcup\zeta_{m-1}\sqcup\zeta^{\fowa})\\
	=
	\p_{x'}(\Xi^{\bawa}) \p_{x'}(\Xi^{\fowa} \given \Xi^{\bawa},\ldots,\Xi_{m-1}) \prod_{i=1}^{m-1} \p_{x'}(\Xi_i \given \Xi^{\bawa},\ldots,\Xi_{i-1}).
\end{multline*}
Then, for $u,v\geq 1$, we can define
\[
	\rho'_{\bawa}(u)
	=
	e^{\iclx u} \sum_{\substack{\zeta^{\bawa}\ni 0\\D(\zeta^{\bawa})=u}} \p_{x'}(\Xi^{\bawa})
	\quad\text{ and }\quad
   \rho'_{\fowa}(v)
   =
   e^{\iclx v} \sum_{\substack{\zeta^{\fowa}\ni 0\\D(\zeta^{\fowa})=v}} \p_{x'}(\Xi^{\fowa}).
\]
By Theorem~\ref{thm:CPdensity}, they satisfy
\[
	\rho'_{\bawa}(u) \leq e^{-\Cr{cpDec}u}
	\quad\text{ and }\quad
	\rho'_{\fowa}(v)\leq e^{- \Cr{cpDec}v}.
\]
Again, all the properties listed in Proposition~\ref{pro:PropertiesOK} hold in the present setting (with essentially the same proof).
This allows us to proceed as in Section~\ref{sec:RWrep} in order to ``couple'' \(C_0\) with a random walk $S'$ on $\Z_{>0}$ with i.i.d.\ increments $X_i'$ in $\Z_{>0}$ having exponential tails. We denote its law and expectation by $\mathbf{Q}'$. The measures associated to the boundaries pieces will be denoted by $\hat\rho'_{\bawa},\hat\rho'_{\fowa}$; they have exponential tails.

\begin{figure}[h]
	\includegraphics[scale=0.7]{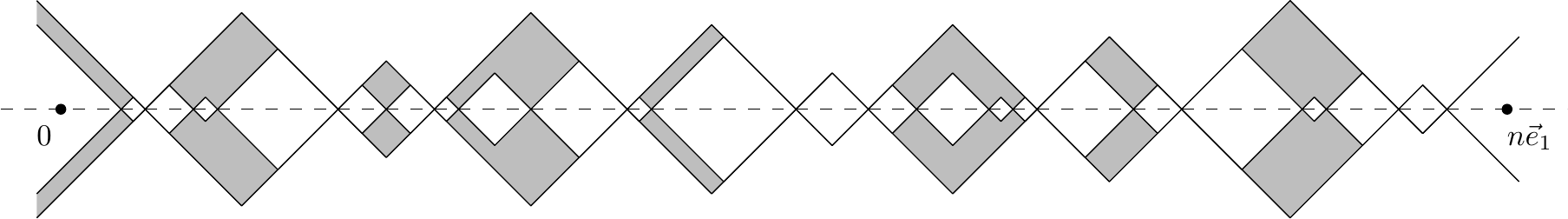}
	\caption{The original process based on the cone-points of \(C_0\) on the line $\Line$, together with the associated process of independent pieces.}
	\label{fig:CG2cp}
\end{figure}

The arguments leading to~\eqref{eq:OZzzz} yield in the present setting
\[
	e^{\iclx n} \p_{x'}(0\leftrightarrow n\eone)
	=
	\sum_{u,v} \hat\rho'_{\bawa}(u) \hat\rho'_{\fowa}(v) \, \mathbf{Q}'\Bigl(\exists t>0:\, \sum_{i=1}^{t} X_i'=n-v-u \Bigr) + e^{-cn}.
\]
We can clearly restrict the sum to $u,v<n/4$. The conclusion then follows from the Renewal Theorem and Theorem~\ref{thm:CPdensity}, since they imply that
\[
	\mathbf{Q}'\Bigl(\exists t>0:\, \sum_{i=1}^{t} X_i'=n-v-u \Bigr)
	\xrightarrow{n\to\infty}
	\frac{1}{\mathbf{Q}'(X_1')} > 0,
\]
uniformly in $u,v<n/4$.

\subsection{$\iclx$ is strictly decreasing when \(x'>\xpc\)}
As discussed in Remark~\ref{rem:finVol}, we can find a sequence $(a_n)_{n\geq 1}$ of large enough numbers such that $\iclx=\lim_{n\to\infty}\iclx^{(n)}$, where
\[
\iclx^{(n)} = -\frac{1}{n}\log \p_{x',\Lambda_{a_n}}^{f}(0\leftrightarrow n\eone).
\]
We can then bound $\frac{\dd}{\dd x'} \iclx^{(n)}$ using Lemma~\ref{lem:RussoGen}:
\[
	\frac{\dd}{\dd x'}\iclx^{(n)}
	\leq
	-\frac{1}{n} \frac{1}{x'(x'+1)} \sum_{e\in\Line} \p_{x',\Lambda_{a_n}}^{f}(e\in\text{Piv}_{0\leftrightarrow n\eone} \given 0\leftrightarrow n\eone).
\]
By Theorem~\ref{thm:CPdensity}, the number of cone-points can be assumed to grow linearly with $n$. As every cone-point induces at least one pivotal edge for $0\leftrightarrow n\eone$, we can find a positive constant such that $\frac{\dd}{\dd x'}\iclx^{(n)} \leq -\C$ uniformly in $n$. Thus, $x'\mapsto\iclx$ is strictly decreasing. Indeed, for $\xpc<x'_1<x'_2<\infty$,
\[
\xi_{x'_2}-\xi_{x'_1}
=
\lim_{n\to\infty} \xi_{x'_2}^{(n)} - \xi_{x'_1}^{(n)}
=
\lim_{n\to\infty} \int_{x'_1}^{x'_2} \frac{\dd}{\dd s} \xi_s^{(n)} \dd s
\leq
-\C(x'_2-x'_1) .
\]
Notice that the constant depends on $x'_2$.

\subsection{Analyticity of $x'\mapsto\iclx$ for \(x'>\xpc\)}

For any \(x'_0>\xpc\), we are going to prove analyticity of $\iclx$ for \(x'\) in a neighbourhood of \(x'_0\). Let us thus fix \(x'_0>\xpc\). We first make the following two assumptions, which will be proved at the end of the section:
\begin{claim}
	\label{claim:analyticityLineToInterval}
		$\iclx$ can be obtained as the limit of $-\frac{1}{n}\log\p_{x'}^{(n)}(0\leftrightarrow n\eone)$ (where $\p_{x'}^{(n)}$ denotes the measure with only $\Line_{[0,n]}$ weights modified).
\end{claim}

\begin{claim}
	\label{claim:analyticityFreeEnergy}
		\(
		f_{x'}=\lim_{n\to\infty}\frac{1}{n}\log\e_{x'_0}^{(n)}\bigl[\bigl(\tfrac{x'}{x'_0}\bigr)^{\ouv_{\Line_{[0,n]}}}\bigr]
		\)
		exists and is analytic in \(x'\) in a small neighbourhood of \(x_0'\).
\end{claim}

Assuming this, we can rewrite
\begin{align*}
\lim_{n\to\infty} \frac{1}{n} \log \p_{x'}^{(n)}(0\leftrightarrow n\eone)
&=
\lim_{n\to\infty} \frac{1}{n} \log \frac{\e_{x'_0}^{(n)}\bigl[\bigl(\frac{x'}{x'_0}\bigr)^{\ouv_{\Line_{[0,n]}}}\IF{0\leftrightarrow n\eone}\bigr]}{\e_{x'_0}^{(n)}\bigl[\bigl(\frac{x'}{x'_0}\bigr)^{\ouv_{\Line_{[0,n]}}}\bigr]} \\
&=
\lim_{n\to\infty} \frac{1}{n} \log\bigl\{e^{-(f_{x'}+\xi_{x'_0})n}\e_{x'_0}^{(n)}\bigl[\bigl(\tfrac{x'}{x'_0}\bigr)^{\ouv_{\Line_{[0,n]}}} \bigm\vert 0\leftrightarrow n\eone\bigr]\bigr\} .
\end{align*}
The same construction as in the previous subsection (coarse-graining and finite energy), together with the strict monotonicity of \(\iclx\) on \((\xpc,\infty)\), guarantee that there exists $\epsilon_0>0$ such that, for any \(x'\) in a neighbourhood of \(x'_0\), we have
\begin{align*}
e^{-(f_{x'}+\xi_{x'_0})n}& \e_{x'_0}^{(n)}\bigl[\bigl(\tfrac{x'}{x'_0}\bigr)^{\ouv_{\Line_{[0,n]}}} \bigm\vert 0\leftrightarrow n\eone\bigr]\\
&=
e^{-(f_{x'}+\xi_{x'_0})n}  \e_{x'_0}^{(n)}\bigl[\bigl(\tfrac{x'}{x'_0}\bigr)^{\ouv_{\Line_{[0,n]}}}\IF{\text{Cp}(\Line_{[0,n]})>\epsilon_0n} \bigm\vert 0\leftrightarrow n\eone\bigr] (1+o(1)),
\end{align*}
where $\text{Cp}(\Line_{[0,n]})$ is the number of cone-points for the cluster of $\Line_{[0,n]}$ (the union of all the clusters containing at least one vertex of $\Line_{[0,n]}$) lying on $\Line_{[0,n]}$. Let $\vartheta^b,\vartheta_1,\ldots,\vartheta_m,\vartheta^f$ be the (random) sequence of diamonds-confined $\Line_{[0,n]}$-clusters (that is, the irreducible, in the sense of Section~\ref{sec:pureExpDec}, cone-confined pieces of the cluster of $\Line_{[0,n]}$). We write $D_1,\ldots,D_m$ the sequence of diamonds containing $\vartheta_1,\ldots,\vartheta_m$ and $|D_i|\geq 1$ the $\eone$ displacement (or length) of $D_i$. We stress at this point that $\vartheta_i$ contains the whole information about all the clusters touching $\Line\cap D_i$. We also denote \(D^b,D^f\) the diamonds containing $\vartheta^b,\vartheta^f$ (their left (resp.\ right) endpoints might not lie on $\Line_{[0,n]}$). As before, the length of these irreducible components has exponential tails.
We obtain
\begin{multline}
\label{eq:diamondExpect1}
\e_{x'_0}^{(n)}\bigl[\bigl(\tfrac{x'}{x'_0}\bigr)^{\ouv_{\Line_{[0,n]}}}  \IF{\text{Cp}(\Line_{[0,n]})>\epsilon_0 n}\bigm\vert 0\leftrightarrow n\eone \bigr]=\\
\e_{x'_0}^{(n)}\Bigl[\bigl(\tfrac{x'}{x'_0}\bigr)^{\ouv_{D^b\cap\Line_{[0,n]}}}\bigl(\tfrac{x'}{x'_0}\bigr)^{\ouv_{D^f\cap\Line_{[0,n]}}}\prod_{i=1}^{m}\bigl(\tfrac{x'}{x'_0}\bigr)^{\ouv_{D_i\cap\Line}} \IF{\text{Cp}(\Line_{[0,n]})>\epsilon_0 n}\bigm\vert 0\leftrightarrow n\eone\bigr].
\end{multline}
Proceeding as in Sections~\ref{sec:RWrep} and~\ref{sec:pureExpDec},
we can partition $\vartheta^b,\vartheta_1,\ldots,\vartheta_m,\vartheta^f$ into finite strings of irreducible pieces \(\tilde\vartheta^b,\tilde\vartheta_1,\dots,\tilde\vartheta_{k-1},\tilde\vartheta^f\), and construct a probability measure 
\(\mathsf{Q}\) on the irreducible components $\tilde{\vartheta}_i$ and two finite measures $\mathsf{p}_L,\mathsf{p}_R$ on \(\tilde\theta^b\) and \(\tilde\theta^f\), respectively. All three measures have exponential tails, so that, up to an error of order \(e^{-cn}\) with \(c>0\) uniform in \(x'\) in a small neighbourhood of \(x'_0\), \eqref{eq:diamondExpect1} becomes
\begin{align}
\label{eq:analyticity1}
C\sum_{k=1}^{n} \sum_{\ell_0+\cdots+\ell_k=n} \mathsf{p}_L(\ell_0)\mathsf{p}_R(\ell_k)\prod_{i=1}^{k-1} \mathsf{Q}\bigl[\bigl(\tfrac{x'}{x_0'}\bigr)^{\ouv_{\tilde{D}_i\cap\Line}}\bigm\vert |\tilde{D}_i|=\ell_i \bigr] \mathsf{Q}(|\tilde{D}_i|=\ell_i),
\end{align}
where the $C$ comes from the pure exponential decay behaviour of $\p_{x_0'}(0\leftrightarrow n\eone)$ and the associated conditioning, and where
\[
\mathsf{p}_{L/R}(\ell) = \mathsf{p}_{L/R}\bigl[\bigl(\tfrac{x'}{x_0'}\bigr)^{\ouv_{\tilde{D}^{b/f}\cap\Line_{[0,n]}}} \IF{|\tilde{D}^{b/f}|=\ell} \bigr],
\]
are exponentially decaying in $\ell$ for \(x'\) in a small neighbourhood of \(x_0'\).
Notice that \(\mathsf{Q}\bigl[\bigl(\tfrac{x'}{x_0'}\bigr)^{\ouv_{\tilde{D}_i\cap\Line}}\bigm\vert |\tilde{D}_i|=\ell_i \bigr]\) is a polynomial in $x'$ of degree at most $\ell_i$. Denote $q_\ell=\mathsf{Q}(|\tilde{D}_1|=\ell)$ (which is exponentially decaying in $\ell$) and $p_\ell(x') = \mathsf{Q}\bigl[\bigl(\tfrac{x'}{x_0'}\bigr)^{\ouv_{\tilde{D}_1\cap\Line}}\bigm\vert |\tilde{D}_1|=\ell\bigr] e^{-(f_{x'}+\xi_{x_0'})\ell}$; observe that $p_\ell(x')$ is an analytic function of $x'$. Define $\alpha_0(x')=1$ and
\begin{align}
\alpha_n(x')
&=
\sum_{k=1}^{n} \sum_{\ell_1+\dots+\ell_k=n} \prod_{i=1}^{k} \mathsf{Q}\bigl[\bigl(\tfrac{x'}{x}\bigr)^{\ouv_{\tilde{D}_i\cap\Line}}\bigm\vert |\tilde{D}_i|=\ell_i\bigr] \mathsf{Q}(|\tilde{D}_i|=\ell_i) e^{-(f_{x'}+\xi_{x_0'})\ell_i}\notag\\
&=
\sum_{k=1}^{n} \sum_{\ell_1+\cdots+\ell_k=n} \prod_{i=1}^{k} q_{\ell_i} p_{\ell_i}(x').
\label{eq:acoef}
\end{align}

Denote $d_n = e^{-(f_{x'}+\xi_{x'_0})n}\e_{x'_0}^{(n)}\bigl[\bigl(\tfrac{x'}{x'_0}\bigr)^{\ouv_{\Line_{[0,n]}}} \bigm\vert 0\leftrightarrow n\eone\bigr]$ and \(\mathbb{D}(z)=\sum_{n\geq 1} d_n z^n\). By definition, the radius of convergence of \(\mathbb{D}\) is $e^{\xi_{x'}}$. Then, consider the generating functions
\[
\mathbb{A}(t) = \sum_{n=0}^{\infty}  \alpha_n(x') e^{tn}\quad\text{ and }\quad
\mathbb{B}(t) = \sum_{\ell=1}^{\infty} q_\ell\, p_\ell(x') e^{t\ell}.
\]
By~\eqref{eq:diamondExpect1} and~\eqref{eq:analyticity1}, the radius of convergence of \(\mathbb{A}\) is given by the one of \(\mathbb{D}(e^t)\), so that the radius of convergence of $\mathbb{A}(t)$ is equal to $\iclx$. Moreover, notice that $\mathbb{B}(t)$ converges for all $t<t_0$ for some $t_0>\iclx$ as $q_\ell$ is exponentially decaying in $\ell$.
Then, \eqref{eq:acoef} implies that
\begin{align*}
\mathbb{A}(t) = \frac{1}{1-\mathbb{B}(t)}.
\end{align*}
Thus $\mathbb{B}(\iclx)=1$, which provides an implicit expression for $\iclx$. Defining
\[
\Phi(w,z) = \sum_{\ell=1}^{\infty} q_\ell\, p_\ell(w)\, e^{z\ell},
\]
analyticity follows from solving $\Phi(w,z)=1$ for $z$ in a neighborhood of $(x',\iclx)$. Analyticity of $\Phi(w,z)$ close to $(x',\iclx)$ follows from the exponential decay property of $q_\ell$, and $\frac{\partial\Phi}{\partial z}\big|_{(x',\iclx)}\neq 0$ from a direct computation. The claim follows using the (analytic version of the) implicit function theorem.
\begin{proof}[Proof of Claim~\ref{claim:analyticityLineToInterval}]
	To simplify notations, we will write $j\equiv j\eone$ when the meaning is clear from the context.
	First notice that $\p_{x'}^{(n)}(0\leftrightarrow n)\leq \p_{x'}(0\leftrightarrow n)$ by monotonicity, since \(x'\geq x\). Thus $\iclx\leq -\lim_{n\to\infty}\frac{1}{n} \log \p_{x'}^{(n)}(0\leftrightarrow n)$. To obtain the reverse inequality, we partition, for $M>0$,
	\begin{multline}
	\p_{x'}(0\leftrightarrow n)
	=
	\p_{x'}(0\leftrightarrow n, 0\leftrightarrow(\Line_{<-M}\cup\Line_{>n+M}))\\
	+
	\p_{x'}(0\leftrightarrow n, 0\nleftrightarrow(\Line_{<-M}\cup\Line_{>n+M})).\label{eq:LineToInterval1}
	\end{multline}
	Now, we bound separately the two terms in the RHS.
	\begin{align}
		\p_{x'}(0\leftrightarrow n, 0\leftrightarrow&(\Line_{<-M}\cup\Line_{>n+M})) \nonumber\\
		&\leq
		\sum_{k<-M}\p_{x'}(k\leftrightarrow 0 \leftrightarrow n) + \sum_{k>n+M}\p_{x'}(0\leftrightarrow n\leftrightarrow k)\nonumber\\
		&\leq
		2\sum_{k=0}^{\infty}\p_{x'}(0\leftrightarrow n+M+k)\nonumber\\
		&\leq
		2e^{-\iclx n}e^{-\iclx M}\frac{1}{1-e^{-\iclx}}\nonumber\\
		&\leq
		C\p_{x'}(0\leftrightarrow n)e^{-\iclx M}\frac{1}{1-e^{-\iclx}}\equiv C_{x'}e^{-M\iclx}\p_{x'}(0\leftrightarrow n), \label{eq:LineToInterval2}
	\end{align}
	where the first inequality is a union bound, the second uses invariance under translation, the third is by Lemma~\ref{lem:BasicProp}, Item~\ref{it:unifZd}, and the fourth follows from Theorem~\ref{thm:PinnedRegime}, Item~\ref{Item:quatre} with $C\geq 0$ not depending on $n$. Then,
	\begin{align}
		\p_{x'}(0&\leftrightarrow n, 0\nleftrightarrow(\Line_{<-M}\cup\Line_{>n+M})) =\nonumber\\
		&=
		\!\!\!\!\!\!\!\sum_{\substack{C\ni 0,n\\C\cap(\Line_{<-M}\cup\Line_{>n+M})=\varnothing}} \!\!\!\!\!\!\!\p_{x'}(\partial C_0 \text{ closed})\p_{x'}(C\text{ open}\given\partial C_0 \text{ closed})\nonumber\\
		&\leq
		\!\!\!\!\!\!\!\sum_{\substack{C\ni 0,n\\C\cap(\Line_{<-M}\cup\Line_{>n+M})=\varnothing}} \!\!\!\!\!\!\!\p_{x'}^{(-M,n+M)}(\partial C_0 \text{ closed})\p_{x'}^{(-M,n+M)}(C\text{ open}\given\partial C_0 \text{ closed})\nonumber\\
		&\leq
		\p_{x'}^{(-M,n+M)}(0\leftrightarrow n)\nonumber\\
		&=
		\frac{\e\Big[\big(\frac{x'}{x}\big)^{\ouv_{\Line_{[-M,n+M]}}} \IF{0\leftrightarrow n}\Big]}{\e\Big[\big(\frac{x'}{x}\big)^{\ouv_{\Line_{[-M,n+M]}}}\Big]}\leq\Big(\frac{x'}{x}\Big)^{2M} \p_{x'}^{(n)}(0\leftrightarrow n), \label{eq:LineToInterval3}
	\end{align} where $\p_{x'}^{(-M,n+M)}$ denotes the measure with modified weights on $\Line_{[-M,n+M]}$. The first inequality is by monotonicity (and $\p_{x'}(\cdot\given\partial C_0 \text{ closed})=\p_{x'}^{(-M,n+M)}(\cdot\given\partial C_0 \text{ closed})$ as the interior of $\partial C_0$ does not contain edges from $\Line_{<-M}$ or $\Line_{>n+M}$).
	
	Putting together \eqref{eq:LineToInterval1},\eqref{eq:LineToInterval2}, \eqref{eq:LineToInterval3}, and choosing $M$ large enough so that $C_{x'}e^{-M\iclx}\leq\frac{1}{2}$ we obtain:
	\[
	\p_{x'}(0\leftrightarrow n)\leq \Big(\frac{x'}{x}\Big)^{2M}\frac{1}{1-C_{x'}e^{-M\iclx}}\p_{x'}^{(n)}(0\leftrightarrow n)\leq 2\Big(\frac{x'}{x}\Big)^{2M}\p_{x'}^{(n)}(0\leftrightarrow n),
	\] implying $\iclx \geq -\lim\limits_{n\to\infty}\frac{1}{n}\log \p_{x'}^{(n)}(0\leftrightarrow n\eone)$.
\end{proof}
\begin{proof}[Proof of Claim~\ref{claim:analyticityFreeEnergy}]
The proof will be done using the same line of ideas as described above for the analyticity of $\iclx$. First notice that, for $x'\geq x'_0$,
\begin{align*}
\e_{x'_0}^{(n)}\bigl[\bigl(\tfrac{x'}{x}\bigr)^{\ouv_{\Line_{[0,n+m]}}}\bigr]
\geq \e_{x'_0}^{(n)}\bigl[\bigl(\tfrac{x'}{x}\bigr)^{\ouv_{\Line_{[0,m]}}}\bigr]
\e_{x'_0}^{(n)}\bigl[\bigl(\tfrac{x'}{x}\bigr)^{\ouv_{\Line_{[0,n]}}}\bigr],
\end{align*}
by FKG, as $\ouv_{\Line_{[0,n]}}$ is an increasing function, and by translation invariance of $\e$ (for $x'<x'_0$, the reverse inequality holds). Thus, existence of $f_{x'}$ follows by Fekete's Lemma. Analyticity of $f_{x'}$ follows the same lines as $\iclx$: the same representation of $\Line_{[0,n]}$ cluster holds under $\e_{x'_0}^{(n)}$, and the rest of the argument carries out in the same (in fact, simpler) fashion as in the \(\iclx\) case.
\end{proof}

\subsection{Interface localization}
\label{sec:LocalizationLT}

Theorem~\ref{thm:PottsScalingInterface} is an essentially immediate corollary of the analysis leading to Theorem~\ref{thm:CPdensity} and classical tools for the analysis of the random-cluster model (see~\cite{Grimmett-2006}): the Edwards--Sokal coupling and the coupling between the high- and low-temperature random-cluster measures on \(\Ztwo\). It is enough to make the following observations:
\begin{itemize}
\item Whenever \(\{i,j\}\in\Etwo\) is such that \(\{i,j\}^*\) is part of the interface, the edge \(\{i,j\}\) is closed in the random-cluster configuration associated to the Potts configuration by the Edwards--Sokal coupling. Note that this random-cluster model has wired boundary condition and a constraint that \(\setof{i\in\Z^2}{i^\perp>0} \setminus \Lambda_n \) must no be connected (in \(\Lambda_n)\) to \(\setof{i\in\Z^2}{i^\perp\leq 0} \setminus \Lambda_n \).
\item By the standard coupling between the random-cluster model on \(\Z^2\) and its dual (which has parameters \(p^*<p_{\rm c}\), \(J^*>1\) and the same value of \(q\)), the latter has free boundary condition and is conditioned on the two dual vertices \((-n-\tfrac12,\tfrac12)\) and \((n+\tfrac12,\tfrac12)\) being connected. Let us denote by \(C_n\) the corresponding cluster.
\item By the above, the Potts interface is a subset of the (dual) cluster \(C_n\).
\item The analysis leading to Theorem~\ref{thm:CPdensity} can be repeated essentially verbatim, the fact that one is working in a finite system having no incidence.
\item This implies that the cone-points of \(C_n\) are also cone-points of the Potts interface, from which the desired result follows immediately.
\end{itemize}

\begin{figure}[h]
	\includegraphics[scale=0.6]{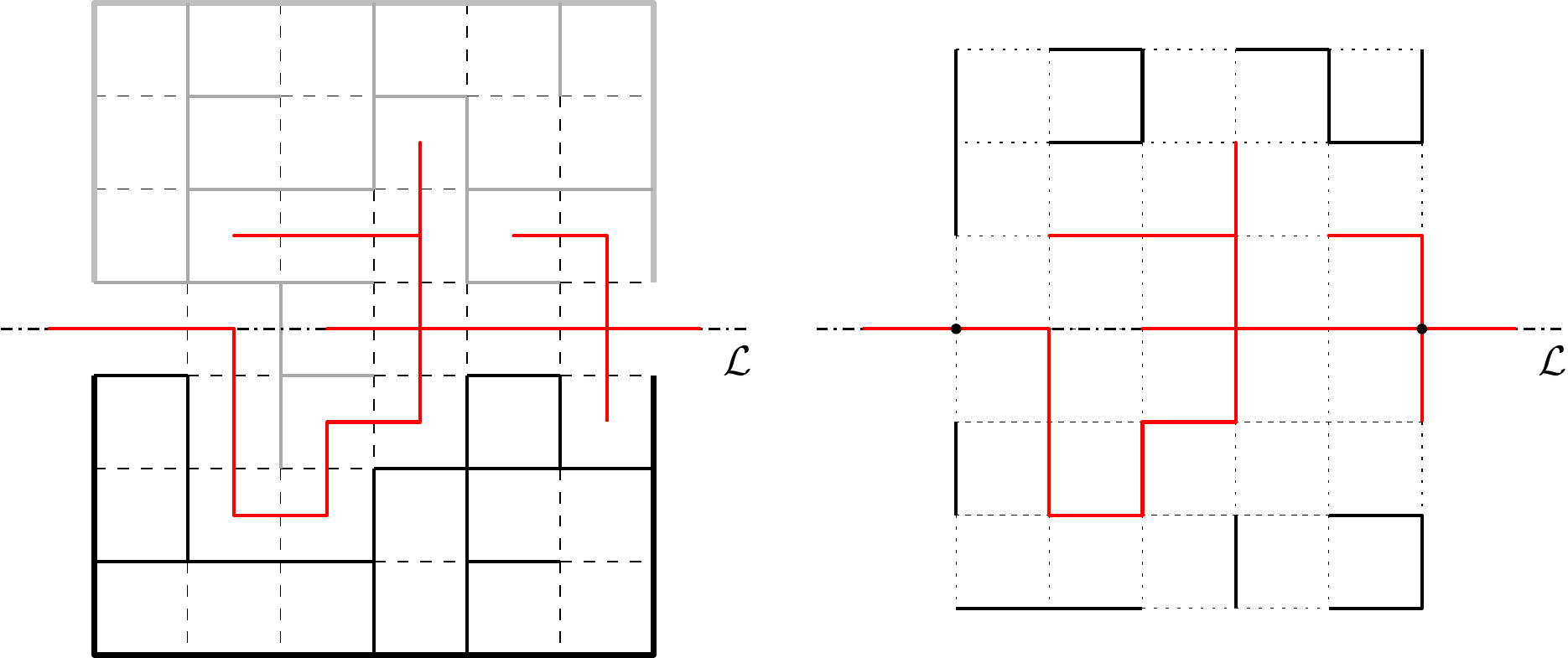}
	\caption{The FK representation of low temperature Potts model with Dobrushin boundary condition (the top is conditioned to not intersect the bottom) and the corresponding high temperature dual FK configuration (where the two points are conditioned to be connected).}
	\label{fig:interfaceFK}
\end{figure}


\appendix


\section{Couplings}
\label{app:Couplings}

We sketch here the proofs of the existence of some couplings used in the paper. Similar construction (with more details) can be found, for example, in \cite{Graham+Grimmett-2011} and \cite{Duminil-Copin+Manolescu-2016}.
\begin{lemma}
	Let $G$ be a finite graph and let $\p_{\vec{x},q}$ be the random-cluster measure with edges weights $(x_e)_{e\in E_G}$ and cluster weight $q$ on $G$. Then, for any $e\in E_G$, there exists a coupling $(\omega,\eta)\sim\Phi$ of $\p_{\vec{x},q}(\cdot \given \omega_e= 1)$ and $\p_{\vec{x},q}(\cdot \given \omega_e= 0)$ such that
	\begin{enumerate}[label=(\roman*)]
		\item $\omega\sim\p_{\vec{x},q}(\cdot \given \omega_e=1)$ and $\eta\sim\p_{\vec{x},q}(\cdot \given \eta_e=0)$,
		\item $\Phi(\omega\geq\eta)=1$.
	\end{enumerate} 
\end{lemma}
\begin{proof}
	This lemma is standard and follows from a Markov chain argument: start from $\omega^{(0)}\geq\eta^{(0)}$ and perform a heat bath dynamic simultaneously on the two configurations. Having constructed \(\omega^{(n-1)}, \eta^{(n-1)}\), construct \(\omega^{(n)}, \eta^{(n)}\) in the following way: select an edge $f$ uniformly at random from $E_G$; resample its state in \(\omega^{(n-1)}\) according to $\p_{\vec{x},q}(\cdot \given \omega_g=\omega^{(n-1)}_g\ \forall g\notin \{f,e\},\omega_e=1)$ to obtain $\omega^{(n)}$; resample its state in \(\eta^{(n-1)}\) according to $\p_{\vec{x},q}(\cdot \given \eta_g=\eta^{(n-1)}_{g}\ \forall g\notin\{f,e\},\eta_e=0)$ to obtain $\eta^{(n)}$. The two dynamics can be coupled so that for every $n$, the law of $\omega^{(n+1)}$ dominates the law of $\eta^{(n+1)}$. Letting $n\to\infty$, this gives the desired coupling.
\end{proof}

\begin{lemma}
	Let $G$ be a finite graph, let $\p_{\vec{x},q}$ be the random cluster measure with edges weights $(x_e)_{e\in E_G}$ and cluster weight $q$ on $G$ and, for $E\subset E_G$, let $\p_{\vec{y},q}$ be the random cluster measure with edges weights $y_e=
	\begin{cases}
		x_e 	& \text{ if } e\notin E\\
		y_e<x_e & \text{ if } e\in E
	\end{cases}$, and cluster weight $q$. Then, there exists a coupling $(\omega,\eta)\sim\Phi$ of $\p_{\vec{x},q}$ and $\p_{\vec{y},q}$ such that
	\begin{enumerate}[label=(\roman*)]
		\item $\omega\sim\p_{\vec{x},q}$ and $\eta\sim\p_{\vec{y},q}$,
		\item $\Phi(\omega\geq\eta)=1$,
		\item $\Phi(\omega|_{(C_E(\omega))^\comp} = \eta|_{(C_E(\omega))^\comp})=1$.
	\end{enumerate} 
\end{lemma}
\begin{proof}
	This coupling is slightly more involved and is done via an exploration process. Fix an arbitrary ordering of $E_G$. We will explore the configurations by exploring the cluster of $E$. Denote $C_E^{(n)}(\omega)$ the cluster of $E$ in $\omega$ (that is, the union of the clusters of the endpoints of the edges in $E$) restricted to the explored edges after step $n$ (it always contains the endpoints of the edges in $E$) and let $\partial C_E^{(n)}(\omega)$ be the unexplored edge-boundary of $C_E^{(n)}(\omega)$. At step $n$, sample the smallest edge $e_n\in\partial C_E^{(n-1)}(\omega)$ as follows: sample $U_n\sim\text{Unif}([0,1])$ and set
	\[
		\omega_{e_n}
		=
		\IF{U_n\leq \p_{\vec{x},q}(\cdot \given \omega_{e_1},\ldots,\omega_{e_{n-1}})}
		\quad\text{ and }\quad
		\eta_{e_n}
		=
		\IF{U_n\leq \p_{\vec{y},q}(\cdot \given \eta_{e_1},\ldots,\eta_{e_{n-1}})}.
	\]
	In this way, when an edge is open in \(\eta\), it is also open in \(\omega\). Observe that, once the cluster of $E$ in $\omega$ is explored, its boundary will be closed in both configurations. We can thus sample the remaining edges in both configurations according to $\p_{\vec{x},q;(C_E(\omega)\cup\partial C_E(\omega))^\comp}^f$, so the two agree outside of $C_E(\omega)$.
\end{proof}

\section{Basic results in FK percolation}

\subsection{A decoupling inequality}
The following lemma is inspired by an analogous claim in~\cite{Campanino+Ioffe+Velenik-2008}.
\begin{lemma}
	\label{lem:LfreeEst}
	Let $R>0$ and let $A$ be an increasing event depending only on edges in a finite set $D\subset\Zd$. Define $D_R=\bigcup_{i\in D}(i+\llbracket-R,R\rrbracket^d)$. Then, for all $R$ large enough,
	\[
	\frac{\p_{D_R}^{\wired}(A)}{\p(A)}
	\leq
	\bigl( 1 - |\partial D_{R/2}| |\partial D_R| e^{-\Cr{decayFK}R/2} \bigr)^{-1} .
	\]
\end{lemma}
\begin{proof}
	Notice first that, for $u\in\partial D_{R/2}$ and $v\in\partial D_R$, the distance between $u$ and $v$ is at least $R/2$. Then, partitioning according to whether the event $\partial D_R\leftrightarrow\partial D_{R/2}$ occurs, we get
	\begin{align*}
	\p_{D_R}^{\wired}(A)
	&\leq
	\p_{D_R}^{\wired}(A)\, \p_{D_R\setminus D}^{\wired}(\partial D_R\leftrightarrow\partial D_{R/2}) + \p(A)\\
	&\leq
	\p_{D_R}^{\wired}(A) |\partial D_{R/2}| |\partial D_R| e^{-\Cr{decayFK}R/2} + \p(A) ,
	\end{align*}
	where we used monotonicity in volume and in boundary conditions for the first inequality and Lemma~\ref{lem:expoDecFKwired}
	for the second one.
\end{proof}

\subsection{A Russo-like formula}
\label{app:Russo}

There exist various extensions of the Russo formula from Bernoulli percolation to FK percolation. However, we will need the following version, which we did not find in the literature. Recall that an edge \(e\) is pivotal for an event \(A\) in a configuration \(\omega\) if the value of \(\IF{\omega\in A}\) depends on the value of \(\omega_e\). Denote \(\textnormal{Piv}_A(\omega)\) the set of edges pivotal for \(A\) in \(\omega\).
\begin{lemma}
	\label{lem:RussoGen}
	Let $\p_{\vec{x},q}$ be the random-cluster measure on a finite graph $G$, with weights $(x_e)_{e\in E_G}$ and $q\geq 1$.
	Let $E\subset E_G$ a collection of edges in $G$. Denote by $\p_{\vec{x},q}^s$ the random-cluster measure obtained by modifying the weights $\vec x$ by setting $x_e=s,\ \forall e\in E$. Then, for $s_2>s_1$ and any nondecreasing event \(A\), we have
	\[
	\frac{\p_{\vec{x},q}^{s_2}(A)}{\p_{\vec{x},q}^{s_1}(A)}
	\geq
	\exp\Bigl(\int_{s_1}^{s_2} \frac{1}{s(1+s)} \sum_{e\in E} \p_{\vec{x},q}^s(e\in\textnormal{Piv}_A \given A)\dd s \Bigr) .
	\]
\end{lemma}
\begin{proof}
	First, we compute
	\begin{align}
	\frac{\dd}{\dd s}\log \p_{\vec{x},q}^s(A)
	&=
	\frac{1}{s\p_{\vec{x},q}^s(A)} \textnormal{Cov}_{\vec{x},q}^s(\mathds{1}_A,\ouv_E)\nonumber\\
	&=
	\frac{1}{s\p_{\vec{x},q}^s(A)} \Bigl( \sum_{e\in E} \p_{\vec{x},q}^s(\omega_e=1)
	\bigl(\p_{\vec{x},q}^s(A \given \omega_e=1) - \p_{\vec{x},q}^s(A)\bigr)\Bigr)\nonumber\\
	&=
	\frac{1}{s\p_{\vec{x},q}^s(A)} \Bigl( \sum_{e\in E} \p_{\vec{x},q}^s(\omega_e=1) \p_{\vec{x},q}^s(\omega_e=0)\notag\\
	&\pushright{\times\bigl(\p_{\vec{x},q}^s(A \given \omega_e=1) - \p_{\vec{x},q}^s(A \given \omega_e=0)\bigr)\Bigr).}\label{eq:deriv}
	\end{align}
	Consider a coupling $(\omega,\eta)\sim\Phi$ of $\p_{\vec{x},q}^s(\cdot \given \omega_e=1)$ and $\p_{\vec{x},q}^s(\cdot \given \omega_e=0)$ such that $\omega\geq\eta$ and $\omega\sim \p_{\vec{x},q}^s(\cdot \given \omega_e=1)$ and $\eta\sim \p_{\vec{x},q}^s(\cdot \given \omega_e=0)$. Then compute
	\begin{align*}
	\p_{\vec{x},q}^s( A \given \omega_e=1) - \p_{\vec{x},q}^s(A \given \omega_e=0)
	&=
	\Phi(\mathds{1}_A(\omega) - \mathds{1}_A(\eta)) \\
	&=
	\Phi(\mathds{1}_A(\omega) \mathds{1}_{A^\comp}(\eta))\\
	&\geq 
	\Phi(\IF{e\in\textnormal{Piv}_A}(\omega))\\
	&=
	\p_{\vec{x},q}^s(e\in\textnormal{Piv}_A \given \omega_e=1)\\
	&=
	\frac{\p_{\vec{x},q}^s(e\in\textnormal{Piv}_A,\omega_e=1)}{\p_{\vec{x},q}^s(\omega_e=1)}
	=
	\frac{\p_{\vec{x},q}^s(e\in\textnormal{Piv}_A,A)}{\p_{\vec{x},q}^s(\omega_e=1)} ,
	\end{align*}
	where we used, in the second line, that $A$ is increasing and $\omega\geq\eta$, so that $\eta\in A\implies\omega\in A$ and $\omega\in A^\comp\implies\eta\in A^\comp$ and thus $\mathds{1}_A(\omega)-\mathds{1}_A(\eta) = \mathds{1}_A(\omega)\mathds{1}_{A^\comp}(\eta)$; we have also used the fact that $e\in\textnormal{Piv}_A(\omega)\implies \omega\in A$ and $\eta\in A^\comp$ for the inequality.
	Plugging this into~\eqref{eq:deriv} gives
	\begin{align*}
	\frac{\dd}{\dd s}\log(\p_{\vec{x},q}^s(A))
	&\geq
	\frac{1}{s\p_s(A)} \Bigl( \sum_{e\in E} \p_{\vec{x},q}^s(\omega_e=0) \p_{\vec{x},q}^s(e\in\textnormal{Piv}_A,A)\Bigr)\\
	&=
	\frac{1}{s} \Bigl( \sum_{e\in E} \p_{\vec{x},q}^s(\omega_e=0) \p_{\vec{x},q}^s(e\in\textnormal{Piv}_A \given A)\Bigr)\\
	&\geq
	\frac{1}{s(1+s)} \Bigl( \sum_{e\in E} \p_{\vec{x},q}^s(e\in\textnormal{Piv}_A \given A) \Bigr),
	\end{align*}
	where the last inequality follows from finite energy of $\p_{\vec{x},q}^s$. Integrating both sides between $s_1$ and $s_2$ and taking the exponential leads to the desired inequality.
\end{proof}


\section{Renewal for long-range memory process}
\label{sec:RenLongRange}

The goal of this appendix is to present a way to factorize measures on sequences with exponential mixing. The procedure employed is a representation of the mixing property as a memory-percolation picture. The ideas used here are inspired from the construction done in \cite{Comets+Fernandez+Ferrari-2002}, but our set-up being a bit different (we deal with general kernels instead of probability kernels and we need ``finite volume'' estimates rather than estimates on the stationary measure), the results from \cite{Comets+Fernandez+Ferrari-2002} do not immediately apply, so we provide here a self-contained exposition. 

\subsection{Setting, Notations and Definition}
\label{app:Ren_Setting}

We will work with $\alp$ an alphabet (finite or countable), and $\bc_L,\bc_R$ two sets containing $\emptyset$ (finite or countable). The objects of study will be measures on sequences of the form \[(b_L,x_1,x_2,\dots,x_n,b_R)\in \bc_L\times\alp^n\times\bc_R.\]
We will say and assume:
\begin{itemize}
	\item elements of $\alp$ are called letters, sequences (or concatenation) of letters are called words;
	\item \(\alp\) does not contain words;
	\item for $\bfx \in\alp^n$, denote \(|\bfx|=n\) the length of the word $\bfx$.
\end{itemize}
As we work with sequences, it will be useful to have a few operations on them. We first define the concatenation operation.
\begin{definition}
	For $\mathbf{x}=(\dots,x_{k-1},x_{k})$ a right-finite sequence and $\mathbf{y}=(y_l,y_{l+1},\dots)$ a left-finite sequence, the \emph{concatenation} of $\mathbf{x}$ and $\mathbf{y}$ is the sequence
	\[
		\mathbf{x}\sqcup\mathbf{y}=(\dots,x_{k-1},x_{k},y_l,y_{l+1},\dots).
	\]
	By convention, the labels of the new sequence will be chosen to be consistent with the labels of $\mathbf{x}$:
	\[
		(\mathbf{x}\sqcup\mathbf{y})_{i}
		=
		\begin{cases}
			x_i 		& \text{ if } i\leq k\\
			y_{l+i-k-1} & \text{ if } i > k
		\end{cases}.
	\]
\end{definition}
Elements of \(\alp\) will be considered as one-element sequences for concatenation.
We then define the extraction operation.
\begin{definition}
	For $k\leq l\leq m\leq n \in\Z$ and $\mathbf{x}=(\dots,x_{k},\dots,x_{l},\dots,x_{m},\dots,x_n,\dots)$ a sequence, the \emph{$(l,m)$-extraction} of $\bfx$ is the sequence
	\[
		\bfx_l^m=(x_l,x_{l+1},\dots,x_{m-1},x_m).
	\]
\end{definition}
We will use the following notations:
\begin{itemize}
	\item $\calS=\bigcup_{n\geq 0}\alp^n$ the set of finite sequences ($\alp^0=\{\emptyset\}$), and $\calS^*=\bigcup_{n\geq 1}\alp^n$ the set of non-empty finite sequences;
	\item $\calS^+ = \alp^{\Z_{\geq 0}}$, resp.\ $\calS^-= \alp^{\Z_{< 0}}$, the set of right-infinite, resp.\ left-infinite, sequences;
	\item $\compl{\bc}_L = \setof{b_L\sqcup \mathbf{x}}{b_L\in\bc_L, \mathbf{x}\in\calS}$ and $\compl{\bc}_R=\setof{\mathbf{x}\sqcup b_R}{b_R\in\bc_R, \mathbf{x}\in\calS}$.
\end{itemize}

In all this Appendix, when not explicitly said otherwise, $b_L,b_R,\bfx$ will always denote elements of $\bc_L,\bc_R,\calS^*$ and $\bfx=(x_1,\dots,x_n)$.

\bigskip

We will consider measures $\subProb_n$ on $\bc_L\times\calS^n\times\bc_R$ that are given by a kernel $\subProb:\compl{\bc}_L\times (\alp\cup\bc_R)\to \R_+$ and a weight function \(\subProb:\bc_L\to \R_+\) (for simplicity, we denote both by the same letter...). Namely, writing
\[
	\subProb(b,s)\equiv\subProb(s\given b),
\]
we assume that
\begin{align*}
\subProb_n(b_L\sqcup\bfx\sqcup b_R)
&=
\subProb(b_L)\Bigl( \prod_{k=1}^{n} \subProb(x_k\given b_L\sqcup\bfx_1^{k-1}) \Bigr)
\subProb(b_R\given b_L\sqcup\bfx_1^n)
\end{align*}
To lighten the notations, we will sometimes write
\[
	\subProb(b_L\sqcup \bfx_1^n) = \subProb(b_L)\prod_{i=1}^{n}\subProb(x_i\given b_L\sqcup \bfx_1^{i-1}).
\]

We will make the following additional assumptions on $\subProb$:
\begin{enumerate}[label=(H\arabic*)]
	\item \label{hyp:unifSumability} uniform summability: there exists $K<\infty$ such that
	\[
		\sum_{s\in\alp\cup\bc_{R/L}}\subProb(s\given b) \leq K
	\]
	for all $b\in\compl{\bc}_{L/R}$;
	\item \label{hyp:expMix} ratio exponential mixing: there exist $L_0,c>0$ such that
	\[
		\Bigl|1-\frac{\subProb(\mathbf{z}|b_2\sqcup\bfx)}{\subProb(\mathbf{z}|b_1\sqcup\bfx)}\Bigr| \leq e^{-c |\bfx|}
	\]
	for all $n\geq L_0,\mathbf{z},\bfx=(x_1,\dots,x_n)\in\calS,b_1,b_2\in\compl{\bc}_L$;
	\item \label{hyp:subExpDecOfMass} sub-exponential decay (or growth) of the mass:
	\[
		\lim_{n\to\infty} \frac{1}{n}\log(\mu_n) = 0,
	\]
	where $\mu_n=\sum_{b_L,b_R}\sum_{\bfx\in\alp^n}\subProb_n(b_L\sqcup\bfx\sqcup b_R)$.
	\item \label{hyp:azero} there exist \(s\in\alp\) and \(\epsilon_0>0\) such that $\inf_{l\in\compl{\bc}_L} \subProb(s\given l)\geq\epsilon_0$.
\end{enumerate}

\subsection{The Memory Percolation Picture}

A stick-percolation configuration on an interval $J\subset\R$ is a partitioning of $J$ into disjoints open intervals (called \emph{clusters}) and their endpoints (called \emph{cuts}). Given a stick-percolation configuration $\omega$, denote by $\cut(\omega)$ the set of its cuts. We will consider stick-percolation configurations induced by functions $I:\Z\to\Z_{\geq 0}$ via the following procedure: to every $k\in\Z$, associate the (open) interval $(k-I_k-\frac{1}{2},k+\frac{1}{2})$. Then, the connected components of the union of those intervals give the clusters of the stick-percolation configuration, while the complement of this union gives the set of cuts. We will say that an edge $e_k=(k,k+1)$ is a cut if $k+\frac{1}{2}$ is.

This definition extends straightforwardly for stick-percolation configurations on finite subsets of $\Z$.
\begin{figure}[h]
	\label{fig:StickPerco}
	\includegraphics[scale=0.75]{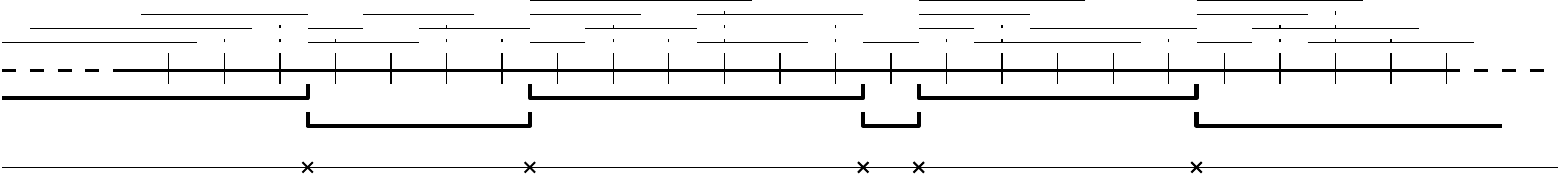}
	\caption{Left stick-percolation configuration and the cut process.}
\end{figure}

\bigskip
With this in hand, we augment each sequence $b_L\sqcup\bfx\sqcup b_R$, \(\bfx\in\alp^n\), with a stick-percolation realization on $[0,n+1]\cap\Z$. This will be done with the help of a memory threshold sequence (following~\cite{Comets+Fernandez+Ferrari-2002}). Let $\bfx=(x_1,\dots,x_n)\in\calS^*, b\in\compl{\bc}_L$ and define
\begin{align*}
a_0(s\given b\sqcup\bfx)\equiv a_0(s \given *)
&=
\inf_{l\in\compl{\bc}_L} \subProb(s\given l) ,\\
a_k(s \given b\sqcup\bfx)\equiv a_k(s \given \bfx_{n-k+1}^{n})
&=
\inf_{l\in\compl{\bc}_L} \subProb(s\given l\sqcup\mathbf{x}_{n-k+1}^{n}) , \\
a_k(b\sqcup\mathbf{x})
&=
\sum_{s\in\alp}a_{k}(s \given b\sqcup\mathbf{x}_{1}^{n}) ,
\end{align*}
if $0\leq k\leq n$, and $a_k(s \given b\sqcup \mathbf{x}) = a_{\infty}(s \given b\sqcup \mathbf{x}) = \subProb(s\given b\sqcup \mathbf{x})$ for $k> n$.
Under Assumption~\ref{hyp:unifSumability}, all those numbers are in $ [0,K]$ for all $k$ and $a_k(s \given b\sqcup \mathbf{x})$, $a_k(b\sqcup \bfx)$ are nondecreasing sequences in $k$ for any $s$, $b$ and $\bfx$.
One can thus consider the ``covered mass at depth $k$'':
\begin{align*}
\Delta_0(s \given b\sqcup \bfx) &= a_0(s \given *),\\
\Delta_k(s \given b\sqcup \bfx) &= a_k(s \given b\sqcup \bfx)- a_{k-1}(s \given b\sqcup \bfx),\\
\Delta_{k}(b\sqcup \bfx)&=\sum_{s\in\alp} \Delta_k(s \given b\sqcup \bfx).
\end{align*}
All these definitions are for $s\in\alp$, but they extend straightforwardly to the case where $s$ is replaced by $b'\in\bc_R$.
Observe that Assumption~\ref{hyp:azero} is equivalent to the existence of \(\epsilon_0>0\) such that $a_0\geq\epsilon_0$.

Now, noticing that
\begin{align*}
\subProb(s\given b\sqcup\bfx_1^{n})
&=
a_{n+1}(s\given b\sqcup\bfx_1^{n})\\
&=
a_{0}(s\given b\sqcup\bfx_1^{n}) + \sum_{k=0}^{n} \bigl\{ a_{k+1}(s\given b\sqcup\bfx_1^{n}) - a_{k}(s\given b\sqcup\bfx_1^{n}) \bigr\}\\
&=
\sum_{k=0}^{n+1}\Delta_{k}(s\given b\sqcup\bfx_1^{n}) ,
\end{align*}
one can write
\begin{align*}
\subProb_n(b_L\sqcup\bfx\sqcup b_R)
&=
\subProb(b_L) \Bigl( \prod_{k=1}^{n} \subProb(x_k\given b_L\sqcup\bfx_1^{k-1}) \Bigr) \subProb(b_R\given b_L\sqcup\bfx_1^{n})\\
&=
\subProb(b_L) \Bigl( \prod_{k=1}^{n} \sum_{i= 0}^{k} \Delta_i(x_k \given b_L\sqcup \bfx_{1}^{k-1}) \Bigr) \Bigl( \sum_{i= 0}^{n+1} \Delta_i(b_R \given b_L\sqcup \bfx_{1}^{n}) \Bigr)\\
&=
\subProb(b_L) \sum_{I\in\calI_n} \Delta_{I_{n+1}}(b_R \given b_L\sqcup \bfx_{1}^{n}) \prod_{k=1}^{n}\Delta_{I_k}(x_k \given b\sqcup \bfx_{1}^{k-1}),
\end{align*}
where $\calI_n = \bsetof{I:\{0,1,\dots,n+1\}\to\Z_{\geq 0}}{I_k\equiv I(k)\leq k}$. Now enters the memory-percolation picture: $I$ can be seen as the realization of a stick-percolation. In this way, each $I$ can be associated to a cluster set that we will represent as the sequence of the lengths of its clusters:
\[
	\clusterSet_n = \bigcup_{k=0}^{n+1}\bsetof{(l_0,\dots,l_k)}{l_i\geq 1,\ \sum_{i=0}^{k}l_i=n+2}.
\]
We will write $I\sim (l_0,\dots,l_k)$ if the cuts of the stick-percolation configuration induced by $I$ are $(l_0-1,l_0), (l_0+l_1-1,l_0+l_1),\dots,(l_0+\dots+l_{k-1}-1,l_0+\dots+l_{k-1})$. One can thus see $\subProb_n$ as a measure on $(\bc_L\sqcup\alp^n\sqcup \bc_R)\times \clusterSet_{n}$ (equipped with the discrete sigma-algebra):
\begin{multline}\label{eq:PsinOnSticks}
\subProb_n(b_L\sqcup\bfx\sqcup b_R, (l_0,l_1,\dots,l_k)) \\
=
\subProb(b_L) \sum_{\substack{I\in\calI_n\\I\sim(l_0,\dots,l_k)}} \Delta_{I_{n+1}}(b_R \given b_L\sqcup \bfx_{1}^{n}) \prod_{i=1}^{n} \Delta_{I_i}(x_i \given b_L\sqcup \bfx_{1}^{i-1}).
\end{multline}

\begin{figure}[h]
	\label{fig:seqWithMemory}
	\includegraphics[scale=0.8]{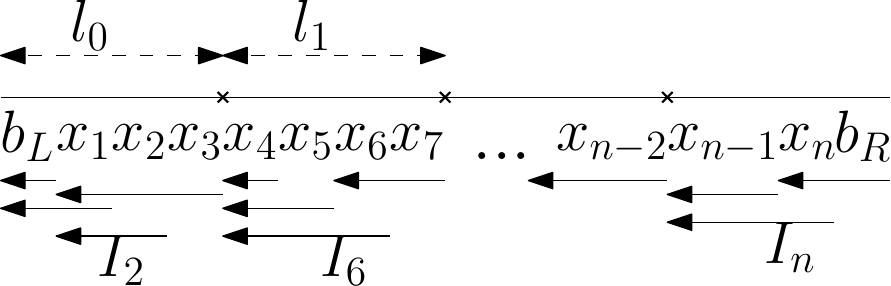}
	\caption{Sequence with memory.}
\end{figure}

Notice that, for a given cluster realization $(l_0,\dots,l_k)$, the value of the weight $\Delta_{I_i}(x_i \given b_L\sqcup \bfx_{1}^{i-1})$ for a given $i$ is independent of the value of \(x_j\) for \(j<i-I_i\). It is this essential property that will be exploited in our analysis.

Define
\begin{align}
\begin{split}
\label{eq:weights}
\rho_L(b_L\sqcup \bfx_1^n)
&=
\subProb(b_L) \sum_{I\sim (n+1)}\prod_{k=1}^{n}\Delta_{I_k}(x_k \given b_L\sqcup \bfx_{1}^{k-1})\\
\rho_R(\bfx_1^n\sqcup b_R) ,
&=
\sum_{I\sim (1,n+1)}\Delta_{I_{n+1}}(b_R \given \bfx_{1}^{n})\prod_{k=1}^{n}\Delta_{I_k}(x_k \given \bfx_{1}^{k-1})\\
\statMes(\bfx_1^n) ,
&=
\sum_{I\sim (1,n)}\prod_{k=1}^{n}\Delta_{I_k}(x_k \given \bfx_{1}^{k-1}) .
\end{split}
\end{align}
These are obviously nonnegative measures on, respectively, $\bc_L\times\calS$,$\calS\times\bc_R$ and $\calS$. Moreover, denoting $M+1$ the (variable) number of clusters in the percolation configuration, and defining $A=\{M\geq 1, l_0+l_M<n+2\}$, we have
\begin{align}
&\subProb_n(b_L\sqcup\bfx\sqcup b_R, A) = \sum_{L=2}^{n+1}\sum_{\substack{l\geq 1,r\geq 1\\ r+l =L}}\subProb_n(b_L\sqcup\bfx\sqcup b_R, l_0=l, l_M =r) \notag\\
&\qquad
= \sum_{\substack{r+l=L=2\\r,l\geq 1}}^{n+1}\sum_{k=1}^{n+2-L}\sum_{\substack{l_1,\dots,l_k\geq 1\\ \sum l_i= n+2-L}}\subProb(b_L)\Bigl\{\sum_{I\sim(l)}\prod_{i=1}^{l-1}\Delta_{I_i}(x_i \given b_L\sqcup \bfx_{1}^{i-1})\Bigr\}\times \notag\\
&\hspace{1.2cm}
\times\prod_{j=1}^{k}\Bigl\{\sum_{I\sim(1,l_j)} \prod_{i=1}^{l_j} \Delta_{I_i}\bigl(x_{l+l_1\dots+l_{j-1}+i}\bgiven\bfx_{l+l_1\dots+l_{j-1}+1}^{l+l_1+\dots+l_{j-1}+i-1}\bigr) \Bigr\}\times \notag\\
&\hspace{1.2cm}
\times\Bigl\{ \sum_{I\sim(1,r)} \Delta_{I_{r}}\bigl(b_R\bgiven\bfx_{l+\dots+l_{k}+1}^{l+\dots+l_{k}+r-1} \bigr) \prod_{i=1}^{r-1} \Delta_{I_i}\bigl(x_{l+\dots+l_{k}+i}\bgiven\bfx_{l+\dots+l_{k}+1}^{l+\dots+l_{k}+i-1}\bigr) \Bigr\} \notag\\
&\qquad
= \sum_{\substack{r+l=L=2\\r,l\geq 1}}^{n+1} \sum_{k=1}^{n+2-L} \sum_{\substack{l_1,\dots,l_k\geq 1\\ \sum_i l_i= n+2-L}} \rho_L\bigl(b_L\sqcup \bfx_1^{l-1} \bigr) \prod_{j=1}^{k} \statMes \bigl(\bfx_{l+\dots+l_{j-1}+1}^{l+\dots+l_{j-1}+l_j}\bigr)\times \notag\\
&\hspace{1.2cm}
\times\rho_R\bigl(\bfx_{l+\dots+l_{k}+1}^{l+\dots+l_{k}+r-1}\sqcup b_R \bigr).
\label{eq:factorisation}
\end{align}

\subsection{Decoupling of Random Sequences}
\label{sec:DecouplingRandomSequences}
We now present a factorisation result for weakly coupled measures. We always see $\subProb_n$ as a measure on $(\bc_L\times\alp^n\times\bc_R)\times\calI_n$ (with the discrete $\sigma$-algebra)
as the percolation picture is induced by the memory values ($I\in\calI_n$) and thus all weights that we consider can be expressed as sums of weight of elements in $(\bc_L\times\alp^n\times\bc_R)\times\calI_n$.

The idea being to approximate $\subProb_n$ by a factorized measure, we introduce the product measure $\p = \statMes^{\mathbb{Z}_{>0}}$, $\bfX=(\bfX_1,\bfX_2,\dots)$ and $\bfB_{L/R}$ sequences ``sampled'' from $\p,\rho_{L/R}$ (one can just think as if they were random variables, and look at them as a convenient way of defining certain sets). Then define $\mathcal{R}_L=\{\exists k\geq 1:\sum_{i=1}^k|\bfX_i|=L\}$ and
\begin{gather}
	\label{eq:factoMeasDef}
	\decouplLaw = \rho_L\times\p\times\rho_R,\\
	\decouplLaw_n(\,\cdot\,)= \decouplLaw(\,\cdot\,,|\bfB_L|+|\bfB_R|<n,\calR_{n-|\bfB_L|-|\bfB_R|}),
\end{gather}
where $\decouplLaw_n$ is understood as a measure on $(\bc_L,\alp^n,\bc_R)\times\clusterSet_n$.
Percolation estimates and the construction described in the previous section allow one to prove the following result.

\begin{lemma}
	\label{lem:decoupling}
	Let $(\subProb_n)_n$ and $\subProb$ be as described in Section~\ref{app:Ren_Setting} and such that Conditions~\ref{hyp:unifSumability}, \ref{hyp:expMix}, \ref{hyp:subExpDecOfMass} and~\ref{hyp:azero} are satisfied. Let $\rho_L,\rho_R,\statMes$ be defined by~\eqref{eq:weights}. Then, $\statMes$ is a probability measure and
	\begin{enumerate}[label=(\roman*)]
		\item \label{item:expDecOfSeqWeights} there exist $\Cl[csts]{wCDec},\Cl{wExpDec}>0$ such that
		\[
				\rho_R(|\bfB_R|=l) \vee \rho_L(|\bfB_L|=l) \vee \statMes(|\bfX|=l)\leq \Cr{wCDec}e^{-\Cr{wExpDec}l};
		\]
		\item \label{item:totalVarEstimate} there exist $\Cl[csts]{TVC},\Cl{TVexpc}>0$ such that, for any \(f:\bc_L\times\alp^n\times\bc_R\to \R\) bounded,
		\[
			\big|\sum f(\bfy)\subProb_n(\bfy)-\sum f(\bfy)\decouplLaw_n(\bfy)\big|\leq \norm{f}_{\infty}\Cr{TVC}e^{-\Cr{TVexpc}n},
		\] where the sum is over \(\bfy \in \bc_L\times\alp^n\times\bc_R\).
	\end{enumerate}
\end{lemma}
\begin{proof}
	We start by showing Item~\ref{item:expDecOfSeqWeights}, as the second point follows from it and~\eqref{eq:factorisation}.
	
	\smallskip
	To lighten notations, we will use the notation \(s_{k}^{n}\equiv(s_i)_{i=k}^{n}\) for any kind of sequence (not just words) and write \(s_{k}^{n}=\tilde{s}_{k}^{n}\) instead of: \(s_i=\tilde{s}_i\) for \(i=k,k+1,...,n\).
	First observe that, using~\ref{hyp:expMix} and the definition of $a_n$,
	\begin{align}
	\label{eq:seqElExpMix}
		\Bigl| \frac{a_{n}(s\given\bfx_{1}^{n})}{\subProb(s\given b\sqcup\bfx_{1}^{n})} - 1 \Bigr|
		=
		1 - \frac{a_{n}(s\given\bfx_{1}^{n})}{\subProb(s\given b\sqcup\bfx_{1}^{n})}
		\leq
		e^{-cn},
	\end{align}
	uniformly in $b\in\compl{\bc}_L$ and $s\in\alp$ (or $s\in\bc_R$) whenever $n\geq L_0$. Thus, for any $b\in\compl{\bc}_L$ and $\bfx\in\alp^n, n\geq L_0$,
	\begin{align*}
	\sum_{s\in\alp} \bigl\{ \subProb(s\given b\sqcup\bfx_{1}^{n}) - a_{n}(s\given\bfx_{1}^{n}) \bigr\}
	&\leq e^{-cn} \sum_{s\in\alp} \subProb(s\given b\sqcup\bfx_{1}^{n}).
	\end{align*}
	This and the fact that \(a_n\geq a_0\geq \epsilon_0\) imply that
	\begin{align}
	\label{eq:seqExpMix}
		\frac{\sum_{s\in\alp} a_{n}(s\given\bfx_{1}^{n})}{\sum_{s\in\alp} \subProb(s\given b\sqcup\bfx_{1}^{n})}
		\geq
		\begin{cases}
			1-e^{-cn}	& \text{ if } n\geq L_0, \\
			\frac{\epsilon_0}{K}	& \text{ if } n<L_0.
		\end{cases}
	\end{align}
	We can then use \eqref{eq:seqElExpMix} to obtain a ``uniform'' exponential decay estimate on $I_k$: for $L_0\leq l< k$,
	\begin{align}
		&\subProb_n(I_k\leq l, I_{k+1}^{n+1}=r_{k+1}^{n+1}) = \notag\\
		&\qquad = \sum_{b_L,b_R,\bfx_1^n}\subProb(b_L\sqcup\bfx_1^{k-1})a_{l}(x_k\given b_L\sqcup\bfx_1^{k-1})\prod_{i=k+1}^{n}\Delta_{r_i}(x_i\given b_L\sqcup \bfx_1^{i-1})\notag\\
		&\qquad \geq \sum_{b_L,b_R,\bfx_1^n}\subProb(b_L\sqcup\bfx_1^{k}) (1-e^{-cl})\prod_{i=k+1}^{n}\Delta_{r_i}(x_i\given b_L\sqcup \bfx_1^{i-1})\notag\\
		&\qquad = (1-e^{-cl})\subProb_n(I_{k+1}^{n+1}=r_{k+1}^{n+1}) .
		\label{eq:UnifExpDecComput}
	\end{align}
	For $l\leq L_0$, the same computation and \eqref{eq:seqExpMix} gives
	\[
	\subProb_n(I_k\leq l, I_{k+1}^{n+1}=r_{k+1}^{n+1}) \geq \epsilon_0\subProb_n(I_{k+1}^{n+1}=r_{k+1}^{n+1}).
	\]
	Reformulating, one has
	\begin{equation}
	\label{eq:UnifExpDec}
		\frac{\subProb_n(I_k> l, I_{k+1}^{n+1}=r_{k+1}^{n+1})}{\subProb_n(I_{k+1}^{n+1}=r_{k+1}^{n+1})}
		\leq
		\begin{cases}
			e^{-cl}			& \text{ if } l\geq L_0,\\
			1-\frac{\epsilon_0}{K}	& \text{ if } l<L_0.
		\end{cases}
	\end{equation}
	For $0< i\leq j\leq n+1$, let $X_{[i,j]}=\max\setof{i-m+I_m}{i\leq m\leq j}$ be ``the distance reached by $[i,j]$''; note that it is nonnegative. Doing (almost) the same computation as in~\eqref{eq:UnifExpDecComput}, one obtains the following
	\begin{claim}
		\label{claim:UnifExpDec}
		There exist $c>0, L_1\geq 0$ such that,
		\[
			\frac{\subProb_n(X_{[i,j]}> l, I_{j+1}^{n+1}=r_{j+1}^{n+1})}{\subProb_n(I_{j+1}^{n+1}=r_{j+1}^{n+1})}
			\leq
			\begin{cases}
				e^{-cl}					& \text{ if } l\geq L_1,\\
				1 - \bigl(\frac{\epsilon_0}{K}\bigr)^{L_1}	& \text{ if } l<L_1,
			\end{cases}
		\]
		uniformly in $i$, $j$ and $r_{j+1}^{n+1}$. In particular, there exist $C\geq 0, c>0$ such that:
		\[
		\frac{\subProb_n(X_{[i,j]}> l, I_{j+1}^{n+1}=r_{j+1}^{n+1})}{\subProb_n(I_{j+1}^{n+1}=r_{j+1}^{n+1})}
		\leq
		Ce^{-cl}.
		\]
	\end{claim}
	We will also need a uniform cut estimate.
	\begin{claim}
		\label{claim:UnifCut}
		There exists \(\epsilon>0\) such that
		\[
			\frac{\subProb_n(X_{[i,j]}= 0, I_{j+1}^{n+1}=r_{j+1}^{n+1})}{\subProb_n(I_{j+1}^{n+1}=r_{j+1}^{n+1})} \geq \epsilon ,
		\]
		uniformly in $i$, $j$ and $r_{j+1}^{n+1}$.
	\end{claim}
	\begin{proof}
		Proceeding as in~\eqref{eq:UnifExpDecComput},
		\begin{align*}
			\subProb_n(X_{[i,j]}= 0, I_{j+1}^{n+1}=r_{j+1}^{n+1})
			&=
			\subProb_n(I_{i+k}\leq k, 0\leq k\leq j-i, I_{j+1}^{n+1}=r_{j+1}^{n+1})\\
			&\geq
			\Bigl(\frac{\epsilon_0}{K}\Bigr)^{L_0}\prod_{k=L_0}^{j-i}(1-e^{-ck})\subProb_n(I_{j+1}^{n+1}=r_{j+1}^{n+1})\\
			&\geq
			\subProb_n(I_{j+1}^{n+1}=r_{j+1}^{n+1})\Bigl(\frac{\epsilon_0}{K}\Bigr)^{L_0}\prod_{k=1}^{\infty}(1-e^{-ck})\\
			&=
			\epsilon \subProb_n(I_{j+1}^{n+1}=r_{j+1}^{n+1}),
		\end{align*}
		since the infinite product converges.
	\end{proof}
	We now use Claims~\ref{claim:UnifExpDec} and~\ref{claim:UnifCut} to implement an exploration argument which will imply that having no cuts in a long interval carries an exponentially small measure. This in turn implies exponential decay of $\rho_L$, $\rho_R$ and $\statMes$ (Item~\ref{item:expDecOfSeqWeights}). Fix $l$, $n$ large enough and $m\in[0,n+1]\cap\Z$ such that $[m-l,m]\subset[0,n+1]$.
	Define
	\[
		D_{m}(l) = \{[m-l,m]\cap\cut = \emptyset\} .
	\]
	We want to prove the following
	\begin{claim}\label{claim:expDecNoCut}
		There exist $L_2\geq 0$ and $c>0$ such that, for all $l\geq L_2$ and $n \geq m\geq l$,
		\[
			\subProb_n(D_{m}(l)) \leq e^{-cl} .
		\]
	\end{claim}
	\begin{proof}
		The idea is the following: look at the furthest point reached by $[m,n+1]$; call it $i_1$. With measure at least $\epsilon$, $e_{i_1-1}$ is a cut. If not, look at the furthest point reached by $[i_1,n+1]$, and so on and so forth. To make this precise, we introduce
		\begin{gather*}
			Y_1=X_{[m,n+1]}, Y_2=X_{[m-Y_1,m-1]},Y_3=X_{[m-Y_1-Y_2,m-Y_1-1]},\dots,\\
			S_k = \sum_{i=1}^{k} Y_i,\ \text{ so that }\ Y_{k}=X_{[m-S_{k-1},m-S_{k-2}-1]},\\
			T = \min\setof{k}{S_k\geq l}.
		\end{gather*}
		All these quantities are functions of the memory configuration $I$. Now, for any $1>\delta>0$,
		\begin{align*}
			\subProb_n(D_{m}(l))
			&=
			\subProb_n(D_{m}(l),T\geq \delta l) + \subProb_n(D_{m}(l),T< \delta l)\\
			&\leq
			(1-\epsilon)^{\delta l} + \subProb_n(T< \delta l),
		\end{align*}
		via the uniformity in Claim~\ref{claim:UnifCut}. Finally, for $t> 0$,
		\begin{align*}
			\subProb_n\bigl( e^{tS_{\delta l}} \bigr) &=
			\sum_{k_1,...,k_{\delta l}=0}^{\infty} e^{t\sum_{i=1}^{\delta l}k_i}\subProb_n(Y_1=k_1,...,Y_{\delta l}=k_{\delta l})\\
			&\leq \sum_{k_1,...,k_{\delta l}=0}^{\infty} e^{t\sum_{i=1}^{\delta l}k_i}\prod_{i=1}^{\delta l}Ce^{-c k_i}\\
			&= \left(\frac{C}{1-e^{-c/2}}\right)^{\delta l},
		\end{align*}
		for any $t\leq c/2$, where we used the uniform exponential decay property of the $Y_i's$ (Claim~\ref{claim:UnifExpDec}) in the inequality. This gives:
		\[
			\subProb_n(T<\delta l) = \subProb_n(S_{\delta l}>l) \leq e^{-cl}
		\]
		for some $c>0$ and $\delta$ small enough, via the application of the exponential version of Markov inequality.
	\end{proof}
	With~\eqref{eq:factorisation}, Claim~\ref{claim:expDecNoCut} implies exponential decay of $\rho_L$, $\rho_R$ and $\statMes$ (item~\ref{item:expDecOfSeqWeights}), as well as the bound $\subProb_{n}(A)\leq e^{-cn}$ (where $A$ is defined just above~\eqref{eq:factorisation}).

	\bigskip	
	To conclude the proof of Lemma~\ref{lem:decoupling}, we must still establish Item~\ref{item:totalVarEstimate} and show that $\statMes$ is indeed a probability measure. We start with the latter. As $\statMes$ is a positive measure, we have to prove that $\sum_{\bfx\in\calS^*} \statMes(\bfx)=1$. This will be done using a standard renewal argument. We will need the weights
	\begin{gather*}
	\rho_n^L = \sum_{b_L\in\bc_L,\bfx\in\alp^{n-1}} \rho_L(b_L\sqcup\bfx),\ \rho_n^R = \sum_{b_R\in\bc_R,\bfx\in\alp^{n-1}} \rho_L(\bfx\sqcup b_R),\ \nu_n = \sum_{\bfx\in\alp^n} \statMes(\bfx),
	\end{gather*}
	and the associated generating functions (recall \(\mu_n=\sum_{b_L,b_R}\sum_{\bfx\in\alp^n}\subProb_n(b_L\sqcup\bfx\sqcup b_R)\) from~\ref{hyp:subExpDecOfMass})
	\begin{gather*}
	\bbA(z) = \sum_{n=1}^{\infty}\mu_n z^n, \quad
	\bbB(z) = \sum_{n=1}^{\infty}\nu_n z^n, \quad
	\bbC_L(z) = \sum_{n=0}^{\infty} \rho_n^L z^n, \quad
	\bbC_R(z) = \sum_{n=0}^{\infty} \rho_n^R z^n.
	\end{gather*}
	Since 
	\begin{align*}
		\sum_{\bfx\in\calS^*} \statMes(\bfx) = \sum_{n=1}^{\infty}\sum_{\bfx\in\alp^n} \statMes(\bfx) = \sum_{n=1}^{\infty}\nu_n = \bbB(1),
	\end{align*}
	we only need to show that $\bbB(1)=1$. We will deduce this from a functional equation satisfied by the previously introduced generating functions:
	\begin{align}
		\bbA(z)
		&=
		\sum_{n=1}^{\infty} z^{n} \sum_{b_L,b_R} \sum_{\bfx\in\alp^n} \subProb_n(b_L\sqcup\bfx\sqcup b_R) \notag\\
		&=
		\sum_{n=1}^{\infty} z^{n} \subProb_n(A^c) + \sum_{n=1}^{\infty} z^{n} \sum_{b_L,b_R} \sum_{\bfx\in\alp^n} \subProb_n(b_L\sqcup\bfx\sqcup b_R,A) \notag\\
		&=
		g_{A^c}(z) + \sum_{k=1}^{\infty} \sum_{n=k}^{\infty} \sum_{\substack{l,r,l_1,\dots,l_k\geq 1\\ l+r+\sum l_i=n+2}} \rho_{l-1}^L z^{l-1} \rho_{r-1}^R z^{r-1} \prod_{i=1}^{k} \nu_{l_i}z^{l_i} \notag\\
		&=
		g_{A^c}(z) + \sum_{k=1}^{\infty} \bbC_L(z)\bbB(z)^k\bbC_R(z) \notag\\
		&=
		g_{A^c}(z) - \bbC_L(z)\bbC_R(z) + \frac{\bbC_L(z)\bbC_R(z)}{1-\bbB(z)}.
		\label{eq:reneq}
	\end{align}
	Now, denoting $r_{g_{A_c}}$, $r_{\bbA}$, $r_{\bbB}$, $r_{\bbC_{L}}$ and $r_{\bbC_{R}}$ the radii of convergence of, respectively, $g_{A_c}$, $\bbA$, $\bbB$, $\bbC_L$ and $\bbC_R$, \ref{item:expDecOfSeqWeights} and the exponential decay of $\subProb_n(A^c)$ imply that
	\begin{gather*}
		r_{g_{A_c}} > 1, \qquad
		r_{\bbB}>1, \qquad
		r_{\bbC_{L}} > 1,
		\qquad r_{\bbC_{R}} > 1.
	\end{gather*}
	Furthermore, Properties~\ref{hyp:subExpDecOfMass} implies that
	\[
		r_{\bbA} = 1.
	\]
	Together with~\eqref{eq:reneq}, this yields $\bbB(r_{\bbA})=1$ and thus $\bbB(1)=1$.
	
	\smallskip
	Finally, we prove Item~\ref{item:totalVarEstimate}. For \(f:\bc_L\times\alp^n\times\bc_R\to \R\) bounded,
	\begin{align*}
		\big|\sum f(\bfy)\subProb_n(\bfy)-\sum f(\bfy)&\decouplLaw_n(\bfy)\big|\\
		&\leq \big|\sum f(\bfy)\subProb_n(\bfy,A)-\sum f(\bfy)\decouplLaw_n(\bfy)\big| + \norm{f}_{\infty}\subProb_n(A^c)\\
		&\leq  \norm{f}_{\infty}Ce^{-cn}
	\end{align*}
	by exponential decay of $\subProb_n(A^c)$ (implied by Claim~\ref{claim:expDecNoCut}) and equation~\eqref{eq:factorisation}.
\end{proof}

Since all notations and estimates are provided here, we prove a few technical points which are not directly useful in this paper but which might be of use in later investigations.

\begin{lemma}
	\label{lem:finEnergBC}
	Under the assumption of Lemma~\ref{lem:decoupling}, there exists $\epsilon>0$ such that
	\[
		\frac{\subProb_n(b_L\sqcup\bfx_1^k,(k,k+1)\in \cut)}{\subProb_n(b_L\sqcup\bfx_1^k)} \geq \epsilon,
	\]
	uniformly in $b_L$, $k\geq 0$ and $\bfx_1^k$.
\end{lemma}
\begin{proof}
	This follows form the same computation as in the argument leading to Claim~\ref{claim:UnifCut} in the proof of Lemma~\ref{lem:decoupling}.
\end{proof}

\begin{lemma}
	\label{lem:finEnergProc}
	If there exist $\delta>0$ and an element $s_0\in\alp$ such that 
	\[
		\subProb(s_0\given b)\geq \delta,
	\]
	for all $b\in\compl{\bc}_L$, then
	\[
		\statMes\bigl(\bfX=(s_0)\bigr)\geq \delta.
	\]
\end{lemma}
\begin{proof}
	By definition of $\statMes$, $\statMes((s_0)) = a_0(s_0) = \inf_{b\in\compl{\bc}_L}\subProb(s_0\given b)\geq \delta$.
\end{proof}

Before ending this sub-section, we observe two facts about the boundary pieces:
\begin{remark}
	\label{rem:finEnergBC}
	\begin{itemize}
		\item The same argument as in Lemma~\ref{lem:finEnergProc} gives the same result for $\statMes$ replaced by $\rho_R$ and $b_R\in\bc_R$ instead of $s_0\in\alp$.
		\item If $\subProb(b_L)\geq \delta$ for some $b_L$, then $\rho_L(b_L)\geq \delta$ (using the definition of $\rho_L$).
	\end{itemize}
\end{remark}

\subsection{Application to Random Walks with Exponentially Decaying Memory}
\label{sec:RenMem}
In order to avoid confusion, we continue to use $|\cdot|$ for the length of a sequence and use $\norm{\cdot}$ for the norm in $\R^d$.
We now apply results of the previous section to the setting where $\bc_L,\calS,\bc_R$ come equipped with a \emph{displacement application}: i.e. a function $\displace:\bc_L\cup\bc_R\cup\calS\to\Z^d$. This naturally induces an application $\trajApp$ from $\bc_L\times\calS^n\times\bc_R$ to the space of trajectories of $n+2$-steps random walk in $\Z^d$: denoting $\displace(\bfx_1^{n})=\sum_{i=1}^{n}\displace(x_i)$ (and similarly for $b_L\sqcup\bfx\sqcup b_R$, etc.),
\[
\trajApp_{u}(b_L\sqcup\bfx_1^n\sqcup b_R) = \bigl(u,u+\displace(b_L),u+\displace(b_L\sqcup\bfx_1^1),\dots,u+\displace(b_L\sqcup\bfx_1^n\sqcup b_R)\bigr),
\]
where $u\in\Z^d$ is the starting point of the trajectory. We continue to denote $\subProb_n$ the push-forward of $\subProb_n$ by $\trajApp$. It will be convenient to denote 
\[
\tilde{S} = (\tilde{S}_m)_{m=0}^{n+1} \equiv\tilde{S}(u,b_L\sqcup\bfx_1^n\sqcup b_R)=\trajApp_{u}(b_L\sqcup\bfx_1^n\sqcup b_R).
\]
In turn, $\trajApp$ induces an application $\trajAppFact$ from $(\bc_L\times\calS^n\times\bc_R,\clusterSet_n)$ to the trajectories of random walk with $\leq n+2$ steps via (denote $\bfy_1^{n+2}=b_L\sqcup\bfx_1^n\sqcup b_R$)
\[
\trajAppFact_{u}(\bfy,(l_0,\dots,l_k))=(u,u+\displace(\bfy_1^{l_0}),u+\displace(\bfy_1^{l_0+l_1}),\dots,u+\displace(\bfy_{1}^{n+2}) ).
\]
Again, denote
\[
\bar{S} = (\bar{S}_m)_{m=0}^{k+1} \equiv \bar{S}(u,b_L\sqcup\bfx_1^n\sqcup b_R,(l_0,\dots,l_k))=\trajAppFact_{u}(b_L\sqcup\bfx_1^n\sqcup b_R,(l_0,\dots,l_k)).
\]
The goal of this section is to give properties of the push-forward measure of $\subProb_n$ by $\trajAppFact$, denoted $\traj{\subProb_n}$, under hypotheses~\ref{hyp:unifSumability},~\ref{hyp:expMix},~\ref{hyp:subExpDecOfMass} and~\ref{hyp:azero} and some additional properties, namely:

\begin{enumerate}[label=(P\arabic*)]
	\item \label{prop:expDecSteps} there exist $C,c>0$ such that
	\[
		\subProb(\norm{\displace(x)}=l\given b) \leq Ce^{-cl}
	\]
	uniformly in $b$ (in particular, there exists $\Cl{expDecSteps}>0$ such that $\subProb(\norm{\displace(x)}=l\given b)\leq e^{-\Cr{expDecSteps}l}$ for $l$ large enough);
	\item \label{prop:directed} directedness: $\displace(x) \cdot \eone >0$ for all $x\in\bc_L\cup\bc_R\cup\calS$. In this case, it makes sense to distinguish the displacement along the first coordinate from the others and to denote
	\[
	\tilde{S}_k = (\tilde{S}_k^{\parallel},\tilde{S}_k^{\perp}) \in \Z\times\Z^{d-1},
	\] and similarly for $\bar{S}$;
	\item \label{prop:aperiod} aperiodicity: there exists $s_0\in\calS$ with $\displace(s_0)=\eone$ such that $\subProb(s\given b)\geq \epsilon_1>0$ uniformly over $b\in\compl{\bc}$;
	\item \label{prop:irred} irreducibility: there exist $r>0$ and $s_i\in\calS$, $i=1,\dots,d-1$, with $\displace(s_i)=(r,q_i)$ such that $\subProb(s_i\given b)\geq \epsilon_2>0$ uniformly over $b\in\compl{\bc}$, where $(q_i)_j=\IF{i=j}$, $j=1,\dots,d-1$;
	\item \label{prop:trajSym} trajectory symmetry (under~\ref{prop:directed}, starting point $u=0$):
	\begin{multline*}
	\subProb_n\Bigl(\tilde{S}=\bigl(0_{d},(v_L^{\parallel},v_L^{\perp}),(v_1^{\parallel},v_1^{\perp}),\dots,(v_R^{\parallel},v_R^{\perp})\bigr)\Bigr) \\
	=
	\subProb_n\Bigl(\tilde{S}=\bigl(0_{d},(v_L^{\parallel},-v_L^{\perp}),(v_1^{\parallel},-v_1^{\perp}),\dots,(v_R^{\parallel},-v_R^{\perp})\bigr)\Bigr).
	\end{multline*}
\end{enumerate}

Recall that \(\calS^*=\bigcup_{k=1}^\infty \alp^n\) and see the beginning of Section~\ref{sec:DecouplingRandomSequences} for the definitions of
\(\p\), \(\bfX\), \(\mathcal{R}_L\), \(\bfB_{L/R}\), \(\decouplLaw\) and \(\decouplLaw_n\).
Denote the push-forward of $\decouplLaw$ (resp. $\decouplLaw_n$) by $\traj{\decouplLaw}$ (resp.\ $\traj{\decouplLaw_n}$). Notice that, by construction, $\traj{\decouplLaw_n}$ and $\traj{\subProb_n}$ are both measures on \(\mathrm{Traj}_n=\bigcup_{k=1}^{n}(\Z^d)^{2+k}\). Recall $\rho_L$, $\rho_R$ and $\statMes$ defined in the previous section. We use the same notations for their push-forward. The goal of this section is to prove the following
\begin{theorem}
	\label{thm:procToRW}
	If $\subProb$ satisfies hypotheses~\ref{hyp:unifSumability},~\ref{hyp:expMix},~\ref{hyp:subExpDecOfMass} and~\ref{hyp:azero}, and Property~\ref{prop:expDecSteps}, then \(\statMes\) is a probability measure on $\Z^d$ and:
	\begin{enumerate}
		\item \label{it:TVThmprocToRW} there exist $c>0$, $C\geq 0$ such that, for all \(n\) and any bounded \(f:\mathrm{Traj}_n\to\R\),
		\[
			\big|\sum_{v\in \mathrm{Traj}_n} f(v)\traj{\subProb_n}(v)-\sum_{v\in \mathrm{Traj}_n} f(v)\traj{\decouplLaw_n}(v) \big| \leq \norm{f}_{\infty}Ce^{-cn}.
		\]
		\item Let $\bfX$ be a random variable with law $\statMes$ and write $X=\displace(\bfX)\in\Z^{d}$ (and define similarly $B_L,B_R$ from $\bfB_L,\bfB_R$). There exist $c>0$, $C\geq 0$ such that
		\[
			\statMes(\|X\|=l) \leq Ce^{-cl }
			\quad\text{ and }\quad
		    \rho_{L/R}(\|B_{L/R}\|=l) \leq Ce^{-cl }.
		\]
	\end{enumerate}
	Given $S_0\in\Z^d$ and an i.i.d.\ sequence $X_1=\displace(\bfX_1),X_2=\displace(\bfX_2),\dots$ (with $\bfX_i\sim\statMes$), denote
	\[
	S_k=S_0+\sum_{i=1}^{k}X_i 
	\]
	and write \(S^\parallel = (S^\parallel_k)_{k\geq 0}\) and \(S^\perp = (S^\perp_k)_{k\geq 0}\).
	Under~\ref{prop:directed} we have:
	\begin{itemize}
		\item if~\ref{prop:aperiod} is satisfied, $S^{\perp}$ and $S^{\parallel}$ are aperiodic,
		\item if~\ref{prop:aperiod} and~\ref{prop:irred} are satisfied, \(S^{\perp}\) is also irreducible and $S^{\parallel}$ can attain every $k>S_0^{\parallel}$ with positive probability.
		\item If~\ref{prop:trajSym} is satisfied, $\statMes(X_1^{\perp}=u)=\statMes(X_1^{\perp}=-u)$ and
		\[
		\p_{(0,u)}\bigl(\exists n: S_n=(L,v) \bigr) 
		=
		\p_{(0,v)}\bigl(\exists n: S_n=(L,u) \bigr),
		\]
		where \(\p_{(k,u)} \) denotes the law of $S$ for $S_0=(k,u)$.
	\end{itemize}
\end{theorem}

Before starting the proof, we make a small remark on the boundary conditions:
\begin{remark}\label{rem:BCtrivial}
	By Theorem~\ref{thm:procToRW}, for $n$ large enough, $\subProb_n(\,\cdot\,,b_L=(0,0),b_R=(0,0))\geq \frac{1}{2}\decouplLaw(\,\cdot\,,\bfB_L=(0,0),\bfB_R=(0,0))$. If $\subProb((0,0))\geq\delta_1>0$ and $a_0((0,0))\geq \delta_2>0$, then one can apply Remark~\ref{rem:finEnergBC} to obtain:
	\[\subProb_n(\cdot,b_L=(0,0),b_R=(0,0))\geq \frac{\delta_1\delta_2}{2} \p_{(0,0)}\bigl(\cdot,\exists k: S_k=(n,0) \bigr). \]
\end{remark}
\begin{proof}
	Using hypotheses~\ref{hyp:unifSumability},~\ref{hyp:expMix},~\ref{hyp:subExpDecOfMass} and~\ref{hyp:azero}, we can deduce from Lemma~\ref{lem:decoupling} that
	\begin{itemize}
		\item $\rho_L$, $\rho_R$ are positive measures on $\calS^*$ with finite total mass and satisfying
		\[
			\rho_{L/R}(|\bfB_{L/R}|=l) \leq Ce^{-cl} ;
		\]
		\item $\statMes$ is a probability measure on $\calS^*$ satisfying
		\[
			\statMes(|\bfX|=l) \leq Ce^{-\Cl{expDecSeq}l} ;
		\]
		\item Item~\ref{it:TVThmprocToRW} holds.
	\end{itemize}
	We thus only need to deduce exponential decay in $\norm{\cdot}$ from exponential decay in $\left|\cdot\right|$. Let us denote expectation under $\statMes$ by $\ebf$.
	Exponential decay in $\norm{\cdot}$ follows from
	\begin{claim}
	\label{claim:expDecRWstep}
		There exists $t_0>0$ such that
		\[
			\ebf\Bigl[e^{t \norm{X}}\Bigr]<\infty,
		\]
		for all $t\leq t_0$.
	\end{claim}
\begin{proof}
	We can assume, without loss of generality, that \(t\geq 0\). In that case,
	\begin{align}
	\ebf\bigl[e^{t \norm{X}}\bigr]
	&\leq
	\sum_{n=1}^{\infty} \ebf\bigl[e^{t \sum_{i=1}^{n}\norm{\displace(\bfX_i)}}\IF{|\bfX|=n}\bigr]\notag \\
	&\leq
	\sum_{n=1}^{\infty} \sqrt{\ebf\bigl[e^{2t\sum_{i=1}^{n}\norm{X_i}}\bigr] \ebf\bigl[\IF{|\bfX|=n}\bigr]},
	\label{eq:CS}
	\end{align}
	by the Cauchy--Schwartz inequality. Then,
	\begin{align*}
	\ebf\bigl[e^{2t\sum_{i=1}^{n}\norm{X_i}}\bigr]
	&\leq
	\sum_{k=0}^{n} \sum_{\substack{A\subset\{1,...,n\}\\A=\{a_1,\ldots,a_k\}}} e^{2tm_0(n-k)} \ebf\bigl[e^{2t\sum_{i=1}^{k}\|X_{a_i}\|}\IF{\|X_{a_i}\|\geq m_0,i=1,\ldots,k}\bigr]\\
	&\leq
	\sum_{k=0}^{n} \sum_{\substack{A\subset\{1,...,n\}\\A=\{a_1,\ldots,a_k\}}} e^{2tm_0(n-k)} \sum_{\ell_1,\ldots,\ell_k=m_0}^{\infty} e^{2t\sum_{i=1}^{k}\ell_i} e^{-\Cr{expDecSteps} \sum_{i=1}^{k}\ell_i}\\
	&=
	\sum_{k=0}^{n} \binom{n}{k} e^{2t m_0 (n-k)} \Bigl(\sum_{\ell=m_0}^{\infty} e^{-(\Cr{expDecSteps}-2t)\ell} \Bigr)^k\\
	&=
	\Bigl(e^{2t m_0} + \frac{e^{-(\Cr{expDecSteps}-2t)m_0}}{1-e^{-(\Cr{expDecSteps}-2t)}} \Bigr)^n,
	\end{align*}
	for large enough \(m_0\) (the first inequality holds for any $m_0>0$, but we need $m_0$ to be large enough to use~\ref{prop:expDecSteps} (uniform exponential decay of the steps) in the second inequality). Plugging this into~\eqref{eq:CS} yields
	\begin{align*}
	\ebf\bigl[e^{t \norm{X}}\bigr]
	&\leq
	C\sum_{n=1}^{\infty} \Bigl(e^{2t m_0} + \frac{e^{-(\Cr{expDecSteps}-2t)m_0}}{1-e^{-(\Cr{expDecSteps}-2t)}}\Bigr)^{n/2} e^{-\Cr{expDecSeq} n/2}
	<
	\infty ,
	\end{align*}
	provided that $\bigl(e^{2t m_0} + \frac{e^{-(\Cr{expDecSteps}-2t)m_0}}{1-e^{-(\Cr{expDecSteps}-2t)}} \bigr) e^{-\Cr{expDecSeq}} < 1$, which is true for $t\in [0,t_0)$ for some $t_0=t_0(m_0)>0$, once $m_0$ is chosen large enough.
\end{proof}
The previous argument extends easily to obtain exponential decay in $\norm{\cdot}$ under $\rho_{L/R}$.

\bigskip
We now turn to the additional properties. The aperiodicity of $S^{\perp}$ follows immediately from~\ref{prop:aperiod}, since the latter gives $\statMes(X^{\perp}=0)\geq \epsilon_1>0$. The aperiodicity of $S^{\parallel}$ is done identically and so is the irreducibility of $S^{\perp}$ under~\ref{prop:irred}.

The symmetry is slightly less obvious. Start by observing that, by definition of $a_k(s|b_L\sqcup\bfx_1^n)$ and $\Delta_k(s|b_L\sqcup\bfx_1^n)$, \ref{prop:trajSym} implies that
\begin{align*}
a_k((v^{\parallel},v^{\perp}) \given b_L\sqcup\bfx_1^n)
&=
\sum_{\substack{s\in\alp:\\\displace(s)=(v^{\parallel},v^{\perp})}} a_k(s \given b_L\sqcup\bfx_1^n)\\
&=
\sum_{\substack{s\in\alp:\\\displace(s)=(v^{\parallel},-v^{\perp})}} a_k(s \given b_L\sqcup\bfx_1^n)
=
a_k((v^{\parallel},-v^{\perp}) \given b_L\sqcup\bfx_1^n),
\end{align*}
and, thus,
\[
\Delta_k((v^{\parallel},v^{\perp})|b_L\sqcup\bfx_1^n) = \Delta_k((v^{\parallel},-v^{\perp})|b_L\sqcup\bfx_1^n).
\]
Using this in the expansion of $\statMes\bigl(\displace(\bfX)=(v^{\parallel},v^{\perp})\bigr)$ as 
\[
\sum_{k}\sum_{v_1+\dots+v_k=v}\sum_{x_i:\displace(x_i)=v_i}\statMes(\bfX=(x_1,\dots,x_k)),
\] and using the definition of $\statMes$, one straightforwardly obtains
\[
	\statMes\bigl(\displace(\bfX)=(v^{\parallel},v^{\perp})\bigr) = \statMes\bigl(\displace(\bfX)=(v^{\parallel},-v^{\perp})\bigr) .
\]

Finally, for $L>0$ and $u,v\in\Z^{d-1}$,
\[
\p_{(0,u)}\bigl(\exists n: S_n=(L,v) \bigr) = \p_{(0,v)}\bigl(\exists n: S_n=(L,u) \bigr)
\] follows by summing over possible trajectories and applying the previous symmetry result.
\end{proof}


\section{Acknowledgments}
The authors thank the two referees for their careful reading and their questions and comments that substantially improved the readability of this work. They also acknowledge the support of the Swiss National Science Foundation through the NCCR SwissMAP.


\bibliographystyle{plain}
\bibliography{OV17}

\begin{thebibliography}{10}

\bibitem{Abraham-1980}
D.~B. Abraham.
\newblock Solvable model with a roughening transition for a planar {I}sing
  ferromagnet.
\newblock {\em Phys. Rev. Lett.}, 44(18):1165--1168, 1980.

\bibitem{Abraham-1981}
D.~B. Abraham.
\newblock Binding of a domain wall in the planar {I}sing ferromagnet.
\newblock {\em J. Phys. A}, 14(9):L369--L372, 1981.

\bibitem{Burkhardt-1981}
T.~W. Burkhardt.
\newblock Localisation-delocalisation transition in a solid-on-solid model with
  a pinning potential.
\newblock {\em J. Phys. A}, 14(3):L63, 1981.

\bibitem{Campanino+Ioffe+Velenik-2003}
M.~Campanino, D.~Ioffe, and Y.~Velenik.
\newblock Ornstein-{Z}ernike theory for finite range {I}sing models above
  {$T_c$}.
\newblock {\em Probab. Theory Related Fields}, 125(3):305--349, 2003.

\bibitem{Campanino+Ioffe+Velenik-2008}
M.~Campanino, D.~Ioffe, and Y.~Velenik.
\newblock Fluctuation theory of connectivities for subcritical random cluster
  models.
\newblock {\em Ann. Probab.}, 36(4):1287--1321, 2008.

\bibitem{Chalker-1981}
J.~T. Chalker.
\newblock The pinning of a domain wall by weakened bonds in two dimensions.
\newblock {\em J. Phys. A}, 14(9):2431, 1981.

\bibitem{Chui+Weeks-1981}
S.~T. Chui and J.~D. Weeks.
\newblock Pinning and roughening of one-dimensional models of interfaces and
  steps.
\newblock {\em Phys. Rev. B}, 23:2438--2441, 1981.

\bibitem{Comets+Fernandez+Ferrari-2002}
F.~Comets, R.~Fern\'andez, and P.~A. Ferrari.
\newblock Processes with long memory: regenerative construction and perfect
  simulation.
\newblock {\em Ann. Appl. Probab.}, 12(3):921--943, 2002.

\bibitem{Delfino-2016}
G.~Delfino.
\newblock Interface localization near criticality.
\newblock {\em Journal of High Energy Physics}, 2016(5):32, 2016.

\bibitem{Duminil-Copin+Manolescu-2016}
H.~Duminil-Copin and I.~Manolescu.
\newblock The phase transitions of the planar random-cluster and potts models
  with $q\geq 1$ are sharp.
\newblock {\em Probability Theory and Related Fields}, 164(3):865--892, 2016.

\bibitem{Duminil-Copin+Raoufi+Tassion-2017}
H.~Duminil-Copin, A.~Raoufi, and V.~Tassion.
\newblock Sharp phase transition for the random-cluster and {P}otts models via
  decision trees.
\newblock Preprint, arXiv:1705.03104, 2017.

\bibitem{Fisher-1984}
M.~E. Fisher.
\newblock Walks, walls, wetting, and melting.
\newblock {\em J. Statist. Phys.}, 34(5-6):667--729, 1984.

\bibitem{Friedli+Ioffe+Velenik-2013}
S.~Friedli, D.~Ioffe, and Y.~Velenik.
\newblock Subcritical percolation with a line of defects.
\newblock {\em Ann. Probab.}, 41(3B):2013--2046, 2013.

\bibitem{Giacomin-2007}
G.~Giacomin.
\newblock {\em Random polymer models}.
\newblock Imperial College Press, London, 2007.

\bibitem{Graham+Grimmett-2011}
B.~Graham and G.~Grimmett.
\newblock Sharp thresholds for the random-cluster and ising models.
\newblock {\em Ann. Appl. Probab.}, 21(1):240--265, 2011.

\bibitem{Grimmett-1999}
G.~Grimmett.
\newblock {\em Percolation}.
\newblock Springer, Berlin Heidelberg, 1999.

\bibitem{Grimmett-2006}
Geoffrey Grimmett.
\newblock {\em The random-cluster model}, volume 333 of {\em Grundlehren der
  Mathematischen Wissenschaften [Fundamental Principles of Mathematical
  Sciences]}.
\newblock Springer-Verlag, Berlin, 2006.

\bibitem{Ioffe-2015}
Dmitry Ioffe.
\newblock Multidimensional random polymers: a renewal approach.
\newblock In {\em Random walks, random fields, and disordered systems}, volume
  2144 of {\em Lecture Notes in Math.}, pages 147--210. Springer, Cham, 2015.

\bibitem{Kroll-1981}
D.~M. Kroll.
\newblock Solid-on-solid model for the interface pinning transition in {I}sing
  ferromagnets.
\newblock {\em Z. Phys. B}, 41(4):345--348, 1981.

\bibitem{McCoy_Perk-1980}
B.~M. McCoy and J.~H.~H. Perk.
\newblock Two-spin correlation functions of an {I}sing model with continuous
  exponents.
\newblock {\em Phys. Rev. Lett.}, 44:840--844, 1980.

\bibitem{Vallade+Lajzerowicz-1981}
M.~Vallade and J.~Lajzerowicz.
\newblock Transition rugueuse et localisation pour une singularité linéaire
  dans un espace à deux ou trois dimensions.
\newblock {\em J. Physique}, 42(11):1505--1514, 1981.

\bibitem{vanLeeuwen+Hilhorst-1981}
J.~M.~J. van Leeuwen and H.~J. Hilhorst.
\newblock Pinning of a rough interface by an external potential.
\newblock {\em Phys. A}, 107(2):319--329, 1981.

\bibitem{Velenik-2006}
Y.~Velenik.
\newblock Localization and delocalization of random interfaces.
\newblock {\em Probab. Surv.}, 3:112--169, 2006.

\end{thebibliography}

\end{document}